%% file: quantumarticle.tex
\documentclass[a4paper,onecolumn,11pt,accepted=2025-04-27]{quantumarticle}
\pdfoutput=1
\usepackage[utf8]{inputenc}
\usepackage[english]{babel}
\usepackage[T1]{fontenc}
\usepackage{amsmath}
\usepackage{amssymb}
\usepackage{physics}
\usepackage{thmtools, thm-restate}
\usepackage[dvipsnames,svgnames]{xcolor}
\usepackage{wrapfig}
\usepackage{comment}
\usepackage{hyperref}
\usepackage[capitalize,nameinlink]{cleveref}

\usepackage{bbold}
\usepackage[inline]{enumitem}
\usepackage{tikz}
\usepackage{lipsum}
\usepackage{algorithm}
\usepackage{algpseudocode}

\newtheorem{definition}{Definition}

\newtheorem{proposition}{Proposition}
\newtheorem{lemma}{Lemma}

\newenvironment{proof}{\paragraph{\textit{Proof}.}}{\hfill$\square$}

\newcommand{\pone}{\alpha} 
\newcommand{\ptwo}{\beta} 
\newcommand{\joint}{\Psi} 
\newcommand{\mom}{\Phi} 
\newcommand{\jopt}{\psi} 
\newcommand{\mopt}{\phi} 
\newcommand{\sone}{\mathcal{A}} 
\newcommand{\stwo}{\mathcal{B}} 
\newcommand{\sjoint}{\mathcal{C}} 
\newcommand{\step}{\eta} 

\newcommand{\orthproj}{\textrm{Orth}\Pi} 
\newcommand{\mirproj}{\textnormal{Mir}\Pi} 
\newcommand{\mirproxproj}{\textnormal{MirProx}\Pi} 
\newcommand{\proxproj}{\textnormal{Prox}\Pi} 
\newcommand{\regfunc}{h} 
\newcommand{\reg}[1]{\regfunc({#1})} 
\newcommand{\bregdiv}{\mathcal{BD}} 

\renewcommand{\grad}{\mathcal{F}} 
\newcommand{\outcomes}{\Omega} 
\newcommand{\out}{\omega} 
\newcommand{\pout}{P_\out} 
\newcommand{\Id}{\mathbb{1}} 
\newcommand{\util}{\mathcal{U}} 
\newcommand{\payob}{U} 
\newcommand{\exutil}{u} 
\newcommand{\bigo}{\mathcal{O}} 
\newcommand{\cfone}{\mathcal{CF}_\pone} 
\newcommand{\cftwo}{\mathcal{CF}_\ptwo} 
\newcommand{\cont}{\gamma} 
\newcommand{\dg}{\mathcal{DG}} 
\newcommand{\super}{\Xi} 

\renewcommand{\tr}{\textnormal{Tr}} 
\newcommand{\pauli}{\mathcal{P}} 

\newcommand{\inner}[2]{\tr \left[{#1}\: {#2}\right]}

\title{A Quadratic Speedup in Finding Nash Equilibria  of Quantum Zero-Sum Games}

\author[1]{Francisca Vasconcelos}
\orcid{}
\email{francisca@berkeley.edu}
\author[2]{Emmanouil-Vasileios Vlatakis-Gkaragkounis}
\orcid{}
\email{vlatakis@wisc.edu}
\author[3]{Panayotis Mertikopoulos}
\orcid{0000-0003-2026-9616}
\author[4]{Georgios Piliouras}
\orcid{}
\author[5]{Michael I. Jordan}
\orcid{}
\affil[1]{UC Berkeley}
\affil[2]{UC Berkeley, UW Madison}
\affil[3]{Univ. Grenoble Alpes, CNRS, Inria, Grenoble INP, LIG, 38000 Grenoble, France}
\affil[4]{Google DeepMind}
\affil[5]{UC Berkeley, Inria Paris}

\begin{document}

\maketitle
\begin{abstract}
    Recent developments in domains such as non-local games, quantum interactive proofs, and quantum generative adversarial networks have renewed interest in quantum game theory and, specifically, quantum zero-sum games. Central to classical game theory is the efficient algorithmic computation of Nash equilibria, which represent optimal strategies for both players. In 2008, Jain and Watrous proposed the first classical algorithm for computing equilibria in quantum zero-sum games using the Matrix Multiplicative Weight Updates (MMWU) method to achieve a convergence rate of $\bigo(d/\epsilon^2)$ iterations to $\epsilon$-Nash equilibria in the $4^d$-dimensional spectraplex. In this work, we propose a hierarchy of quantum optimization algorithms that generalize MMWU via an extra-gradient mechanism. Notably, within this proposed hierarchy, we introduce the Optimistic Matrix Multiplicative Weights Update (OMMWU) algorithm and establish its average-iterate convergence complexity as $\bigo(d/\epsilon)$ iterations to $\epsilon$-Nash equilibria. This quadratic speed-up relative to Jain and Watrous' original algorithm sets a new benchmark for computing $\epsilon$-Nash equilibria in quantum zero-sum games. 
\end{abstract}

\newpage
\tableofcontents

\newpage
\input{introduction}

\input{prelims}

\input{convergence}

\input{experiments}

\input{conclusion}

\input{acknowledgements}

\newpage
\bibliographystyle{quantum}
\bibliography{quantum}

\newpage
\appendix
\input{proofs_prelims}

\end{document}

%% file: introduction.tex
\section{Introduction}
\subsection{Motivation}
Ever since von Neumann's groundbreaking work in the 1920's \cite{v1928theorie}, game theory has become a foundational pillar of modern mathematics, economics, and computer science.
The central solution concept in non-cooperative game theory is that of a \emph{Nash equilibrium}, i.e., a strategy of joint play in which no single player can benefit from a unilateral deviation~\cite{nash1951non}.
Despite its mathematical significance, the computational perspective on general Nash equilibria is murkier, since even approximating such an equilibrium is a PPAD-complete problem for general classical games.
On the other hand,  Nash equilibria in two-player zero-sum games 
(i.e., where one player's gain is the other's loss) can be computed in polynomial time, 
and such games have found practical applications in a wide variety of domains, 
from resource allocation problems~\cite{pillai2014resource}, to political strategy~\cite{brams2011game}, and the training of machine learning models~\cite{silver2016mastering,GPAM+14,PDSG17}.

When the players' actions are the states of a quantum system, classical game theory no longer applies, and one must consider \emph{quantum games}. These games constitute a natural framework for the study of the exchange and processing of quantum information, resulting in various useful applications to theoretical, computational, and cryptographic settings. The original interest in quantum games dates back to seminal work by Bell \cite{bell64} and CSHS \cite{cshs69}, explaining the Einstein-Podolsky-Rosen (EPR) Paradox. Since then, substantial work on  quantum ``non-local games'' \cite{mermin90, hardy93, Aravind:01, greenberger2007going, reichardt_classical_2013} has lead to novel methods for testing whether two or more spatially separated, non-communicating
quantum devices can generate correlations not reproducible by any pair of classical devices---resulting in entanglement-based tests of quantum advantage. Much of this work culminated in establishing MIP$^*$=RE \cite{mipre}, which, among many notable results, demonstrated that there exists an efficient reduction from the Halting Problem to deciding whether a two-player non-local game has entangled value at most 1 or $\frac{1}{2}$. Furthermore,  it was recently proven that non-local games can be compiled into single-prover interactive games \cite{KalaiLV023}. In general, quantum interactive proofs with competing provers can be modeled as competitive quantum refereed games \cite{Gutoski05,GW05,Gutoski07} and multi-prover quantum interactive proofs \cite{Kobayashi03, Cleve04, kempe_using_2009, Kempe11} can be modeled as cooperative quantum games. Finally, there has been substantial work on quantum coin-flipping \cite{Ambainis01, Spekkens02, mochon2007quantum, Miller20}, a game model where two players directly exchange quantum information, studying how two parties with competing interests can carry out a fair coin flip across a quantum communication channel.

In this work, we consider the specific class of quantum \emph{zero-sum} games, in which two players are in direct competition and do not share entanglement. One of the earliest works in quantum game theory \cite{meyer1996} studied a ``matching pennies''-type zero-sum game to prove that quantum strategies are at least as good as and, in some cases, can outperform classical strategies. Further work on quantum zero-sum games led to a proof that the quantum complexity class QRG(1), of problems having one-turned quantum refereed games, is contained in PSPACE \cite{jain2009parallel}. Finally, in the context of quantum machine learning, there has been substantial recent interest in quantum generative adversarial networks (QGANs) \cite{KilloranQGANs,ChakrabartiHLFW19}, for which the training of the competing generator and discriminator networks can be modeled as a quantum zero-sum game.

As in the classical setting, finding Nash equilibria of \emph{general} quantum games is computationally prohibitive---PPAD-complete to be exact \cite{Bostanci2022quantumgametheory}. However, the situation is more favorable for quantum \emph{zero-sum} games. 
As demonstrated by Jain and Watrous \cite{jain2009parallel}, Nash equilibria of $d$-qubit two-player quantum zero-sum games can be calculated to $\epsilon$-accuracy in polynomial time and linear dependence on the number of qubits (logarithmic in the spectraplex dimensionality), i.e., $\bigo(d/\epsilon^2)$, even with payoff-based information on the players' side \cite{LMBB23}.
In view of this, equilibrium strategies in quantum zero-sum games can in principle be computed algorithmically, thus providing the required guarantees of implementability for proof schemes or cryptographic ciphers that rely on the computation of quantum Nash equilibria.
Our goal is to pursue this objective further,  bringing tools from online learning and computational learning theory to bear to improve the design of algorithms for the computation of approximate Nash equilibria in quantum zero-sum games.

\subsection{Prior Work} \label{sec:prior_work}

Our work builds upon the work of Jain and Watrous \cite{jain2009parallel} for non-interactive quantum zero-sum games.  To provide some context,  \cite{jain2009parallel} leveraged feedback in the form of quantum channels (governed by superoperators) to introduce the Matrix Multiplicative Weights Update (MMWU) algorithm, a two-player variant of the matrix exponentiated gradient algorithm of Tsuda et al. \cite{tsuda2005matrix}, which is itself a special case of Mirror Descent \cite{NY83,BECK2003167}. 
Jain, Piliouras and Sim \cite{jain2022matrix} recently revisited the setting of MMWU in non-interactive quantum zero-sum games by studying its day-to-day behavior. They showed that a continuous-time analogue of MMWU, that they call \emph{quantum replicator dynamics}, leads to cyclic behavior (formally Poincar\'{e} recurrence) in quantum zero-sum games with interior equilibrium as the quantum relative entropy between a fully mixed Nash equilibrium and the evolving state of the system is time-invariant.
Beyond the setting of quantum zero-sum games, variants of the MMWU algorithm have also been used in many important applications such as
proving QIP=PSPACE \cite{DBLP:journals/jacm/JainJUW11},
solving SDPs \cite{arora2007combinatorial},
finding balanced separators \cite{chen2023simple},
enhancing spectral sparsification \cite{allen2015spectral},
covariance optimization tasks \cite{MBM12,MM16,BMB20}, and
matrix learning \cite{KSST12}. We refer the interested reader to the survey by Arora, Hazan, and Kale \cite[Section 5]{arora2012multiplicative}, as well as the numerous follow-up works referring to it, for a more detailed treatment of the history and applications of MMWU.

For the quantum zero-sum games of interest in this work, despite recent work on no-regret learning dynamics for zero-sum and general quantum games \cite{LMBB23,lin2023noregret}, to the best of our knowledge, the $\bigo(d/\epsilon^2)$ convergence rate of the MMWU algorithm is the state of the art for finding $\epsilon$-approximate Nash equilibria.
In this work, we ask:
\begin{center}
    \textbf{\textsf{Can we improve on the performance of MMWU in quantum zero-sum games?}}
\end{center}
More precisely, our paper's goal and main contribution is to provide an accelerated version of the MMWU algorithm which achieves $\epsilon$-accuracy in linear, $O(d/\epsilon)$ time, instead of $O(d/\epsilon^2)$.

To achieve this quadratic speedup, our point of departure is the literature on classical \emph{finite} games. The classical version of the MMWU algorithm---known, among other names, as the Multiplicative Weights Update (MWU) method \cite{arora2012multiplicative}---similarly achieves an $\bigo_d(1/\epsilon^2)$ iteration complexity.%
\footnote{Hereon, we will use the $\bigo_d(\cdot)$ notation to hide the dimensionality dependence of the iteration complexity. There will be explicit discussion of dimensionality dependence in the quantum setting in \Cref{sec:dim_dep}.} 
However, Nemirovski \cite{Nemirovski2004Prox} and Auslender \& Teboulle \cite{AT05} showed that, by a suitable modification of the MWU method---now commonly known as the \emph{mirror-prox} family of algorithms---it is possible to accelerate this rate to $\bigo_d(1/\epsilon)$, a rate which was recently shown by Ouyang and Xu \cite{OX21} to be order-optimal in the setting of min-max convex-concave problems.

The main algorithmic insight of the mirror-prox template is the combination of an extra-gradient step in the spirit of Korpelevich \cite{korpelevich1976extragradient} with an iterate-averaging mechanism à la Polyak \& Juditsky \cite{PJ92}.
Coupled with the smoothness of the underlying objective, this combination yields a $\bigo(1/\epsilon)$ convergence rate, albeit with a possibly suboptimal dependence on the dimension of the problem.
In the case of the simplex---which is the state space of classical, finite games---the dependence on the dimension can be improved dramatically by instantiating the mirror-prox with an entropic regularizer, which essentially boils down to a variant of MWU ``with advice'' \cite{Nemirovski2004Prox}.
The only notable drawback of the resulting algorithm is that it requires not one, but two oracle queries per iteration, which doubles the computation cost, and requires additional coordination\,/\,communication from the players' side.
Building on an original idea by Popov \cite{popov1980modification}, this ``cost doubling'' issue was mitigated by the so-called \emph{optimistic mirror descent} (OMD) proposal of Raklhin and Sridharan \cite{RS13-NIPS}, which optimistically reused past gradient information as a surrogate for the extra-gradient step of the mirror-prox algorithm.
In this way, when applied to mixed extensions of classical, finite games, OMD ultimately achieves an order-optimal convergence rate in terms of both $\epsilon$ and $d$, all the while using a single oracle call per iteration.
We summarize this hierarchy of optimization methods and the relevant trade-offs in \Cref{fig:classical_methods}.

\subsection{Our Contributions and Methodology}
Our work focuses on developing an accelerated algorithm with similar, linear-time performance guarantees in the \emph{quantum} setting.
The main contributions of this work are thus twofold:
\textit{a}) we propose the  \emph{Optimistic Matrix Multiplicative Weights Update (OMMWU) algorithm} (\Cref{alg:OMMWU}) for finding $\epsilon$-Nash equilibria in quantum zero-sum games;
and
\textit{b})
we show that OMMWU achieves the following iteration complexity guarantees.

\begin{restatable}[OMMWU Iteration Complexity]{corollary}{colommwu}\label{thm:main-result-entroy}%
    In a $4^d$-dimensional spectraplex, OMMWU  computes average-iterate $\epsilon$-Nash equilibria in $\bigo(d/{\epsilon})$ iterations.
\end{restatable}

\noindent In order to achieve this result, we forge a link between quantum zero-sum games and established works in classical game theory and optimization. In addition to OMMWU---which surpasses the
$\bigo(d/\epsilon^2)$ rate of MMWU with a single gradient call per iteration---we also propose a hierarchical family of algorithms and techniques that can be applied to various semidefinite programming problems (quantum or otherwise) and thus may be of independent interest.
\\
\\
\textsf{\textbf{Methodology.}} In terms of methodology, we move away from  the \emph{channel-based} view of quantum zero-sum games in \cite{jain2009parallel} and instead adopt a \emph{gradient-based} perspective, similar to prior work on learning in general quantum games~\cite{lotidis2023learning,LMB23-CDC} (see also~\cite{lin2023noregret} connecting the same online-learning, regret-minimization perspective to new classes of quantum correlated equilibria). In this framework, Alice and Bob's feedback is characterized not by traditional quantum superoperators, but rather by mathematical gradient operators,
which comprise the game's gradient feedback operator $\mathcal{F}$. We prove several important properties of $\mathcal{F}$, namely establishing Lipschitz continuity and monotonicity.
In doing so, we conclude that solving a quantum zero-sum game to a desired accuracy essentially boils down to solving a smooth, semidefinite convex-concave problem with first-order oracle information.

\begin{algorithm}[t] 
\caption{Optimistic Matrix Multiplicative Weights Update (OMMWU)}\label{alg:OMMWU}
\begin{algorithmic}
\State \textbf{Accuracy Parameter:} $\epsilon$
\State \textbf{Regularization Function:} $h:\{\mathcal{A},\mathcal{B}\}\to\mathbb{R}$
\State \textbf{Strong Convexity Parameter (of $h$):} $\mu_h$
\State \textbf{Diameter (of $\bregdiv_\regfunc$):} 
$\mathcal{D}_h=\sup_{X,Y\in \sjoint} \bregdiv_\regfunc(X\|Y)$
\State \textbf{Lipschitz Parameter (of $\grad$):} $\gamma_\grad$
\State \textbf{Proximal Map:} $
    \proxproj^{\regfunc,\step}_{\sone,\stwo}(X,Y):=
    {\arg\min}_{C\in\sjoint}\{ \langle Y,C-X\rangle -\frac{1}{\step}\bregdiv_\regfunc (C||X)\}.$
\vspace{0.1in}
\State $\step \gets  \ {\mu_h}/(2\gamma_\grad)$ 
\Comment{Step Size}
\State $N \gets \lceil \mathcal{D}_h/(\step\cdot\epsilon) \rceil$ 
\Comment{Number of Rounds}
\State $(\pone_0,\ptwo_0) \gets \left(\frac{1}{2^n}\Id_\sone, \frac{1}{2^m}\Id_\stwo\right)$ \Comment{State Initialization}
\State $(\hat{\pone}_0,\hat{\ptwo}_0) \gets (\pone_0,\ptwo_0)$ \Comment{Intermediate State Initialization}

\vspace{0.1in}
\For{$t\in[1,N-1]$}

\vspace{0.1in}

\State $\pone_{t+1}\gets\Lambda (\log \hat{\pone}_t +\eta \grad_\pone(\ptwo_t))$ \Comment{State Updates}
\vspace{0.05in}
\State $\ptwo_{t+1}\gets \Lambda (\log \hat{\ptwo}_t +\eta \grad_\ptwo(\pone_t))$

\vspace{0.1in}
\State $\hat{\pone}_{t+1}\gets
\Lambda (\log \hat{\pone}_t +\eta \grad_\pone(\ptwo_{t+1}))$ \Comment{Momentum Updates}
\vspace{0.05in}
\State $\hat{\ptwo}_{t+1}\gets\Lambda (\log \hat{\ptwo}_t +\eta \grad_\ptwo(\pone_{t+1}))$

\vspace{0.1in}
\EndFor
\end{algorithmic}
\vspace{0.1in}
\Return $\bigg( \bar{\pone}=\frac{1}{N}\sum_{t=0}^{N-1}\pone_t, \bar{\ptwo}=\frac{1}{N}\sum_{t=0}^{N-1}\ptwo_t \bigg)$
\end{algorithm} 

Our approach is based on endowing the problem's feasible region---a product of spectraplexes---with an OMD template in the spirit of Rakhlin and Sridharan \cite{RS13-NIPS}.
A difficulty that arises when trying to combine these elements is that the product structure of standard exponential weight algorithms does not carry over automatically to the non-commutative matrix variables that arise in the quantum setting. 
The key observation that allows us to retain the speed-up of a mirror prox method and the state-of-art dimensionality dependence is that the von Neumann entropy (which is the underlying regularizer of the MMWU algorithm) can be encoded in a matrix exponentiation step, as per the original MMWU algorithm (see also~\cite{canyakmaz2023multiplicative} for generalizations of this connection for optimization over symmetric cones).
This allows us to implement an optimistic update structure in the problem's dual space, where the primary arithmetic operation is ordinary addition---which is an abelian operation on the space of Hermitian matrices.
Thus, drawing inspiration from a classical hierarchy of optimization methods, we propose a \emph{design of quantum zero-sum game algorithms} (depicted in \Cref{fig:quantum_methods}) achieving an optimized convergence rate with a single gradient call per iteration.

Within this design, we demonstrate that  the  MMWU algorithm of \cite{jain2009parallel} instantiates (via a von Neumann entropy regularizer) the \emph{Matrix Dual Averaging (MDA)} method (\Cref{alg:mda}). Following a classical proof technique in the spirit of Ene and Nguyên \cite{ene2022adaptive}, we also prove that another of the proposed algorithms in our hierarchy---the \textit{Optimistic Matrix Mirror Prox (OMMP)} method (\Cref{alg:SCMMP})---yields a quadratic speedup relative to the convergence rate of MDA, while still requiring only one gradient call per iteration:
\begin{restatable}[Main Result]{theorem}{mainres}\label{thm:main-result}%
    The OMMP method computes $\epsilon$-Nash equilibria in $\bigo_d({1}/{\epsilon})$ steps.
\end{restatable}
\noindent Thus, OMMP establishes a new benchmark for computing $\epsilon$-Nash finite-valued quantum zero-sum games. It is by instantiating OMMP with a von Neumann entropy regularizer, that the OMMWU algorithm and $\bigo(d/{\epsilon})$ iteration complexity is achieved.
\\
\\
\textbf{\textsf{Proof Insights.}}
The core idea of an \emph{optimistic} algorithm is to leverage the gradient from the previous iteration as a forecast for the subsequent iteration's gradient. This is predicated on the presumption that if the game's adversary employs a consistent (or stabilized) strategy, then the discrepancy between the gradients of successive iterations will be negligible. Therefore, the formulation of our algorithm is crafted by devising a parallel to such methods, circumventing the limitations imposed by the non-commutativity of the matrix operator, which is a common hurdle in quantum problems. Upon establishing the algorithm in this manner, we demonstrate that the (quantum) superoperators, corresponding to the method's gradients, exhibit two significant characteristics: \emph{monotonicity} and \emph{Lipschitz continuity/smoothness}. 

Monotonicity is a valuable property that allows us to: $(i)$ transform any Nash equilibrium into a more manageable variational inequality problem, and $(ii)$, akin to convexity in single-agent optimization, it guarantees a consistent progression towards the equilibrium.

Lipschitz continuity, in contrast, is the essential quality that accelerates the convergence rate from \(O(1/\epsilon^2)\) to \(O(1/\epsilon)\). To comprehend the significance of this attribute, let's view the situation through the lens of the minimization participant—named Bob, with his maximization counterpart being Alice. Bob finds himself involved in an online convex minimization game. In the absence of any predictive information about the forthcoming functions, which are influenced by Alice's choices, folklore results suggests that the optimal rate is \(\Theta(1/\sqrt{T})\), or equivalently \(\Theta(1/\epsilon^2)\) iterations. However, should Alice's strategy remain relatively static, Bob can anticipate that the upcoming minimization challenges will be similar, affording him a critical advantage for achieving a rate of \(O(1/T)\), or \(O(1/\epsilon)\).

This insight reveals that the algorithm's performance is deeply tied to the Lipschitz parameter of the gradient operator. Here, the concept of mirror proxies plays a vital role in the proof. Utilizing mirror steps, distinct regularizers enable either more assertive or more refined movements within the strategy space of the spectraplex for both parties. Contrary to the exponential dependency implied by a Frobenius norm, we show that the von Neumann entropy regularizer maintains a dimension-independent Lipschitz constant under a norm that is congruent with the specific needs of our problem.
\\
\\
\textbf{\textsf{Future Work.}} In addition to our novel methodology, we conclude with the following conjecture that aims to tighten our method's performance guarantees and is likely of interest to both the quantum and optimization communities:
\begin{restatable}{conjecture}{mainconj} \label{conj:mainconj}%
    The rates of \Cref{thm:main-result} are tight:
    specifically, there exists a quantum two-player, zero-sum game for which OMMP methods require \(\Omega_d(1/\epsilon)\) iterations to compute \(\epsilon\)-Nash equilibria.
\end{restatable}

\subsection{Paper Organization}
The paper is divided as follows. \Cref{sec:prelims} reviews preliminaries, such as the quantum zero-sum games setup and mathematical concepts necessary for our proofs. \Cref{sec:algos} presents the series of quantum optimization methods for computing approximate Nash equilibria of quantum zero-sum games, as depicted in \Cref{fig:quantum_methods}. It begins with the MMWU method of \cite{jain2009parallel} and shows how the algorithm is an instantiation of the more general Matrix Dual Averaging (MDA) method.
 It then introduces  the Matrix Mirror Prox (MMP) method, which achieves the optimal rate, but requires two gradients per iteration. The section concludes with our proposed OMMP algorithm, which requires only a single gradient call per iteration. Finally, \Cref{sec:convergence} elaborates the convergence analysis of OMMP algorithm, proving that it achieves an $\bigo_d(1/\epsilon)$ rate. The paper concludes with a discussion of different implementations of the OMMP method, via different regularizers. We show that a von Neumann entropy regularizer, instantiating the OMMWU algorithm, achieves a better dimensionality dependence than the Frobenius regularizer. Thus, the OMMWU algorithm achieves the desired $\bigo(d/\epsilon)$ and a quadratic speedup relative to MMWU. For concision, we defer some details and proofs to the Appendix.

\begin{figure}[htpb!]
    \centering
    \includegraphics[width=\textwidth]{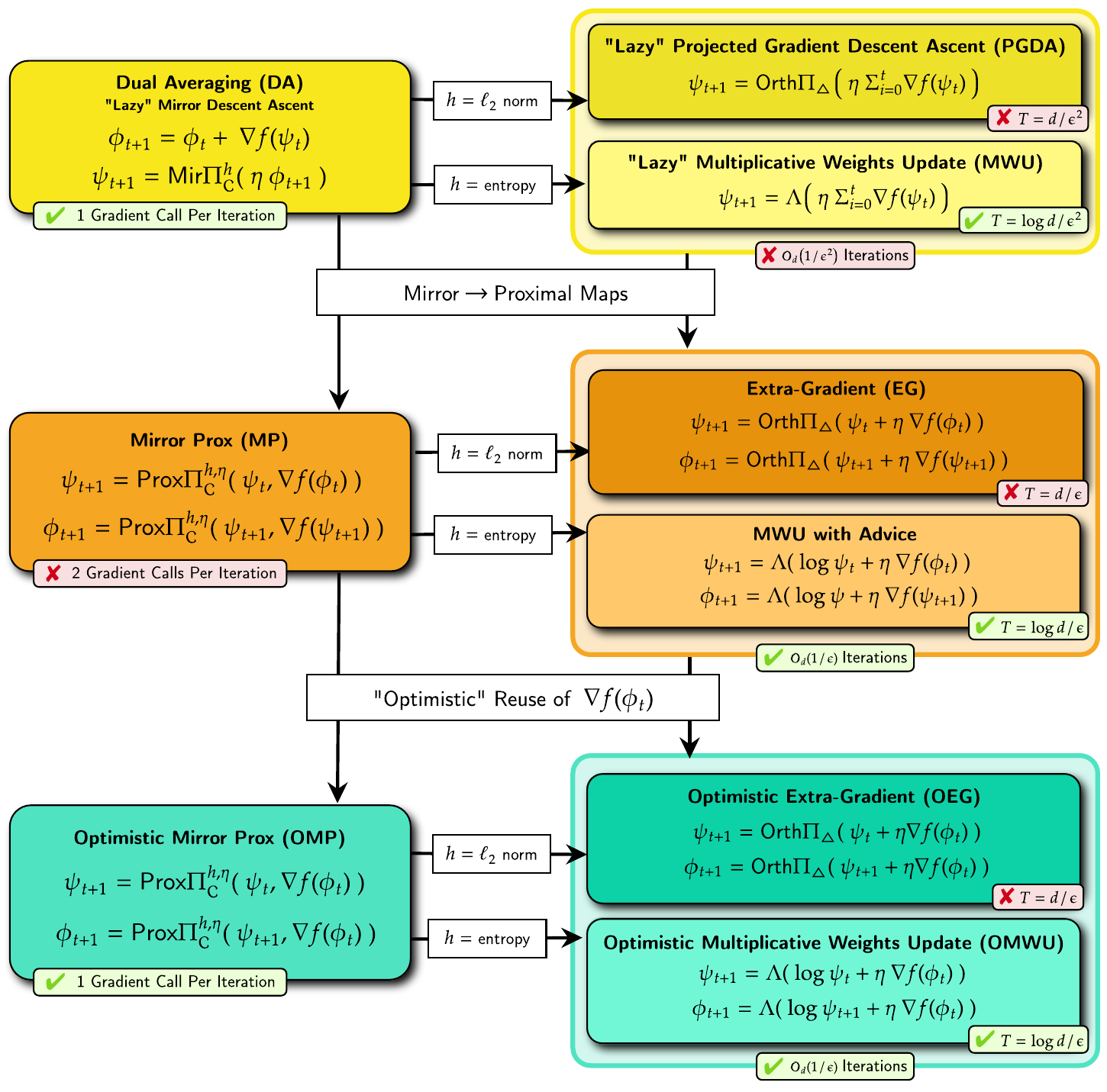}
    \caption{{\textbf{\sffamily Design of Classical Zero-Sum Game Algorithms.}}  This diagram provides the update rules for and relationships between different learning algorithms for classical zero-sum games.
    The left-hand side presents the Bregman generalized methods, parameterized by distance-generating/regularization function $h$. The right-hand side presents instantiations based on the  $\ell_2$ norm and entropy function, inducing an orthogonal projection $\orthproj$ and logit map $\Lambda$, respectively. Note that moving from $h=\ell_2$ to $h=$entropy results in a logarithmic improvement in the total number of rounds $T$, by reducing the dependence on the simplex dimension $d$.  Furthermore, moving from the mirror map $\mirproj$ of the Dual Averaging method (a ``lazy'' variant of the classic Mirror Descent Ascent algorithm) to the proximal maps $\proxproj$ of the Mirror Prox method results in a quadratic improvement in convergence, achieving the desired $O(1/\epsilon)$ rate. Furthermore, by reusing the ``past gradient,'' the Single-Call Mirror Prox method retains the rate of Mirror Prox, but reduces the total number of gradient calls per iteration from two to one. }
    \label{fig:classical_methods}
\end{figure}

\begin{figure}[htpb!]
    \centering
    \includegraphics[width=\textwidth]{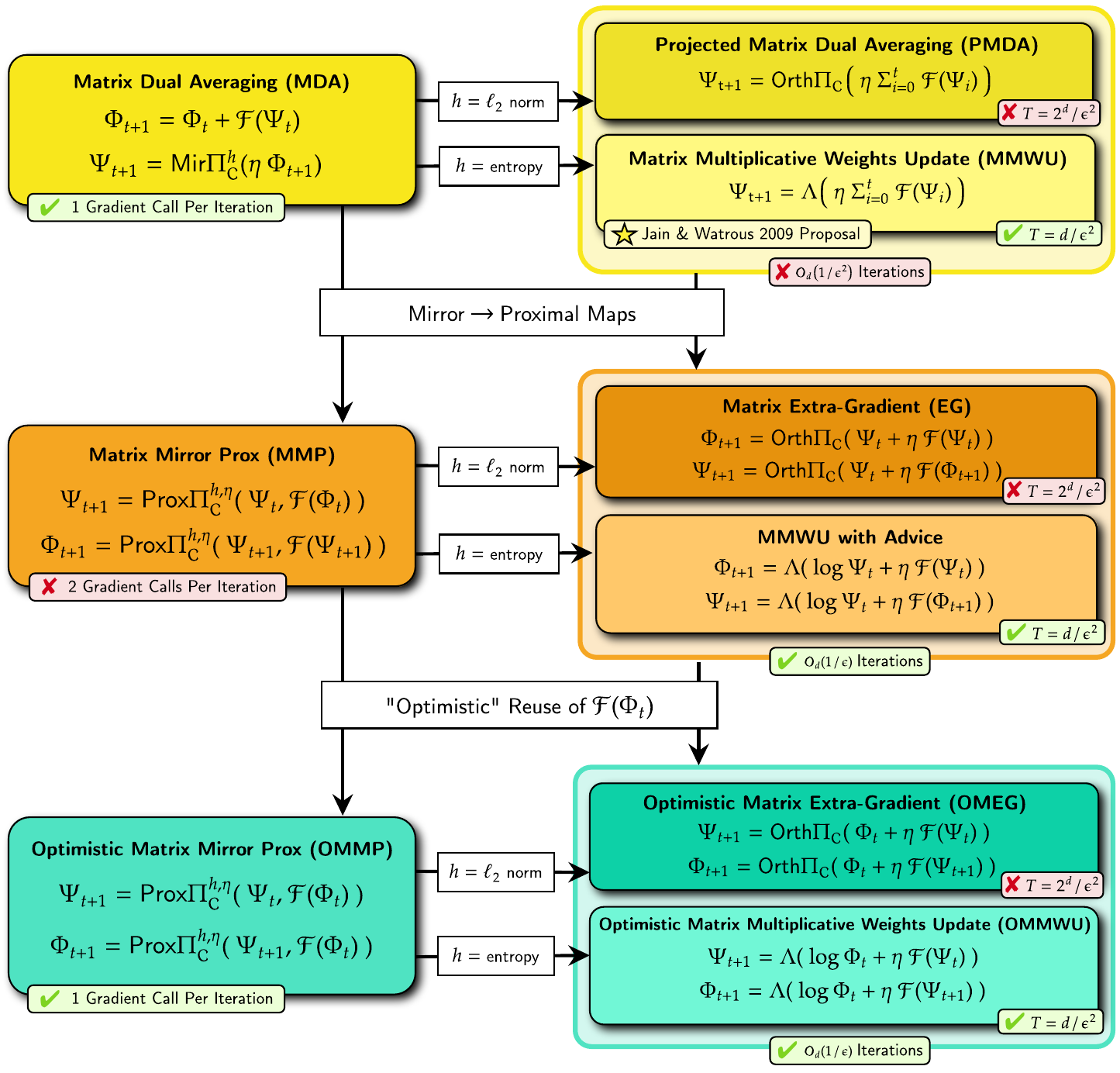}
    \caption{{\textbf{\sffamily Design of Quantum Zero-Sum Game Algorithms.}} This diagram provides the update rules for and relationships between the learning algorithms for quantum zero-sum games
    as proposed in this work. The left-hand side presents the Bregman generalized methods, parameterized by distance-generating/regularization function $h$. The right-hand side presents instantiations based on the von Neumann entropy function and Frobenius ($\ell_2$) norm, inducing an orthogonal projection $\orthproj$ and logit map $\Lambda$, respectively. Note that moving from $\ell_2$ to an entropy function as the regularizer results in a logarithmic improvement in the total number of rounds $T$, by reducing the dependence on the spectraplex dimension $D=4^d$.  Furthermore, moving from the mirror map $\mirproj$ of the Matrix Dual Averaging method (Jain and Watrous' MMWU proposal \cite{jain2009parallel}) to the proximal maps $\proxproj$ of the Matrix Mirror Prox method results in a quadratic improvement in convergence, achieving the desired $O(1/\epsilon)$ rate. Furthermore, by reusing the ``past gradient,'' the Single-Call Matrix Mirror Prox method retains the rate of Matrix Mirror Prox, but reduces the total number of gradient calls per iteration from two to one.  }
    \label{fig:quantum_methods}
\end{figure}

%% file: prelims.tex
\newpage
\section{Preliminaries} \label{sec:prelims}

\subsection{Quantum Zero-Sum Games}
In this work, we consider the restricted class of non-cooperative zero-sum quantum games. Within this setting, two competing players, Alice and Bob, each transmit a mixed quantum state to a referee (henceforth referred to as Roger), who performs a joint measurement on both states. Roger also possesses a utility function that attributes payoffs to the players based on the measurement outcomes. Since this is a zero-sum game, Bob's payoff will be the negative value of Alice's---e.g., if Alice wins \$5, Bob will lose \$5. We will now describe the game, as depicted in \Cref{fig:game_depict}, in more detail.

\begin{figure}[h!]
    \centering
    \includegraphics[width=.7\textwidth]{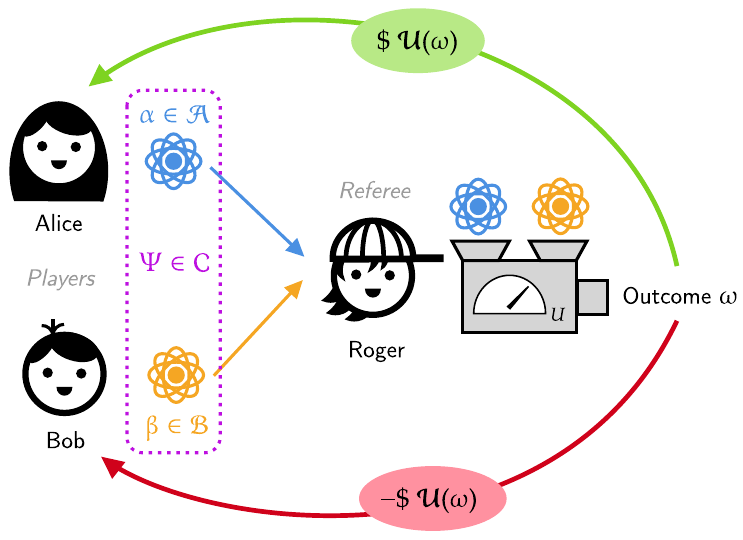}
    \caption{An illustration of a single round of the quantum zero-sum game.}
    \label{fig:game_depict}
\end{figure}

Concretely, Alice and Bob, play the game by independently preparing multi-qubit quantum states $\pone$ and $\ptwo$, respectively, to be sent to Roger. It is important to note that these are \textit{mixed} states, corresponding to ``mixed strategies,'' and that Alice and Bob do not share entanglement. Thus, Alice and Bob's set of possible moves is constrained to the set of density matrices lying in their respective \emph{spectraplexes}.  We will assume Alice prepares $n$-qubit states in her spectraplex,
\begin{align}
    \sone = \{\pone \in\mathcal{H}^{n}_+ : \tr(\pone)=1\},
\end{align}
while Bob prepares $m$-qubit states in his spectraplex,
\begin{align}
    \stwo = \{\ptwo \in\mathcal{H}^{m}_+ : \tr(\ptwo)=1\},
\end{align}
where $\mathcal{H}^{d}_+$ denotes the set of $2^d \times 2^d$-dimensional, positive Hermitian matrices. In \Cref{appendix:spectraplex}, we show  that all spectraplexes are compact and convex sets. 
Denote the joint set of their states as the direct sum
\begin{align}
    \joint = (\pone, \ptwo) \in \sjoint = \sone \oplus \stwo,
\end{align}
with tensor product denoted by $ \joint^\otimes = \pone \otimes \ptwo$.
As a general notational point, we will use lower-case variables (e.g., $\alpha$ or $\beta$) to refer to states in the individual spectraplexes of either Alice ($\sone$) or Bob ($\stwo$) and upper-case variables (e.g., $\Psi$ or $\Phi$) to denote joint states, which lie in the joint space ($\sjoint$).

Alice and Bob send their states to Roger, who performs a joint measurement. This joint measurement has a finite number of possible measurement outcomes, $\outcomes$, defining an $(n+m)$-qubit \emph{positive operator-valued measurement (POVM)},
\begin{align}
    \{ \pout \}_{\omega\in\outcomes}, \text{ where } \sum_{\omega\in\Omega} P_\omega = \Id.
\end{align}
Thus, Roger's probability of observing outcome $\omega\in\Omega$ is 
\begin{equation}
    p_\out (\joint)= \tr \left(\pout^\dagger \joint^\otimes\right).
\end{equation}
Dependent on the measurement outcome, Roger gives Alice and Bob a reward/payoff. Let 
\begin{equation} \label{eqn:util_func}
    \util:\outcomes\rightarrow\mathbb{R}
\end{equation}
denote Alice's \emph{utility} function, which we will assume to be finite-valued and have image $[-1,1]$ under normalization.  We define Alice's \emph{payoff observable} as
\begin{align} \label{eqn:pay_obs_decomp}
    \payob = \sum_{\out\in\outcomes} \util(\out) \pout,
\end{align}  
such that Alice's \emph{expected payoff} is
\begin{equation}
    \exutil_{\text{Alice}}:=\exutil (\joint_t) =\tr (\payob^\dagger \joint^\otimes)=\sum_{\out\in\outcomes} \util (\out)\tr (\pout^\dagger \joint^\otimes) = \sum_{\out\in\outcomes} \util(\out) p_\out (\joint). 
\end{equation}
Since the game is zero-sum, this implies that Bob's utility is $-\util$, with payoff observable $-\payob$ and expected payoff $
    \exutil_{\text{Bob}}:=-\exutil (\joint)$.
In playing the game, Alice and Bob aim to maximize their expected payoff, written as $ \max_\pone \exutil(\pone,\ptwo)$ and $ \max_\ptwo \exutil(\pone,\ptwo)$, respectively.
Since the game is zero-sum, in maximizing his expected payoff, Bob can be seen as trying to minimize Alice's payoff: $
    \max_\ptwo -\exutil(\pone,\ptwo)=\min_\ptwo \exutil(\pone,\ptwo).$
    
Finally, since $\exutil (\pone,\ptwo)$ is a bilinear function and the spectraplexes are compact and convex sets, von Neumann's Min-Max Theorem \cite{v1928theorie} states
\begin{equation}
    \max_\pone \min_\ptwo \exutil (\pone,\ptwo)=\min_\ptwo \max_\pone  \exutil(\pone,\ptwo).
\end{equation}

\subsection{Quantum Nash Equilibria} \label{sec:nash_eq}
A \emph{Nash equilibrium} of the quantum game is a pair of states $(\pone^*, \ptwo^*)$, such that each player has no incentive to change to a different state unilaterally:
\begin{align}
    \exutil (\pone^*,\ptwo^*) &\geq \exutil (\pone,\ptwo^*), \:\: \forall \pone \in \sone \label{eqn:nashone}\\
    \exutil (\pone^*,\ptwo^*) &\leq \exutil (\pone^*,\ptwo), \:\: \forall \ptwo \in \stwo. \label{eqn:nashtwo}
\end{align}
Since $\sjoint$ is convex and $\exutil$ is linear in $\pone$ and $\ptwo$, the existence of Nash equilibria follows from Debreu's equilibrium existence theorem \cite{debreu1952social}.

\begin{figure}[t!]
    \centering
    \includegraphics[width=.7\textwidth]{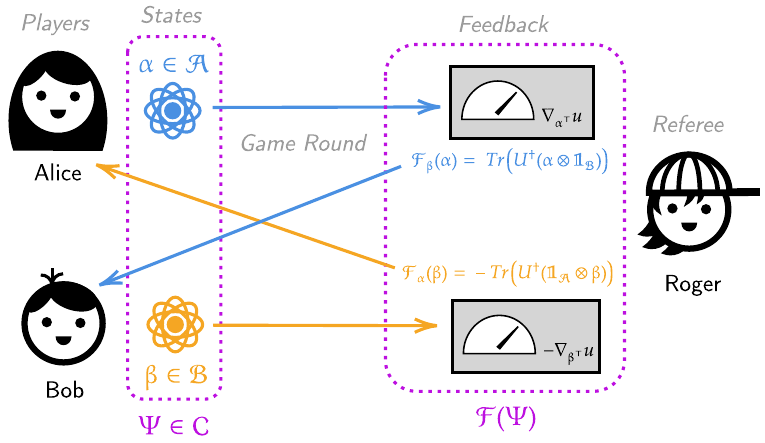}
    \caption{An illustration of the game feedback in a single round of the quantum zero-sum game (gradient-based view).}
    \label{fig:game_feedback}
\end{figure}

In playing the game, Alice and Bob receive feedback from Roger. In the original work on quantum zero-sum games, Jain and Watrous \cite{jain2009parallel} define Alice's feedback as a quantum channel, given by the superoperator $\super :  \stwo  \rightarrow \sone$, applied solely to Bob's state. The superoperator is parameterized by the payoff observable as
\begin{equation} \label{eqn:alice_super}
    \super (\ptwo) = \tr_\stwo [\payob (\Id_\sone \otimes \ptwo^\top)],
\end{equation}
such that the expected utility can be  expressed as
\begin{equation}
    \exutil(\pone,\ptwo)=\tr [\payob (\pone\otimes\ptwo)]=\tr \left[\pone \:\tr_\stwo [\payob (\Id_\sone \otimes \ptwo)]\right]=\tr [\pone \:\super(\ptwo^\top)].
\end{equation}
Similarly, Bob's feedback is defined by the quantum channel, corresponding to adjoint superoperator $\super^*:\sone\rightarrow\stwo$, applied solely to Alice's state.
This adjoint superoperator is defined as
\begin{equation} \label{eqn:bob_super}
    \super^* (\pone) = \tr_\sone [\payob ( \pone^\top \otimes \Id_\stwo)],
\end{equation}
and is uniquely determined by the condition
\begin{equation} \label{eqn:super_relation}
    \exutil(\pone,\ptwo) = \tr [\pone \:\super(\ptwo^\top)] = \tr [\super^*(\pone^\top) \:\ptwo].
\end{equation}

In this work, we will move away from thinking about Roger's feedback in terms of quantum channels towards a gradient-based view, as leveraged in recent work on learning general quantum games \cite{lotidis2023learning}.
It is mathematically equivalent to think of the channel-based feedback in terms of gradients of the expected utility with respect to Alice and Bob's states. 
For a review of complex differentiation and gradients, refer to \Cref{appendix:complex_diff}.
Namely, we define Alice's feedback, or \emph{payoff gradient}, as  
\begin{align} \label{eqn:pgrad1}
    \grad_\pone (\ptwo) &= \nabla_{\pone^\top} \exutil (\pone, \ptwo)  = \nabla_{\pone^\top} \tr \left[\payob^\dagger (\pone\otimes \ptwo)\right] = \tr_\stwo \left[\payob^\dagger (\Id_n \otimes \ptwo)\right],
\end{align}
which is equivalent to $\super(\ptwo^\top)$ of \Cref{eqn:alice_super}. Similarly Bob's feedback is defined as the payoff gradient
\begin{align} \label{eqn:pgrad2}
    \grad_\ptwo (\pone) &= -\nabla_{\ptwo^\top} \exutil (\pone, \ptwo)  = -\nabla_{\ptwo^\top} \tr \left[\payob^\dagger (\pone\otimes \ptwo)\right] = -\tr_\sone \left[\payob^\dagger (\pone \otimes \Id_m)\right],
\end{align}
which is equivalent to $-\super^*(\pone^\top)$ of \Cref{eqn:bob_super}. Similarly to \Cref{eqn:super_relation}, we have that
\begin{align}
    \exutil(\pone,\ptwo) = \inner{\pone}{\grad_\pone (\ptwo)} = -\inner{\ptwo}{\grad_\ptwo (\pone)}.
\end{align}
Under the gradient-based view, $\tr [\pone \: \grad_\pone (\ptwo)]=\tr [\pone \: \nabla_{\pone^\top} \exutil (\pone, \ptwo)]$ can be interpreted as the directional derivative of $\exutil$, as a function of Alice's state, in the direction of Alice's state. Similarly, $-\tr [\grad_\ptwo (\pone) \:\ptwo]=\tr [\nabla_{\ptwo^\top} \exutil (\pone, \ptwo)\:\ptwo]$ can be interpreted as the directional derivative of $\exutil$, as a function of Bob's state, in the direction of Bob's state.
As will be demonstrated throughout the work, moving away from a channel-based to this gradient-based view of the quantum zero-sum game feedback proves invaluable for mapping gradient-based insights from the classical games and optimization literature directly to the quantum setting.

Henceforth, we will denote the joint state as $\joint=(\pone,\ptwo)$, the joint feedback of both players, as follows:\footnote{In a slight abuse of notation, let $\grad(\joint)=\grad(\pone,\ptwo)$.}
\begin{align} \label{eqn:grad_joint}
    \grad (\joint) = \Big(\grad_\pone (\ptwo), \grad_\ptwo (\pone)\Big),
\end{align}
and the game's Nash equilibrium as $\joint^*=(\pone^*,\ptwo^*)$.
Thus, as depicted in \Cref{fig:game_feedback}, in each round $t$ of the game, after playing the joint state $\joint_t$, Alice will receive the individual payoff gradient $\grad_{\pone_t} (\ptwo_t)$ and Bob will receive $\grad_{\ptwo_t} (\pone_t)$, comprising joint feedback $\grad(\joint_t)$.

By using standard arguments~\cite{scutari2010convex}, it can be shown that the solutions of the variational inequality
\begin{align} \label{eqn:vi}
    \inner{(\joint-\joint^*)}{\grad (\joint^*)} \leq 0, \:\:\:\: \forall\: \joint \in \sjoint,
\end{align}
are the game's Nash equilibria, $\joint^*$.\footnote{Note that in the direct sum notation, for matrices $A_1$, $A_2$, $B_1$, and $B_2$, where $\text{dim}(A_1)=\text{dim}(A_2)=n\times n$ and $\text{dim}(B_1)=\text{dim}(B_2)=m\times m$, we have that $\tr [(A_1,B_1) (A_2,B_2)]=\tr[A_1 A_2]+\tr[B_1 B_2]$.}
Specifically, solutions satisfying \Cref{eqn:vi} are known as \textit{strong solutions} to the variational inequality. Alternatively, \textit{weak solutions} to the variational inequality are $\joint^*\in\sjoint$ satisfying
\begin{equation}\label{eqn:vi_weak}
    \inner{(\joint^*-\joint)}{\grad (\joint)} \geq 0, \:\:\:\: \forall\: \joint \in \sjoint.
\end{equation}
We offer intuitive visualizations of both the weak and strong solutions of variational inequalities as vector inner products in \Cref{fig:vi_soln}.

\begin{figure}[t!]
    \centering
    \includegraphics[width=.5\textwidth]{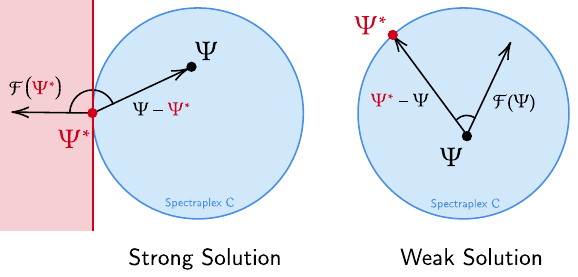}
    \caption{Intuitive illustrations of the strong and weak solutions $\Psi^*$ of the variational inequalities given in \Cref{eqn:vi} and \Cref{eqn:vi_weak}, respectively. For strong solutions, the inner product between $\Psi-\Psi^*$ and $\grad(\Psi^*)$ must be non-positive for all $\Psi$ in the joint spectraplex $\sjoint$. Meanwhile, for weak solutions, the inner product between $\Psi^*-\Psi$ and $\grad(\Psi)$ must be non-negative for all $\Psi\in\sjoint$.}
    \label{fig:vi_soln}
\end{figure}

As we will discuss in \Cref{sec:convergence} and prove in \Cref{sec:f_props}, $\grad$ is both monotone (\Cref{lemma:monotone}) and Lipschitz-continuous (\Cref{lemma:lipschitz}). Additionally, in \Cref{appendix:properties}, we prove the following two important lemmas:
\begin{restatable}{lemma}{strongtoweak}\label{thm:strong_weak}%
   If $\grad$ is monotone, then a strong solution is a weak solution.
\end{restatable} 
\vspace{-0.15in}
\begin{restatable}{lemma}{weaktostrong}\label{thm:weak_strong}%
   If $\grad$ is Lipschitz continuous, then a weak solution is a strong solution.
\end{restatable}
\noindent Therefore, by the monotonicity and Lipschitzness of $\grad$, the strong and weak solutions (described, respectively, by \Cref{eqn:vi} and \Cref{eqn:vi_weak}) are equally valid criteria for Nash equilibria of the finite-valued zero-sum quantum games under consideration. We will generally leverage the weak-solution criterion in our analysis.

Finally, in this work, we focus on the iteration complexity of our methods for computing an \textit{$\epsilon$-approximate Nash equilibrium}, $\tilde{\joint}^*=(\tilde{\pone}^*,\tilde{\ptwo}^*)$, which---as defined by \cite{jain2009parallel}---must satisfy
\[
    \max_{\pone} \exutil(\pone,\tilde{\ptwo}^*)-\epsilon ||\payob||_{\infty} \leq \exutil(\tilde{\pone}^*,\tilde{\ptwo}^*) \leq \min_{\ptwo} \exutil(\tilde{\pone}^*,\ptwo)+\epsilon ||\payob||_{\infty},
\]
where $||U||_{\infty}$ denotes the infinity norm of the payoff observable $\payob$. From this definition and the maximal utility of Nash equilibria it follows that
\begin{align} \label{eqn:dual_gaps}
        \epsilon ||\payob||_{\infty} \geq \exutil(\tilde{\pone}^*,\tilde{\ptwo}^*) - \min_{\ptwo} \exutil(\tilde{\pone}^*,\ptwo) &\geq 0 \\
        \epsilon ||\payob||_{\infty} \geq \max_{\pone} \exutil(\pone,\tilde{\ptwo}^*) - \exutil(\tilde{\pone}^*,\tilde{\ptwo}^*) &\geq 0.
\end{align}
Summing these equations gives the \textit{duality gap}, 
\begin{equation}
    \dg(\pone',\ptwo')= \max_{\pone} \exutil(\pone,\ptwo') - \min_{\ptwo} \exutil(\pone',\ptwo),
\end{equation}
which effectively serves as a measure of proximity between any proposed states, $(\pone', \ptwo')$, and an exact Nash equilibrium, $(\pone^*, \ptwo^*)$. From \Cref{eqn:dual_gaps} it follows that for an $\epsilon$-Nash equilibrium, the duality gap is bounded as
\begin{equation}
    \epsilon' = 2\epsilon||\payob||_{\infty} \geq \dg(\pone',\ptwo')\geq 0.
\end{equation}
Therefore, to evaluate the error of a proposed solution $\tilde{\joint}$, relative to any true Nash equilbrium, we leverage the notion of an \emph{error/merit function}, which corresponds to the duality gap in the setting of min-max optimization:
\begin{align} \label{eqn:error_func}
    \textnormal{Error}(\Tilde{\joint})=\sup_{\joint\in\sjoint} \inner{(\joint-\Tilde{\joint})}{\grad(\joint)}.
\end{align} 
It is easy to check that $\textnormal{Error}(\Tilde{\joint})\geq 0$ for all $\Tilde{\joint}$ and $\textnormal{Error}(\Tilde{\joint})= 0$ if and only if $\Tilde{\joint}=\joint^*$.

\subsection{Convexity, Smoothness, \& Duality} \label{sec:conv_dual}
Crucial to our analysis will be functional notions of convexity, smoothness, and duality~\cite{boyd_vandenberghe_2004, shalev2012online}. We begin by defining convexity and strong convexity of a differentiable function $f:\sjoint \rightarrow \mathbb{R}$.
\begin{definition}[Convex Function] \label{def:convex}
A differentiable function  $f:\sjoint \rightarrow \mathbb{R}$ is \emph{convex} if
\begin{align}
    f (X) \geq f (Y) + \langle\nabla f(Y),X-Y\rangle, \hspace{0.1in} \forall X,Y \in \sjoint.
\end{align}
\end{definition}
\begin{definition}[$\mu$-Strongly Convex Function] \label{def:strong_convex}
A convex differentiable function  $f:\sjoint \rightarrow \mathbb{R}$ is \emph{$\mu$-strongly convex} with respect to norm $\|\cdot\|$ if
\begin{align}
    f (X) \geq f (Y) + \langle\nabla f(Y),X-Y\rangle+\frac{\mu}{2} \|X - Y\|^2, \hspace{0.1in} \forall X,Y \in \sjoint,
\end{align}
or equivalently, following from the definition of the Bregman divergence in \Cref{eqn:breg_div}, 
\begin{align}
    \bregdiv_f (X \| Y) \geq \frac{\mu}{2} \|X - Y\|^2, \hspace{0.1in} \forall X,Y \in \sjoint.
\end{align}
\end{definition}
Since strong convexity implies standard convexity (with $\mu=0$), the optimal points (minima/maxima) of any strongly convex function must satisfy first-order optimality.
\begin{restatable}[First-Order Optimality Condition]{fact}{firstopt}\label{thm:firstopt}%
   For an optimization problem,
   \begin{align*}
       \min f(X) \:\:\text{ subject to } \:\: X\in \mathcal{X},
   \end{align*}
   with convex differentiable $f$, a feasible point $X$ is optimal if and only if
   \begin{align*}
       \langle\nabla f(X), Y-X\rangle \geq 0, \:\: \forall Y\in\mathcal{X}.
   \end{align*}
\end{restatable}
Another property of interest for the convex differentiable function $f$ is smoothness, which implies Lipschitz continuity of $\nabla f$.
\begin{definition}[$\beta$-Smooth Function]\label{def:smoothness-ineq-definition} A convex differentiable function  $f:\sjoint \rightarrow \mathbb{R}$ is $\beta$-smooth with respect to norm $\|\cdot\|$ if
\begin{align}
    f (X) \leq f (Y) + \langle\nabla f(Y),X-Y\rangle+\frac{\beta}{2} \|X - Y\|^2, \hspace{0.1in} \forall X,Y \in \sjoint,
\end{align}
or equivalently, following from the definition of the Bregman divergence in \Cref{eqn:breg_div}, 
\begin{align}
    \bregdiv_f (X \| Y) \leq \frac{\beta}{2} \|X - Y\|^2, \hspace{0.1in} \forall X,Y \in \sjoint.
\end{align}
\end{definition}
Before establishing the Lipschitz condition of $\nabla f$, we introduce the \emph{dual norm}.
\begin{definition}[Dual Norm]\label{defn:dual_norm} If the generic norm of a matrix $X$ is denoted $\|X\|$, the dual norm is defined as
\begin{align}
    \| X\|_* = \sup \{\langle Y, X \rangle : \|Y\| \leq 1\}.
\end{align}
\end{definition}
As shown in \cite{watrous2011theory,meyer2023matrix}, the Frobenius norm $\|\cdot\|_F$ is self-dual, meaning $\| X\|_*=\| X\|_{F}$.

\begin{proposition}[$\beta$-Lipschitz] 
\label{def:smoothness-grad-definition}A convex differentiable function  $f:\sjoint \rightarrow \mathbb{R}$ is $\beta$-smooth with respect to norm $\|\cdot \|$ if its gradient $\nabla f$ is $\beta$-Lipschitz continuous,
\begin{align}
    \| \nabla f(X)-\nabla f(Y)\|_* \leq \beta \|X - Y\| \hspace{0.1in} \forall X,Y \in \sjoint.
\end{align}
\end{proposition}
For the equivalence of \Cref{def:smoothness-grad-definition} and \Cref{def:smoothness-ineq-definition}, see \cite{zhou2018fenchel}.
\noindent Geometrically, $\beta$-smoothness can be interpreted as a convex quadratic \emph{upper} bound,
\begin{align}
    g_X^+(Y) = f (Y) + \langle\nabla f(Y),X-Y\rangle+\frac{\beta}{2} \|X - Y\|^2,
\end{align}
of $f$ such that $g_X^+(X) = f(X)$ and $g_X^+(Y) \geq f(Y), \:\forall Y \in \sjoint$. Meanwhile, $\mu$-strong convexity can be interpreted as a convex quadratic \emph{lower} bound,
\begin{align}
    g_X^-(Y) = f (Y) + \langle\nabla f(Y),X-Y\rangle+\frac{\mu}{2} \|X - Y\|^2,
\end{align}
of $f$ such that $g_X^-(X) = f(X)$ and $g_X^-(Y) \leq f(Y), \:\forall Y \in \sjoint$.   In this way, strong convexity is dual to smoothness.

Fenchel conjugacy relates two equivalent representations of a convex function. In the standard representation, a convex function $f$ is represented by pairs $(X,f(X))$ consisting of coordinates $X$ and corresponding function evaluations $f(X)$. Alternatively, as depicted in \Cref{fig:fenchel_conj}, a convex function can be represented by its tangent at each point. These tangents are encoded as the pairs $(\theta, f^*(\theta))$, where $\theta=\nabla f(X)$ for some $X$ is the gradient and  $-f^*(\theta)$ is the intercept. The function $f^*$, which calculates the tangent's intercept $f^*(\theta)$ from a given gradient $\theta$, is known as the Fenchel conjugate function, formally defined as follows.
\begin{definition}[Fenchel Conjugate] The \emph{Fenchel conjugate} of a function $f:\sjoint \rightarrow\mathbb{R}$ is the function $f^*:\sjoint \rightarrow\mathbb{R}$ such that
\begin{align}
    f^*(\theta) = \max_{X \in\sjoint} \left\{ \langle X, \theta\rangle - f(X) \right\}.
\end{align}
\end{definition}
\noindent Relating back to \Cref{defn:dual_norm} of the dual norm, if it were the case that for all $X$, $f(X)=0$, then $f^*(\theta)=\max_{X \in\sjoint} \langle X, \theta\rangle =||\theta||_*$. Therefore, the Fenchel conjugate can be seen as maximizing the dual norm while also minimizing the penalty induced by function $f$. In this way, $f$ can be seen as a ``regularizer'' function. 

The definition of Fenchel conjugacy immediately implies the \emph{Fenchel-Young Inequality},
\begin{align}
    \forall X \in \sjoint, \hspace{0.1in}  f^*(\theta) \geq \langle X, \theta\rangle - f(X).
\end{align}
Furthermore, taking the gradient of the Fenchel conjugate should result in the value of $X$ which maximizes $\langle X, \theta\rangle - f(X)$, resulting in \emph{Danskin's Theorem}, as follows. For more details, refer to Proposition 4.5.1 of \cite{bertsekas2003convex}.
\begin{proposition}[Danskin's Theorem] \label{lemma:danskin}  Let $f:\sjoint \rightarrow \mathbb{R}$ be a strongly convex function. For all $X \in \sjoint$, we have
\begin{align*}
    \nabla f^* (\theta) = \arg\min_{X\in\sjoint} \left\{ f(X)-\langle X, \theta\rangle\right\}.
\end{align*}
\end{proposition}

Finally, as stated below and explained further in 
Corollary 3.5.11 of \cite{zalinescu2002convex}, 
there is a duality between the the convexity of $f$ and the smoothness of $f^*$.
\begin{lemma}[Strong-Smooth Duality] \label{lemma:convex_smooth}
    A closed convex function $f:\sjoint \rightarrow \mathbb{R}$ is $\alpha$-strongly convex with respect to norm $\|\cdot\|$ if and only if $f^*$ is $\frac{1}{\alpha}$-smooth with respect to the corresponding dual norm $\|\cdot\|_*$. 
\end{lemma}

\begin{figure}[t]
    \centering
    \includegraphics[width=.8\textwidth]{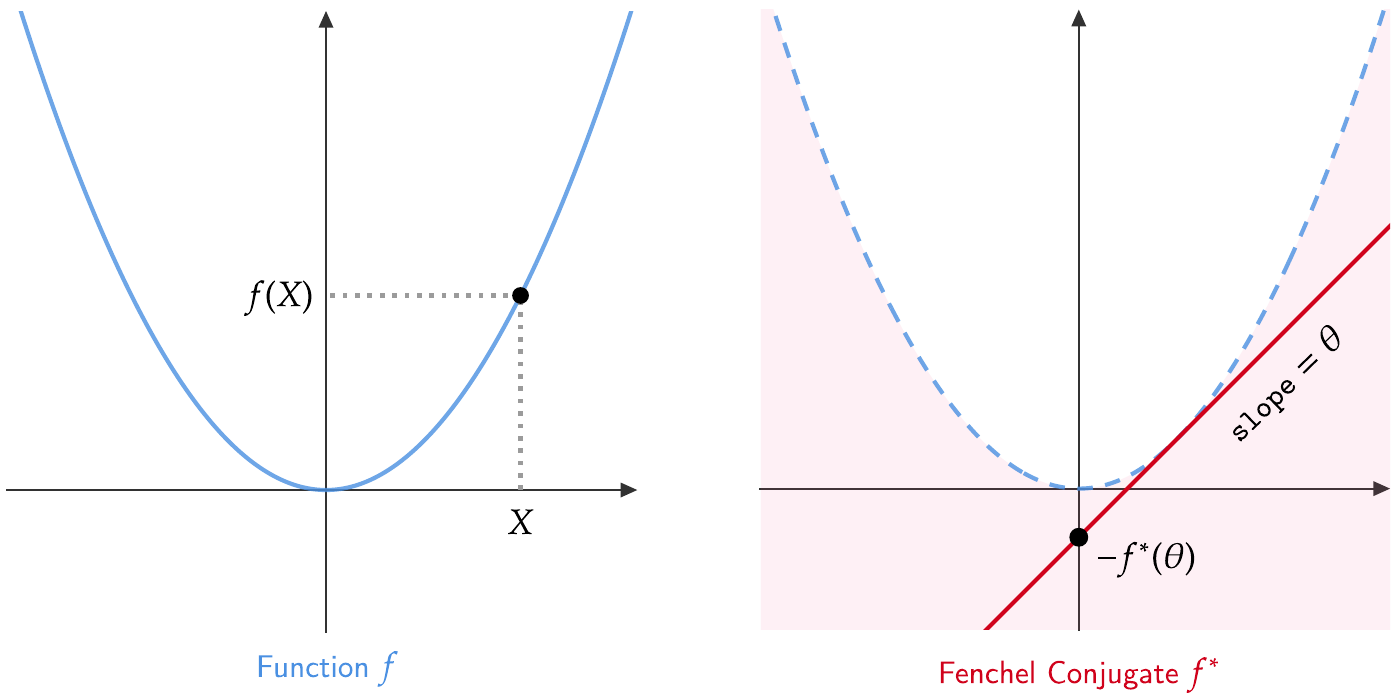}
    \caption{Illustration of a function \( f \) and its Fenchel conjugate \( f^* \). The function \( f \) is represented by a collection of pairs \( (X, f(X)) \), while its Fenchel conjugate \( f^* \) is visualized as the set of tangents to \( f \), characterized by pairs \( (\theta, f^*(\theta)) \). [Adapted from Figure 2.3 in \cite{shalev2012online}.]}
    \label{fig:fenchel_conj}
\end{figure}

\subsection{Bregman Divergence \& Mirror Map}
Critical to understanding our proposed algorithm and its improved convergence over Jain and Watrous' MMWU proposal are the notions of ``mirror maps'' and Bregman divergence. 

Historically, the first proposals for constrained Gradient Descent Ascent leveraged Euclidean ($\ell_2$ norm) projections in order to ensure that the algorithm output would lie within the feasible set. However, it was eventually realized that for certain problems, as we will discuss in the following subsection, regularization based on this $\ell_2$ norm results in suboptimal constants in the convergence rates and, by leveraging different distance metrics, faster rates can be achieved. In order to develop algorithms which could encompass any such desired distance metric, the Bregman divergence was introduced. 

\begin{figure}[t]
    \centering
    \includegraphics[width=.5\textwidth]{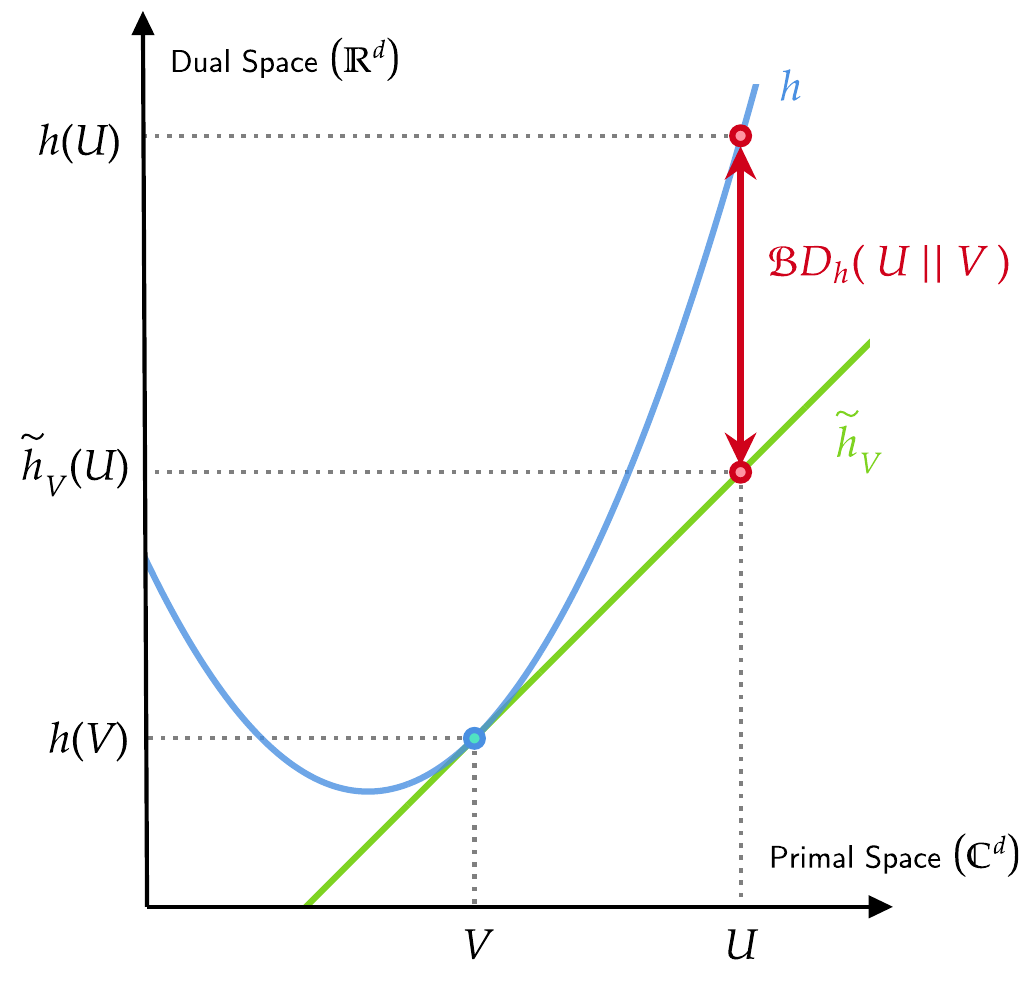}
    \caption{\textbf{\sffamily Bregman Divergence.} The Bregman divergence $\bregdiv_h(U||V)$ between two states in the primal space $U,V \in \mathcal{X}^d$, relative to regularizer $\regfunc$, is given by the difference in the dual space ($\mathbb{R}^{d}$) between $\regfunc(U)$ and its linear approximation via the first-order Taylor expansion of $\regfunc$ around $V$, $\tilde{h}_V (U)$.}
    \label{fig:breg_div}
\end{figure}
The Bregman divergence is parameterized by a ``distance-generating'' or ``regularizer'' function. Formally, a \emph{distance-generating function} (DGF) on the $d$-dimensional spectraplex \(\mathcal{H}_+^d\) is defined as a convex, lower semi-continuous function \(\regfunc: \mathbb{C}^d \rightarrow \mathbb{R} \cup \{\infty\}\), where the \emph{effective} domain of \(h\) is \(\mathcal{X}^d\)
and $\mathbb{C}^d$ is the space of $2^d\times2^d$ complex matrices. Denoting the space of $2^d\times2^d$ real matrices as $\mathbb{R}^{d}$, let
$\partial h: \mathcal{H}_+^d \rightarrow \mathbb{R}^{d}$ denote the \emph{subdifferential} of $\regfunc$ defined as\footnote{The subdifferential is  the set-valued generalization of the derivative for convex functions which are not necessarily differentiable.}
\begin{equation}
\label{eq:subdiff}
\partial\regfunc(X):= \left\{V\in\mathbb{R}^{d}:
		\regfunc(Y) \geq \regfunc(X) + \tr[(Y-X)^\dagger V], \:\:  \forall \: Y\in\mathcal{H}_+^d\right\}.
\end{equation}
 The \emph{domain of subdifferentiability} of $\regfunc$ is defined as
\begin{equation} \label{eq:subdom} 
    \textrm{dom}\partial\regfunc:= \{X\in\mathcal{X}^d:\partial\regfunc (X) \neq \emptyset\}.
\end{equation}

Having defined the DGF $\regfunc$ and its subdifferential $\partial\regfunc$, we are ready to define the Bregman divergence. Given two states $U,V\in\mathcal{H}_+^d$, the Bregman divergence induced by $h$ is the non-negative function $\bregdiv_\regfunc: \textrm{dom}\regfunc\times  \textrm{dom}\partial\regfunc\rightarrow\mathbb{R}$ defined as 
\begin{equation} \label{eqn:breg_div}
    \bregdiv_\regfunc (U || V) := \regfunc(U)-\regfunc(V)-\tr[(U-V)^\dagger\partial \regfunc(V)],
\end{equation}
where the subdifferential maps from states in the \emph{feasible region} $\mathcal{H}_+^d$ of the \emph{primal space}  $\mathbb{C}^d$,  to the \emph{dual space} of subgradients $\mathbb{R}^d$. The function $\partial h$, which performs a primal-to-dual mapping, is often referred to as a \emph{mirror map}.

Intuitively, as illustrated in \Cref{fig:breg_div}, the Bregman divergence $\bregdiv_h(U||V)$ between two states $U$ and $V$ in the primal space, relative to regularizer $\regfunc$, is given by the difference in the dual space between $\regfunc(U)$ and its linear approximation via the first-order Taylor expansion of $\regfunc$ around $V$,
\begin{equation}
    \widetilde{\regfunc}_V (U) = \regfunc(V)+\tr[(U-V)^\dagger \partial\regfunc(V)].
\end{equation}
Since $\regfunc$ is convex, this difference is always non-negative. Note that $\bregdiv_h(U || V)$ is not necessarily equivalent to $ \bregdiv_h(V||U)$. 

Finally,  our analysis will use
the well-known \emph{Three-Point Identity},
\begin{equation}
\label{eq:3point}
    \tr[(Z - X)^\dagger\nabla_{X^\top} \bregdiv_h(X||Y) ] = \bregdiv_h(Z||Y)  - \bregdiv_h(X||Y)  - \bregdiv_h(Z|| X),
\end{equation}
which follows from the definition of the Bregman divergence \cite{mertikopoulos2018optimistic}.

\subsection{Mirror \& Proximal Steps} \label{sec:mirror_prox_step}
We will now introduce ``mirror steps'' and ``proximal steps,'' which incorporate the mirror map generalization of distance in constrained optimization updates.\footnote{For further background, we refer the interested reader to \cite{mertikopoulos2018optimistic} and Chapter 4 of \cite{Bubeck2015}.} These steps implement the algorithmic state updates of the classical games algorithms discussed in the Prior Work (\Cref{sec:prior_work}). Similar to a standard constrained optimization gradient step, the goal of these procedures is to produce an update that lies in the feasible set, closer to the optimum. However, by generalizing the Euclidean distance of traditional Gradient Descent Ascent steps, mirror and proximal steps enable incorporation of domain-specific knowledge about the geometry of the optimization task into the updates. This can lead to substantial improvements in convergence rate.

For intuition, let us first consider  a classical single-agent case as an illustrative example. When optimizing a function \( f \) over a \((d-1)\)-dimensional simplex, restricted by the \( \ell_{\infty} \)-norm, standard Projected Subgradient Descent (PSD) reaches a convergence rate of \( \mathcal{O}(\sqrt{d/T}) \). In contrast, classical Mirror Descent Ascent using a negative Shannon entropy regularizer significantly improves the rate to \( \mathcal{O}(\sqrt{\log d/T}) \). 

Our objective is to adapt this methodology to the more challenging, exponentially larger \( (4^d-1) \)-dimensional spectraplex \( \mathcal{H}^{d}_+ \), in a min-max framework. Specifically, we aim to incorporate domain-specific regularization techniques to achieve an \( \mathcal{O}(d/T) \) convergence rate. This would offer an exponential dimensionality improvement from standard Projected Subgradient Ascent Descent, which has an \( \mathcal{O}(2^{d}/\sqrt{T}) \) convergence rate \cite{nesterov_primal-dual_2009,Ang2023ProjectedGA}, and a quadratic speed-up over the Jain-Watrous algorithm, which has an \( \mathcal{O}(d/\sqrt{T}) \) convergence rate \cite{jain2009parallel}.

Note that intuitions for the following definitions and proofs are provided in \Cref{sec:mirror_prox_append}.
\begin{figure}[t]
    \centering
    \includegraphics[width=.48\textwidth]{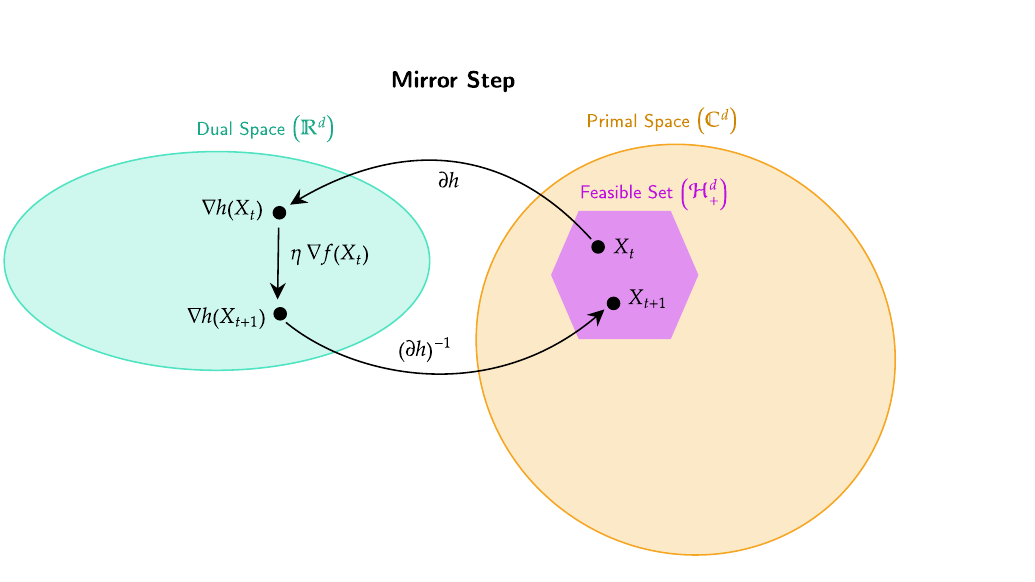}
    \hspace{.02\textwidth}
    \includegraphics[width=.48\textwidth]{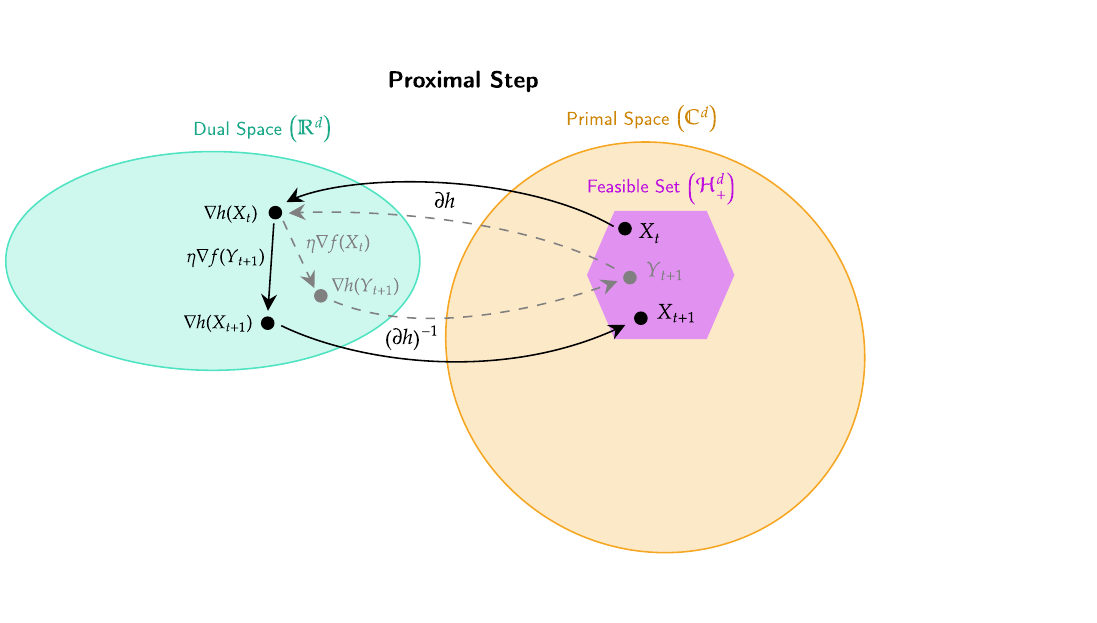}
    \caption{\textbf{\sffamily Mirror \& Proximal Step.}Visualizations of the mirror and proximal steps described in \Cref{alg:mirror} and \Cref{alg:prox}, respectively. [Adapted from Figures 4.1 and 4.2 in \cite{Bubeck2015}.]}
    \label{fig:mirror_prox_map}
\end{figure}

\subsubsection{Mirror Step}
Given a function $f:\mathbb{C}^d\rightarrow\mathbb{R}$ that we aim to optimize, convex feasible set $\mathcal{H}_+^d\subset \mathbb{C}^d$, DGF $\regfunc$, subgradients $\partial\regfunc$,  and initial primal state $X_t\in\mathcal{H}_+^d$, each call to a \textit{mirror step} oracle generates an updated state $X_{t+1}$  via \Cref{alg:mirror} (as depicted in \Cref{fig:mirror_prox_map}). In Step (3), the mirror map, $\partial h^{-1}$, can be implemented as follows.
\begin{restatable}[Mirror Map]{definition}{mirrorproject}\label{prop:mirproj}
   For a state $D$ in the dual space $\mathcal{D}$, the \emph{mirror map}, or \emph{regularized best response}, 
    \begin{align}
        \mirproj_{\mathcal{P}}^\regfunc (D) :=  \underset{P\in\mathcal{P}}{\arg\max} \{\tr[D^\dagger P]-\regfunc (P)\},
    \end{align}
    implements the inverse mirror map, $\partial h^{-1}:\mathcal{D}\rightarrow \mathcal{P}$, mapping $D\in\mathcal{D}$ to the corresponding state $P$ in the feasible region $\mathcal{P}$ of the primal space.
\end{restatable}

\begin{algorithm} [t]
    \caption{Mirror Step}\label{alg:mirror}
    \begin{enumerate}
        \item Initial primal state $X_t \in\mathcal{H}_+^d$ is mapped to a dual state $\nabla\regfunc(X_t)\in
        \partial\regfunc (X_t) \subset \mathbb{R}^d$.
        \item Gradient $\nabla f: \mathbb{C}^d \rightarrow \mathbb{R}^{d}$ is calculated with respect to $X_t$.
        \item In the dual space, the state is updated according to the gradient with step size $\step$, resulting in a new dual state 
        \begin{equation*}
            \nabla\regfunc (X_{t+1}) = \nabla\regfunc (X_t)+ \step \nabla f (X_t) \:\: \in \mathbb{R}^{d}.
        \end{equation*}
        \item Mapping back to the primal space, this corresponds to the updated feasible state
        \begin{equation*}
            X_{t+1}\in\partial h^{-1}(\nabla\regfunc (X_{t+1})+ \step \nabla f (X_t)) \:\: \in \mathcal{H}_+^d.
        \end{equation*}
    \end{enumerate}
\end{algorithm}

\begin{algorithm} [t!]
    \caption{Proximal Step}\label{alg:prox}
    \begin{enumerate}
    \item From initial feasible primal state $X_t \in \mathcal{H}_+^d$, updated feasible primal state $Y_{t+1} \in \mathcal{H}_+^d$ is generated via the Mirror Step procedure of \Cref{alg:mirror}.
    \item Initial feasible primal state $X_t$ is mapped to a dual state $\nabla\regfunc(X_t)\in
    \partial\regfunc (X_t)\subset \mathbb{R}^d$.
    \item Gradient $\nabla f: \mathbb{C}^d \rightarrow \mathbb{R}^{d}$ is calculated with respect to $Y_{t+1}$.
    \item In the dual space, the state is updated according to the gradient with step size $\step$, resulting in a new dual state 
    \begin{equation*}
        \nabla\regfunc (X_{t+1}) = \nabla\regfunc (X_t)+ \step \nabla f (Y_{t+1}) \:\: \in \mathbb{R}^{d}.
    \end{equation*}
    \item Mapping back to the primal space, this corresponds to the updated feasible state
    \begin{equation*}
        X_{t+1}\in\partial h^{-1}(\nabla\regfunc (X_t)+ \step \nabla f (Y_{t+1})) \:\: \in \mathcal{H}_+^d.
    \end{equation*}
\end{enumerate}
\end{algorithm}

\noindent Therefore,  the full mirror step procedure, mapping from $X_t$ to $X_{t+1}$, is achieved via the mirror map of $Y_t = \nabla \regfunc (X_t) + \step \nabla f(X_t)$ on the feasible set $\mathcal{H}_+^d$, as
\begin{equation} \label{eqn:mirror_proj_update}
    X_{t+1} = \mirproj_{\mathcal{H}_+^d}^\regfunc (\nabla \regfunc (X_t) + \step \nabla f(X_t)).
\end{equation}
\noindent Alternatively, the mirror step procedure can be expressed as a \emph{mirror proximal map}. 

\begin{restatable}[Mirror Step]{lemma}{mirrorprox}\label{prop:mirstep}
   For state $X_t$ in the feasible set $\mathcal{H}_+^d$ of primal space $\mathbb{C}^d$,
    \begin{align} \label{eqn:mirror_prox_proj_update}
        \mirproxproj_{\mathcal{X}^d}^{\regfunc,\step,f} (X_t) :=  \underset{X_{t+1}\in\mathcal{H}_+^d}{\arg\min} \:  \left\{  \tr[\nabla f(X_t)^\dagger X_{t+1}]+ \frac{1}{\step}\bregdiv_\regfunc (X_{t+1} || X_t)\right\},
    \end{align}
    implements the mirror step mapping to an updated feasbile set primal state $X_{t+1}\in\mathcal{H}_+^d$.
\end{restatable}  
\subsubsection{Proximal Step} \label{sec:proxy}
We refer to the mirror prox projection as a \emph{proximal map}.  Its generalization comprises the \emph{proximal step} oracle, described in \Cref{alg:prox} and depicted in \Cref{fig:mirror_prox_map}. 
\begin{restatable}[Proximal Map]{definition}{proximaldefn}\label{prop:proximalproj}
   For a state $P$ in the feasible set of the primal space $\mathcal{P}$ and a state $D$ in the dual space, the \emph{proximal map},
    \begin{align}
        \proxproj^{\regfunc}_{\mathcal{P}}(P, D):&= \underset{P^*\in\mathcal{P}}
        {\arg\min}\left\{\tr[D^\dagger (P^*-P)] + \bregdiv_\regfunc (P^*||P)\right\}\nonumber\\
        &= \underset{P^*\in\mathcal{P}}
        {\arg\min}\left\{\tr[D^\dagger P^*] + \bregdiv_\regfunc (P^*||P)\right\},
    \end{align}
    implements the proximal mapping to an updated feasible set primal state $P^*\in\mathcal{P}$.
\end{restatable} 
\noindent Applying this definition to our problem of interest yields the following map.
\begin{restatable}[Proximal Step]{lemma}{proximalproj}\label{prop:proximalprojapplied}
   For state $X_t\in\mathcal{H}_+^d$ in the feasible set of the primal space,
    \begin{align}\label{eqn-prop:proximalprojapplied}
        \proxproj^{\regfunc,\step}_{\mathcal{H}_+^d}(X_t, \nabla f(Y_{t+1})):&= \underset{X_{t+1}\in\mathcal{H}_+^d}
        {\arg\min}\left\{\tr[\nabla f(Y_{t+1})^\dagger (X_{t+1}-X_t)] +\frac{1}{\step}\bregdiv_\regfunc (X_{t+1}||X_t)\right\} \nonumber \\
        &= \underset{X_{t+1}\in\mathcal{H}_+^d}
        {\arg\min}\left\{\tr[\nabla f(Y_{t+1})^\dagger X_{t+1}] +\frac{1}{\step}\bregdiv_\regfunc (X_{t+1}||X_t)\right\},
    \end{align}
    implements the proximal step mapping to an updated feasbile set primal state $X_{t+1}\in\mathcal{H}_+^d$.
\end{restatable} 
\noindent Note that $\proxproj^{\regfunc,\step}_{\mathcal{H}_+^d}(X_t, \nabla f(X_t))=\mirproxproj_{\mathcal{H}_+^d}^{\regfunc,\step,f} (X_t)$, which, by \Cref{prop:mirstep}, implies that the proximal step is a generalization of the mirror step and, thus, can be used to implement a mirror step. Thus, the proximal step procedure of \Cref{alg:prox} can be achieved algorithmically via two proximal maps as: 
\begin{align}
    Y_{t+1} &=\proxproj^{\regfunc,\step}_{\mathcal{H}_+^d}(X_t, \nabla f(X_t)) \label{eqn:prox1}\\
    X_{t+1} &= \proxproj^{\regfunc,\step}_{\mathcal{H}_+^d} (X_t, \nabla f(Y_{t+1})). \label{eqn:prox2}
\end{align}
Note that the proximal step can be alternatively written as a single update:
\begin{align}
    X_{t+1} &= \proxproj^{\regfunc,\step}_{\mathcal{H}_+^d} \bigg(X_t, \nabla f \Big(\proxproj^{\regfunc,\step}_{\mathcal{H}_+^d}(X_t, \nabla f(X_t))\Big)\bigg).
\end{align}

Standard gradient-based optimization methods are generally based on a single step in the direction  of the gradient. In contrast, proximal step methods are based on a two-step approach: 
\begin{enumerate}
    \item From initial point $X_{t}$, apply a ``look-ahead'' step, resulting in a ``prediction'' $Y_{t+1}$. 
    \item Use prediction $Y_{t+1}$ in a ``correction'' step that achieves the update $X_{t+1}$.
\end{enumerate}
As we will discuss in \Cref{sec:algo_mmp}, replacing mirror with proximal steps provides the desired acceleration from $\bigo_d(1/\epsilon^2)$ to $\bigo_d(1/\epsilon)$. Furthermore, the usage of the Bregman divergence enables a geometry-aware projection and, specifically,  can be used to achieve a logarithmic improvement on dependence of the dimensionality of the spectraplex, as will be discussed in \Cref{sec:dim_dep}.

\subsubsection{Regularizers}
In order to guarantee the well-posedness of  the mirror and proximal projections, we will restrict $h$ to the class of functions which are strongly convex.  In the classical games setting, as discussed in the Prior Work (\Cref{sec:prior_work}), the most widely used regularizers are the Euclidean ($\ell_2$) norm ($h(x)=\frac{1}{2}||x||_2^2$) and the negative entropy ($h(x)=\sum_i x_i \log(x_i)$). In this work, we will leverage the quantum analogs of these classical regularizers: the Frobenius norm ($h(X)=\frac{1}{2}||X||_F^2$) and the von Neumann entropy ($h(X)=\tr[X\log X]$). As highlighted in the following lemmas (proved in \Cref{sec:mirror_prox_append}), the mirror and proximal maps take well-known forms when these regularizers are used.

\begin{restatable}[Bregman Divergence of the Frobenius Norm]{lemma}{bregfrob}\label{thm:breg_frob}%
   If $\regfunc$ is the squared Frobenius norm, 
    \begin{equation}
        \reg{X}=\frac{1}{2}||X||_F^2=\frac{1}{2}\tr [X^\dagger X],
    \end{equation}
    then the Bregman divergence is the squared Frobenius distance, 
    \begin{equation} \label{eqn:frob_dist}
        \bregdiv_\regfunc(X||Y)=\tfrac{1}{2}||X-Y||_F^2,
    \end{equation}
     the corresponding mirror map (regularized best-response) is the orthogonal projection,
\begin{align} \label{eqn:orth-project}
    \mirproj^h_\sjoint (Y) = {\arg\min}_{C\in\sjoint} ||Y-C||_F^2=\orthproj_\sjoint (Y),
\end{align}
    and the corresponding proximal map is
\begin{align} \label{eqn:orth-proxy}
    \proxproj^\regfunc_{\sjoint}(X,Y)={\arg\min}_{C\in\sjoint} ||X+Y-C||_F^2= \orthproj_\sjoint (X+Y).
\end{align}
\end{restatable} 

\vspace{-0.1in}
\begin{restatable}[Bregman Divergence of the von Neumann Entropy]{lemma}{bregvn}\label{thm:breg_vn}%
    If $\regfunc$ is the von Neumann entropy,
    \begin{equation}
        \reg{X}=\tr [X \log X],
    \end{equation}
    then the Bregman divergence is the quantum relative entropy
    \begin{equation}
        \bregdiv_\regfunc (X||Y) = \tr [X (\log X - \log Y)],
    \end{equation}   
        the corresponding mirror map (regularized best-response)  is the logit map,
\begin{align} \label{eqn:logit-project}
    \mirproj^h_\sjoint (Y)=\frac{\exp(Y)}{\tr[\exp(Y)]}=\Lambda (Y),
\end{align}
    and the corresponding proximal map is
\begin{align} \label{eqn:logit-proxy}
    \proxproj^\regfunc_{\mathcal{C}}(X,Y)=\frac{\exp(\log X +Y)}{\tr[\exp(\log X +Y)]}=\Lambda (\log X +Y).
\end{align}

\end{restatable}

\noindent In \Cref{sec:algos}, we will leverage these mirror and proximal steps (viewed as oracles) to craft efficient iterative algorithms for finding Nash equilibria of zero-sum quantum games.

\section{Algorithms for Quantum Zero-Sum Games} \label{sec:algos}
As previously discussed, our goal is to efficiently find Nash equilibria, $\joint^*=(\pone^*,\ptwo^*)$, satisfying the variational inequality of \Cref{eqn:vi}. This section describes existing algorithms, as well as our proposed algorithm for finding these Nash equilibria. \Cref{sec:convergence} will demonstrate that our proposed algorithm achieves a quadratic speedup relative to existing algorithms.

\subsection{Matrix Dual Averaging (MDA) Method}
In 2008, Jain and Watrous introduced a parallel algorithm to find approximate Nash-equilibria of quantum zero-sum games \cite{jain2009parallel}. Their algorithm was based on the Matrix Multiplicative Weights Update (MMWU) method \cite{arora2012multiplicative,tsuda2005matrix,kale2007efficient,jain2022matrix}. Here, we will instead present and assess a more general method, which we will refer to as Matrix Dual Averaging (MDA),\footnote{The name is a reference to the analogous algorithm for classical games, i.e., Nesterov's Dual Averaging  \cite{nesterov_primal-dual_2009}. Dual Averaging is sometimes also known as ``Lazy'' Mirror Descent, since the key difference from standard Mirror Descent is that averaging happens in the dual space (rather than mapping from primal to dual space and back in each iteration).} which, with different DGFs, parameterizes a family of algorithms that include Jain and Watrous' MMWU algorithm. 

MDA algorithms are played over several rounds, $t\in[0,N]$, of the quantum zero-sum game. Suppose that, in round $t$, Alice and Bob play states $\joint_t=(\pone_t,\ptwo_t)$.\footnote{Note that, in the first  round ($t=0$), Alice and Bob play the maximally mixed state $\joint_0=(\frac{\Id_n}{2^n},\frac{\Id_m}{2^m})$.} After performing his joint measurement, Roger returns feedback to Alice and Bob in the form of their respective payoff gradients, $\grad_\pone(\joint_t)$ and $\grad_\ptwo(\joint_t)$, as defined in \Cref{eqn:pgrad1} and \Cref{eqn:pgrad2}. Alice and Bob use this feedback to individually update their states so as to iteratively maximize their individual payoff. The general MDA family of algorithms (\Cref{alg:mda}) can be expressed concisely as the updates:
\begin{align}
    \mom_{t+1} &= \mom_t + \grad(\joint_t) \\
    \joint_{t+1} &= \mirproj^h_\sjoint(\step\: \mom_{t+1}),
\end{align}
where $\step$ is a chosen step size and $h$ is a regularizing, convex DGF. As we will  see, it is the choice of DGF $h$ which instantiates different MDA algorithms, including MMWU. Furthermore, relating back to the depiction of the mirror update step in \Cref{fig:mirror_prox_map}, the intermediate states $\mom$ are averaged in the dual space, but are mapped onto feasible states $\joint$ of the primal space by the mirror projection $\mirproj_\sjoint^\regfunc$. Overall, the algorithm only incorporates first-order derivatives in the update steps.

\begin{algorithm}[t] 
\caption{Matrix Dual Averaging}\label{alg:mda}
\begin{algorithmic}
\State \textbf{Accuracy Parameter:} $\epsilon$
\State \textbf{Regularization Function:} $h:\{\mathcal{A},\mathcal{B}\}\to\mathbb{R}$
\State \textbf{Diameter (of $\bregdiv_\regfunc$):} 
$\mathcal{D}_h=\sup_{X,Y\in \sjoint} \bregdiv_\regfunc(X\|Y)$
\State \textbf{Regularized Best Response:} $\mirproj_{\sone,\stwo}^h(\cdot)  = \arg\max_{x\in \{\sone,\stwo\}}\{\langle\cdot,x\rangle-h(x)\}$
\vspace{0.1in}
\State $\step \gets  \ {\mu_h}/(2\gamma_\grad)$ 
\Comment{Step Size}
\State $N \gets \lceil \mathcal{D}_h/(\step\cdot\epsilon^2) \rceil$ 
\Comment{Number of Rounds}
\State $(\pone_0,\ptwo_0) \gets \left(\frac{1}{2^n}\Id_\sone, \frac{1}{2^m}\Id_\stwo\right)$ \Comment{State Initialization}

\vspace{0.1in}
\For{$t\in[1,N-1]$}

\vspace{0.1in}
\State $\cfone^{(t)} \gets \sum_{i=0}^{t-1} \grad_\pone(\joint_i)$ \Comment{Calculate Cumulative Feedback}
\State $\cftwo^{(t)} \gets \sum_{i=0}^{t-1} \grad_\ptwo(\joint_i)$ 

\vspace{0.1in}
\State $\pone_t \gets \mirproj_\sone^h\left(\step\:\cfone^{(t)}\right)$ \Comment{State Updates}
\State $\ptwo_t \gets \mirproj_\stwo^h\left(\step\:\cftwo^{(t)}\right)$

\vspace{0.1in}
\EndFor
\end{algorithmic}
\vspace{0.1in}
\Return $\bigg( \bar{\pone}=\frac{1}{N}\sum_{t=0}^{N-1}\pone_t, \bar{\ptwo}=\frac{1}{N}\sum_{t=0}^{N-1}\ptwo_t \bigg)$
\end{algorithm}

We will now specify how each player chooses their next state via the previously described MDA method. Note that both Alice and Bob perform the same procedure, so we will only specify the procedure from Alice's perspective. 
Upon receiving feedback $\grad_\pone(\joint_t)$ from Roger, Alice calculates her cumulative feedback,
\begin{equation}
    \cfone^{(t)} = \sum_{i=0}^{t-1} \grad_\pone (\joint_i),
\end{equation}
which is simply the sum of the payoff gradients across all prior rounds. To calculate the next state she will play, $\pone_{t+1}$, Alice re-weights the cumulative feedback by the step size $\step$ and performs a mirror projection onto her spectraplex,
\begin{equation}
    \pone_{t+1}=\mirproj^h_{\sone}(\step \: \cfone^{(t)})=\arg\max_{\pone \in \sone}\{\langle \step \: \cfone^{(t)},\pone\rangle-h(\pone)\}.
\end{equation}
We will now specify how different choices of $h$ lead to different instantiations of the MDA algorithm. For example, following from \Cref{thm:breg_frob}, if the regularizer is the squared Frobenius norm, 
\begin{equation}
    h(\pone) = \frac{1}{2}\|\pone\|_F^2,
\end{equation}
then the mirror projection maps Alice's next state onto the orthogonal projection
\begin{equation}
    \pone_{t+1}= \orthproj_\sone (\step\:\cfone^{(t)}).
\end{equation}
Meanwhile, following from \Cref{thm:breg_vn}, if the regularizer is the negative von Neumann entropy,
\begin{equation}
    h(\pone) = \tr[\pone \log \pone],
\end{equation}
then the mirror projection maps Alice's next state onto the logit map
\begin{equation}
    \pone_{t+1}=\Lambda_\sone \left(\step\:\cfone^{(t)}\right) =\frac{\text{exp}\left(\step\:\cfone^{(t)}\right)}{\tr\left(\text{exp}\left(\step\:\cfone^{(t)}\right)\right)}
\end{equation}
which yields the MMWU algorithm of \cite{jain2009parallel}. The following convergence rate for this procedure was established by \cite{jain2009parallel}. 
\begin{restatable}[MDA Rate]{theorem}{mdarate}\label{thm:mdarate}%
    MDA methods compute an $\epsilon$-Nash equilibrium in $\bigo_d({1}/{\epsilon^2})$ steps.
\end{restatable}
\noindent Thus, for quantum zero-sum games MMWU obtains a worse convergence rate than the best known algorithms for classical zero-sum games, which obtain an $\bigo_d(1/\epsilon)$ rate, as described in \Cref{sec:prior_work} and depicted in \Cref{fig:classical_methods}. In the remainder of this work, we will show how the classically efficient Mirror Prox algorithm and its variants can be lifted to the quantum zero-sum game setting in order to obtain the desired $\bigo_d(1/\epsilon)$ rate.

\subsection{Matrix Mirror Prox (MMP)}
{\color{black}The classical Mirror Prox (MP) algorithm is described in \Cref{sec:prior_work}, with updates given in \Cref{fig:classical_methods}. In each step of the iterative MP optimization procedure, the gradient of the players' utilities is calculated followed by the execution of a mirror step (\Cref{alg:mirror}), as illustrated in \Cref{fig:mirror_prox_map}. The mirror step can be implemented via a mirror map of the dual gradient, as in \Cref{eqn:mirror_proj_update}, or the mirror map of the primal state, as in \Cref{eqn:mirror_prox_proj_update}. 
The addition of proximal steps—or extra-gradient steps—is crucial as they introduce a refinement to the basic gradient step. After the initial gradient is computed and the first mirror step is taken, the algorithm does not immediately proceed with this information. Instead, it calculates a second gradient at the point reached after the first mirror step. This second gradient offers a more accurate direction for the update because it accounts for the immediate effects of the initial step, thus incorporating more information about the curvature and constraints of the utility space. This two-pronged gradient approach is what distinguishes the MP algorithm from simpler gradient descent methods, offering a sophisticated tool for navigating the intricate landscapes of variational inequalities.
}

Drawing inspiration from classical MP, we propose an analog for the quantum zero-sum games of interest. This quantum analog, which we will refer to as the Matrix Mirror Prox (MMP) method, has updates:
\begin{align}
    \mom_{t+1} &= \proxproj^{h,\eta}_\sjoint ( \joint_{t},\grad( \joint_{t}))  \\
    \joint_{t+1} &= \proxproj^{h,\eta}_\sjoint ( \joint_{t}, \grad(\mom_{t+1})).
\end{align}
Notice that these update steps are of the same form as the proximal step updates given in \Cref{eqn:prox1} and \Cref{eqn:prox2}. Alternatively, the updates can be expressed purely in terms of the joint states $\joint_t$:
\begin{align}
    \joint_{t+1} &= \proxproj^{h,\eta}_\sjoint \bigg( \joint_{t}, \grad \Big(\proxproj^{h,\eta}_\sjoint ( \joint_{t},\grad( \joint_{t}))\Big)\bigg).
\end{align}
As discussed in \Cref{sec:mirror_prox_step}, the incorporation of proximal steps in the Matrix Mirror Prox algorithm offers two significant advantages over the conventional Projected Subgradient Descent Ascent method:
\begin{enumerate}
    \item \textbf{\textsf{Geometry-Awareness}}: MMP replaces the Frobenius norm of PGDA with a generalized Bregman divergence. Similar to MDA, utilizing a mirror map with this Bregman divergence enables MMP to adapt to the inherent geometry of the problem space. This makes the algorithm highly versatile, capable of effectively handling a wide array of problems, including those with intricate constraints.
    \item \textbf{\textsf{Intermediate Proximal Step}}: {Proximal steps in MMP introduce an intermediate ``prediction'' state (i.e., the state $Y_t$ in \Cref{sec:proxy}), which plays a crucial role. At a high level, this prediction state extrapolates the future trajectory of the descent/ascent, thereby guiding the subsequent state update. This leverages the intuition that the gradient at the next state offers a more informative update gradient than that at the current state.}
\end{enumerate}
Importantly, this combination of geometry-aware mirror maps and forward-looking intermediate proximal steps accelerates the convergence of MMP with von Neumann regularization to min-max points, thereby making MMP a potent tool for solving min-max problems efficiently.

As is the case for classical MP (discussed in \Cref{sec:prior_work}), a key downside of MMP is that it requires two gradient calls per iteration. Specifically, in each round $t$, to calculate the updated joint state $\joint_{t+1}$, Alice and Bob must query Roger with two distinct sets of states, $\joint_t$ and $\mom_{t+1}$, so as to obtain the respective feedback gradients $\grad(\joint_t)$ and $\grad(\mom_{t+1})$. Classically, there exist variants of MP with the same convergence rate that require only a single gradient call per iteration. The remainder of this work will focus on a proposed classically-inspired variant of MMP which requires only a single gradient call per iteration, and a proof that it obtains the desired $\bigo_d(1/\epsilon)$ convergence rate.

\subsection{Optimistic Matrix Mirror-Prox (OMMP) Methods} \label{sec:algo_mmp}

\begin{algorithm}[t] 
\caption{Optimistic Matrix Mirror Prox (OMMP)}\label{alg:SCMMP}
\begin{algorithmic}
\State \textbf{Accuracy Parameter:} $\epsilon$
\State \textbf{Regularization Function:} $h:\{\mathcal{A},\mathcal{B}\}\to\mathbb{R}$
\State \textbf{Strong Convexity Parameter (of $h$):} $\mu_h$
\State \textbf{Diameter (of $\bregdiv_\regfunc$):} 
$\mathcal{D}_h=\sup_{X,Y\in \sjoint} \bregdiv_\regfunc(X\|Y)$
\State \textbf{Lipschitz Parameter (of $\grad$):} $\gamma_\grad$
\State \textbf{Proximal Map:} $
    \proxproj^{\regfunc,\step}_{\sone,\stwo}(X,Y):=
    {\arg\min}_{C\in\sjoint}\{ \langle Y,C-X\rangle -\frac{1}{\step}\bregdiv_\regfunc (C||X)\}.$
\vspace{0.1in}
\State $\step \gets  \ {\mu_h}/(2\gamma_\grad)$ 
\Comment{Step Size}
\State $N \gets \lceil \mathcal{D}_h/(\step\cdot\epsilon) \rceil$ 
\Comment{Number of Rounds}
\State $(\pone_0,\ptwo_0) \gets \left(\frac{1}{2^n}\Id_\sone, \frac{1}{2^m}\Id_\stwo\right)$ \Comment{State Initialization}
\State $(\hat{\pone}_0,\hat{\ptwo}_0) \gets (\pone_0,\ptwo_0)$ \Comment{Optimistic-Momentum Initialization}

\vspace{0.1in}
\For{$t\in[1,N-1]$}

\vspace{0.1in}
\State $\pone_{t+1}\gets\proxproj_\sone^{h,\eta}(\hat{\pone}_t, \grad_\pone(\joint_t))$ \Comment{State Updates}
\vspace{0.05in}
\State $\ptwo_{t+1}\gets\proxproj_\stwo^{h,\eta}( \hat{\ptwo}_t, \grad_\ptwo(\joint_t))$

\vspace{0.1in}
\State $\hat{\pone}_{t+1}\gets\proxproj_\sone^{h,\eta}(\hat{\pone}_t, \grad_\pone(\joint_{t+1}))$ \Comment{Momentum Updates}
\vspace{0.05in}
\State $\hat{\ptwo}_{t+1}\gets\proxproj_\stwo^{h,\eta}(\hat{\ptwo}_t, \grad_\ptwo(\joint_{t+1}))$

\vspace{0.1in}
\EndFor
\end{algorithmic}
\vspace{0.1in}
\Return $\bigg( \bar{\pone}=\frac{1}{N}\sum_{t=0}^{N-1}\pone_t, \bar{\ptwo}=\frac{1}{N}\sum_{t=0}^{N-1}\ptwo_t \bigg)$
\end{algorithm} 

We will now discuss the main contribution of this work, a novel approach to finding Nash equilibria of quantum zero-sum games, which we refer to as the Optimistic Matrix Mirror Prox (OMMP) method. This algorithm draws inspiration from seminal classical game theory work on Optimistic Mirror Prox (OMP)\footnote{There are a variety of classical algorithms of this general form, known as Popov's modified Arrow-Hurwicz algorithm, Past Extra-Gradient, Optimistic Extra-Gradient, or the Optimistic Gradient Variant.} methods, as described in the Prior Work (\Cref{sec:prior_work}). 
The OMMP algorithm, detailed in \Cref{alg:SCMMP}, can be represented succinctly through the following update equations:
\begin{align}
    \joint_{t+1} &= \proxproj^{h,\eta}_\sjoint ( \mom_{t}, \grad(\joint_{t})) \label{eqn:state_peg} \\
    \mom_{t+1} &= \proxproj^{h,\eta}_\sjoint ( \mom_{t},\grad( \joint_{t+1})) \label{eqn:mom_peg}
\end{align}
Here, $\step$ is a predetermined step size, and $\proxproj$ refers to the proximal map, as elaborated in \Cref{alg:prox} and \Cref{prop:proximalproj}.

Although the update equations for OMMP closely resemble those of the standard MMP algorithm, a critical difference exists. As alluded to in the previous section, the OMMP variant necessitates only a single gradient call per iteration. While it may appear that both Alice and Bob require two feedback gradients, $\grad(\joint_t)$ and  $\grad(\joint_{t+1})$, for each round $t$, they will have already acquired $\grad(\joint_t)$ during round $t-1$. This allows them to reuse $\grad(\joint_t)$ and only make a query for a new gradient with one state, $\joint_{t+1}$.

To gain some intuition for this algorithm, we will revisit the original proposal of ``optimism'' in this context, due to Popov \cite{popov1980modification}. In conventional Gradient Descent Ascent, both players update their strategies concurrently, moving in directions opposite to their individual gradients—akin to a tug-of-war. Optimistic Gradient Descent Ascent (OGDA) enhances this by incorporating a ``look-ahead'' feature. Players consider not only the current gradient but also anticipate the opponent's forthcoming strategic adjustments. This anticipatory or ``optimistic'' approach is formalized by adding an extra term to the update rule, essentially an estimate of the future gradient.

This anticipatory element is also how OMMP diverges from standard MMP. Rather than employing a new prediction, optimistic OMMP utilizes a past ``inertia'' point, $\mom_t$. Unlike MMP, which focuses solely on future states, OMMP integrates past gradient information to stabilize its learning trajectory.

%% file: convergence.tex
\section{Convergence Analysis} \label{sec:convergence}

We will now prove one of the main results of the paper:
\mainres*
\noindent This establishes that our proposed OMMP method yields a quadratic speedup compared to previous MDA-based approaches. Concretely, we will leverage the proof structure of Ene and Nguyên \cite{ene2022adaptive} to demonstrate that, with an appropriately chosen step size $\step$ and for a given accuracy parameter $\epsilon$, the \emph{average iterate}
\begin{align}
    \overline{\joint}_T = \frac{1}{T} \sum_{t=1}^T \joint_t
\end{align}
converges to an $\epsilon$-Nash equilibrium in  $O_d(1/\epsilon)$ iterations.\footnote{
It's important to note the difference in the notation used to describe algorithmic convergence rates between the optimization and computer science communities. In the realm of optimization, \emph{convergence rates} such as \(\bigo(e^{-\rho T})\), \(\bigo(1/T)\), and \(\bigo(1/\sqrt{T})\) are commonly employed to indicate how swiftly an algorithm approximates the optimal solution as the iteration count \(T\) increases. In contrast, the computer science literature tends to recast these convergence rates as \textit{number of iterations} required to achieve an \(\epsilon\)-accurate solution, thereby giving a more direct account of computational efficiency. Expressed in terms of accuracy parameter \(\epsilon\), this leads to expressions like \(\bigo(\log (1/\epsilon))\), \(\bigo(1/\epsilon)\), and \(\bigo(1/\epsilon^2)\). 
}

In \Cref{sec:dim_dep}, we will treat the dimension-dependence hidden in the $\bigo_d(\cdot)$ notation, showing in particular that the von Neumann entropy regularizer reduces OMMP's dimension-dependence logarithmically relative to the Frobenius norm regularizer. Since OMMP with the von Neumann entropy regularizer instatiates the OMMWU algorithm, we thus prove that OMMWU obtains the desired $\bigo(d/\epsilon)$ iteration complexity and, thus, achieves a quadratic speedup to MMWU's $\bigo(d/\epsilon^2)$ iteration complexity.

\subsection{Monotonicity of \texorpdfstring{$\grad$}{Lg}}
Critical to our analysis is the monotonicity of the feedback operator $\grad$, which is a key property of the finite-valued quantum zero sum games of interest.
\begin{restatable}[Mononicity of $\grad$]{lemma}{monotone}\label{lemma:monotone}%
   $\grad$ is monotone or, equivalently,
    \begin{align} \label{eqn:monotone}
        \inner{(\grad(X)-\grad(Y))}{(X-Y)} \geq 0, \:\: \forall X,Y\in\sjoint.
    \end{align}
\end{restatable}
\noindent The proof of this result employs techniques involving the Pauli decomposition of the feedback operator, thereby deviating from the conventional proof for monotonicity of classical normal-form game operators. We defer the proof of $\grad$'s monotonicity (in addition to many other properties of $\grad$) to \Cref{sec:f_props}.

\subsection{Error Decomposition} 

Following from our discussion of Nash equilibria in \Cref{sec:nash_eq} and the natural error function, given in \Cref{eqn:error_func}, the error of the average iterate is defined as
\begin{align}
    \textnormal{Error} (\overline{\joint}_T)&=\sup_{Z\in\sjoint} \inner{\grad(Z)}{(\overline{\joint}_T-Z)}= \frac{1}{T} \sup_{Z\in\sjoint} \: \sum_{t=1}^T \inner{\grad (Z)}{(\joint_t - Z)}.
\end{align}
From the monotonicity of $\grad$, given in \Cref{eqn:monotone}, for $\joint_t$ and some state $Z\in\sjoint$,
\begin{align}
    \inner{ \grad(Z)}{(\joint_t - Z)} \leq \inner{\grad(\joint_t)}{(\joint_t - Z)},
\end{align}
which implies that the average-iterate error is upper bounded as
\begin{align}
    \textnormal{Error} (\overline{\joint}_T)\leq \frac{1}{T} \sup_{Z\in\sjoint}  \: \sum_{t=1}^T \inner{\grad (\joint_t)}{(\joint_t - Z)}. \label{eqn:error}
\end{align}
In order to produce an explicit bound, we begin by decomposing the terms of the summation, as follows, resulting in a sum of three distinct terms: $(A_t)$, $(B_t)$, and $(C_t)$.
\begin{align}
    &\inner{\grad (\joint_t)}{(\joint_t - Z)}= \inner{\grad (\joint_t)}{(\joint_t - \mom_t + \mom_t - Z)} \nonumber\\
    &=\inner{\grad (\joint_t)}{(\joint_t - \mom_t) + \grad (\joint_t) (\mom_t - Z)}=\inner{\grad (\joint_t)}{(\joint_t - \mom_t)} + \inner{(\joint_t)}{(\mom_t - Z)}  \nonumber\\
    &=\inner{(\grad (\joint_t)-\grad (\joint_{t-1})+\grad (\joint_{t-1}))}{(\joint_t - \mom_t)} + \inner{\grad (\joint_t)}{(\mom_t - Z)}\nonumber\\
    &=\inner{(\grad (\joint_t)-\grad (\joint_{t-1}))} {(\joint_t - \mom_t)+\grad (\joint_{t-1})(\joint_t - \mom_t)} + \inner{\grad (\joint_t)}{(\mom_t - Z)}\nonumber\\
    &= \underbrace{\inner{\grad(\joint_t)}{(\mom_t-Z)}}_{(A_t)}  + \underbrace{\inner{(\grad(\joint_t)-\grad(\joint_{t-1}))}{(\joint_t - \mom_t)}}_{(B_t)} + \underbrace{\inner{\grad(\joint_{t-1})}{(\joint_t - \mom_t)}}_{(C_t)}. \label{eqn:decomp}
\end{align}
In the following analysis, we will show that terms $(A_t)$ and $(C_t)$ can be upper bounded via the optimality of the proximal updates in the OMMP algorithm.  Furthermore, the term $(B_t)$ can be upper bounded via the strong convexity, smoothness, and duality of the OMMP proximal updates.

\subsection{Optimality Analysis of \texorpdfstring{$(A_t)$}{Lg}}
We begin by upper bounding the term $(A_t)$. Recall from the definition of the OMMP algorithm in \Cref{sec:algo_mmp} that $\mom_t$ is the outcome of the proximal projection of \Cref{eqn:mom_peg}, for a regularization function $\regfunc$. Therefore, $\mom_t$ can be expressed as
\begin{align} \label{eqn:mom_defn}
    \mom_t = \arg\min_{Z\in\sjoint} \mopt_t (Z),
\end{align} 
where $\mopt_t(Z)$ is the proximal optimization term
\begin{align} \label{eqn:mom_min}
    \mopt_t(Z)=
     \inner{\grad(\joint_{t})}{(Z-\mom_{t-1})} +\frac{1}{\step}\bregdiv_\regfunc (Z\|\mom_{t-1}).
\end{align} 
Note that $\mopt_t$ is differentiable, with gradient 
\begin{align} \label{eqn:mom_der}
    \nabla_{Z^\top}\mopt_t(Z)&=\grad(\joint_{t}) + \frac{1}{\step}
    \nabla_{Z^\top}\bregdiv_\regfunc (Z\|\mom_{t-1}).
\end{align}
Using the definitions of \Cref{sec:conv_dual}, in the supplement (\Cref{app:proofs}) we further show that:
\vspace{-0.1in}
\begin{restatable}{proposition}{momconv}\label{prop:mom_conv}%
   For a $\mu$-strongly convex regularization function $\regfunc$ with respect to the norm $\|\cdot\|$, $\mopt_t$ is $\frac{\mu}{\step}$-strongly convex with respect to the norm $\|\cdot\|$.
\end{restatable}
\noindent Since $\mopt_t$ is strongly convex and differentiable it satisfies the First-Order Optimality Condition (\Cref{thm:firstopt}) for a feasible momentum solution $\mom_t$,
\begin{align} 
    \inner{\nabla_{\mom_t^\top} \mopt_t (\mom_t)}{(Z-\mom_t)}\geq 0, \:\: \forall Z \in \sjoint.
\end{align}
Plugging the gradient of \Cref{eqn:mom_der} into this expression gives
\begin{align} \label{eqn:opt_two} 
    &\inner{\Big(\grad(\joint_{t})  + \frac{1}{\step}\nabla_{\mom_t^\top}\bregdiv_\regfunc (\mom_t||\mom_{t-1})\Big)}{(\mom_{t}-Z)}\leq 0 \nonumber\\ 
    &\inner{\grad(\joint_{t})}{(\mom_{t}-Z)}  + \frac{1}{\step}\inner{\nabla_{\mom_t^\top}\bregdiv_\regfunc (\mom_t||\mom_{t-1})}{(\mom_{t}-Z)}\leq 0 \nonumber\\ 
    &\inner{\grad(\joint_{t})}{(\mom_{t}-Z)}  \leq  \frac{1}{\step}\inner{-\nabla_{\mom_t^\top}\bregdiv_\regfunc (\mom_t||\mom_{t-1})}{(\mom_{t}-Z)}. 
\end{align}
From our definition of $(A_t)$ in \Cref{eqn:decomp} and the three-point identity for the Bregman divergence given in \Cref{eq:3point}, it follows that, $\forall Z \in \sjoint$,
\begin{align} \label{eqn:a_t}
     (A_t) &\leq \frac{1}{\step} \Big( \bregdiv_\regfunc (Z||\mom_{t-1})-\bregdiv_\regfunc (\mom_t||\mom_{t-1})-\bregdiv_\regfunc (Z||\mom_t) \Big).
\end{align}

\subsection{Optimality Analysis of \texorpdfstring{$(C_t)$}{Lg}}
We will now upper bound $(C_t)$, similarly to how we upper bounded $(A_t)$. Recall from the definition of the OMMP algorithm in \Cref{sec:algo_mmp} that $\joint_t$ is the outcome of the proximal projection of \Cref{eqn:state_peg}, for a regularization function $\regfunc$. Therefore, $\joint_t$ can be expressed as
\begin{align} \label{eqn:joint_min}
    \joint_t = \arg\min_{Z\in\sjoint} \jopt_t (Z),
\end{align} 
where $\jopt_t(Z)$ is the proximal optimization term
\begin{align} 
    \jopt_t(Z)=
         \inner{\grad(\joint_{t-1})}{(Z-\mom_{t-1})} +\frac{1}{\step}\bregdiv_\regfunc (Z\|\mom_{t-1}).
\end{align} 
Note that $\mopt_t$ is differentiable, with gradient 
\begin{align} \label{eqn:joint_der}
    \nabla_{Z^\top}\jopt_t(Z)&=\grad(\joint_{t-1}) + \frac{1}{\step}
        \nabla_{Z^\top}\bregdiv_\regfunc (Z\|\mom_{t-1}).
\end{align}
Similarly as in the previous section, in the supplement (\Cref{app:proofs}) we further show that:
\vspace{-0.15in}
\begin{restatable}{proposition}{jointconv}\label{prop:joint_conv}%
   For a $\mu$-strongly convex regularization function $\regfunc$ with respect to the norm $\|\cdot\|$, $\jopt_t$ is $\frac{\mu}{\step}$-strongly convex with respect to the norm $\|\cdot\|$.
\end{restatable}
\noindent Since $\jopt_t$ is strongly convex and differentiable, it satifies the First-Order Optimality Condition (\Cref{thm:firstopt}) for a feasible state solution $\joint_t$,
\begin{align} 
    \inner{\nabla_{\joint_t^\top} \jopt_t (\joint_t)}{(Z-\joint_t)}\geq 0, \:\: \forall Z \in \sjoint.
\end{align}
Plugging the gradient of \Cref{eqn:joint_der} into this expression gives, $\forall Z \in \sjoint$,
\begin{align}
    &\inner{ \Big(\grad(\joint_{t-1}) + 
    \frac{1}{\step}
    \nabla_{\joint_{t}^\top}\bregdiv_\regfunc (\joint_{t}||\mom_{t-1}) 
    \Big)}{(\joint_t-Z)}\leq 0 \\
    &\inner{ \grad(\joint_{t-1})}{ (\joint_t-Z)} \leq \frac{1}{\step}\inner{
    \nabla_{\joint_{t}^\top}\bregdiv_\regfunc (\joint_{t}||\mom_{t-1}) 
    }{(Z-\joint_t)}.
\end{align}
Since the expression holds for all $Z\in\sjoint$, we can choose to set $Z=\mom_t$, which implies
\begin{align}
    \inner{ \grad(\joint_{t-1})}{ (\joint_t-\mom_t)} \leq \frac{1}{\step}\inner{
    \nabla_{\joint_{t}^\top}\bregdiv_\regfunc (\joint_{t}||\mom_{t-1}) 
    }{(\mom_t-\joint_t)}.
\end{align}
From our definition of $(C_t)$ in \Cref{eqn:decomp} and the Three-Point Tdentity for the Bregman divergence, given in \Cref{eq:3point}, it follows that            \begin{align} \label{eqn:c_t}
      (C_t) \leq \frac{1}{\step} \Big( \bregdiv_\regfunc (\mom_t||\mom_{t-1})-\bregdiv_\regfunc (\joint_t||\mom_{t-1})-\bregdiv_\regfunc (\mom_t||\joint_t) \Big).
\end{align}

\subsection{An Upper Bound for \texorpdfstring{$(B_t)$}{Lg}}

By the Cauchy-Schwarz inequality for Frobenius norm, we can upper bound $(B_t)$ as
\begin{align} \label{eqn:b_prelim}
    (B_t) &= \inner{(F(\joint_t)-F(\joint_{t-1}))}{(\joint_t - \mom_t)}\leq \|F(\joint_t)-F(\joint_{t-1})\|_F  \|\joint_t - \mom_t\|_F.
\end{align}
or more generally, with respect to the norms $(\|\cdot\|, \|\cdot\|_*)$
\begin{align} \label{eqn:b_prelim_general}
    (B_t) &= \inner{(F(\joint_t)-F(\joint_{t-1}))}{(\joint_t - \mom_t)}\leq \|F(\joint_t)-F(\joint_{t-1})\|_*  \|\joint_t - \mom_t\|.
\end{align}
To further manipulate this 
upper bound for $(B_t)$, we leverage the ideas from Fenchel duality and convexity introduced in \Cref{sec:conv_dual}.

Following from \Cref{eqn:joint_min} and \Cref{eqn:mom_min}, respectively, $\joint_t$ is  the minimizer of $\jopt_t$ and $\mom_t$ is  the minimizer of $\mopt_t$. Re-expressing $\mopt_t$ in terms of $\jopt_t$, 
\begin{align} 
    \mopt_t  (U) = \jopt_t  (U)+\inner{\left(\grad(\joint_t)-\grad(\joint_{t-1})\right)}{(U-\mom_{t-1})}
\end{align}
implies that $\mom_t$ can be expressed in terms of a minimization of $\jopt_t$,
\begin{align} \label{eqn:min_momt}
    \mom_t = \arg\min_{U\in\sjoint}\left\{\jopt_t  (U)+\inner{\left(\grad(\joint_t)-\grad(\joint_{t-1})\right)}{U}\right\}.
\end{align}
Since we assumed that $\regfunc$ is $\mu$-strongly convex, by \Cref{prop:joint_conv}  $\jopt_t$ must be $\tfrac{\mu}{\step}$-strongly convex. Denoting the Fenchel conjugate of $\jopt_t$ as $\jopt_t^*$,
since $\jopt_t$  is strongly convex, Danskin's Theorem (\Cref{lemma:danskin}) states that $\nabla \jopt_t^*$ is well-defined, for all $V \in \sjoint$, as
\begin{align}
    \nabla \jopt^*_t (V) = \arg\min_{U\in\sjoint} \left\{ \jopt_t(U)-\inner{V}{U}\right\}.\label{eqn:min_psi_t}
\end{align}
In the case that $V=0$, \Cref{eqn:mom_defn} implies that
\begin{align}
    \nabla \jopt^*_t (0)=\arg\min_{U\in\sjoint} \jopt_t(U)=\joint_t.
\end{align}
Meanwhile, if $V=\grad (\joint_{t-1})-\grad (\joint_t)$, \Cref{eqn:min_momt} implies that
\begin{align}
    \nabla \jopt^*_t \left(\grad (\joint_{t-1})-\grad (\joint_t)\right)=\arg\min_{U\in\sjoint} \left\{ \jopt_t(U)+\inner{\left(\grad (\joint_t)-\grad (\joint_{t-1})\right)}{U}\right\}=\mom_t.
\end{align}
Therefore, 
\begin{align} \label{eqn:prelim_eq}
    \|\joint_t-\mom_t\|_F &= \| \nabla \jopt^*_t (0)-\nabla \jopt^*_t \left(\grad (\joint_{t-1})-\grad (\joint_t)\right)\|_F.
\end{align}
Furthermore, since $\jopt_t$ is $\frac{\mu}{\step}$-strongly convex with respect to the norm $\|\cdot\|$, Strong-Smooth Duality (\Cref{lemma:convex_smooth}) implies that $\jopt^*_t$ is $\frac{\step}{\mu}$-smooth with respect to the norm $\|\cdot\|_{*}$. By \Cref{def:smoothness-grad-definition}, this implies that $\nabla \jopt^*_t$ is $\frac{\step}{\mu}$-Lipschitz continuous with respect to $\|\cdot\|_{*}$,
\begin{align}
    \| \nabla \jopt^*_t (X)-\nabla \jopt^*_t \left(Y\right)\| =\| \nabla \jopt^*_t (X)-\nabla \jopt^*_t \left(Y\right)\|_{**}
    &\leq \frac{\step}{\mu} \|X-Y\|_{*}.
\end{align}
Plugging this and \Cref{eqn:prelim_eq} into \Cref{eqn:b_prelim} produces the final upper bound for $(B_t)$,
\begin{equation}
    (B_t) \leq \frac{\step}{\mu} \|\grad(\joint_t) - \grad(\joint_{t-1})\|_{*}^2.\label{eqn:b_t}
\end{equation}

\subsection{Average-Iterate Error Upper Bound}

Combining the average-iterate error upper bound of \Cref{eqn:error} and its decomposition, as expressed in \Cref{eqn:decomp}, we have that
\begin{align*}
    \textnormal{Error} (\overline{\joint}_T) & \leq \frac{1}{T} \sup_{Z\in\sjoint}  \sum_{t=1}^T \inner{ \grad (\joint_t)}{(\joint_t - Z)} =\frac{1}{T} \sup_{Z\in\sjoint}  \sum_{t=1}^T (A_t)+(B_t)+(C_t).
\end{align*}
Plugging in the upper bounds for $(A_t)$ from \Cref{eqn:a_t}, $(B_t)$ from \Cref{eqn:b_t}, and $(C_t)$ from \Cref{eqn:c_t} gives 
\begin{align*}
    \textnormal{Error} (\overline{\joint}_T) & \leq \frac{1}{T} \sup_{Z\in\sjoint}  \sum_{t=1}^T \bigg(\frac{1}{\step}\Big( \bregdiv_\regfunc (Z\|\mom_{t-1})-\left(\bregdiv_\regfunc (Z\|\mom_t) +\bregdiv_\regfunc (\joint_t\|\mom_{t-1})+\bregdiv_\regfunc (\mom_t\|\joint_t)\right) \Big) \\
    &\hspace{1.4in}+\frac{\step}{\mu} \|\grad (\joint_t)^\top-\grad(\joint_{t-1})^\top\|_*^2\bigg) \\
    &\leq \frac{1}{T} \sup_{Z\in\sjoint}  \underbrace{\frac{1}{\eta}\sum_{t=1}^T \big( \bregdiv_\regfunc (Z\|\mom_{t-1})- \bregdiv_\regfunc (Z\|\mom_{t})\big)}_{(I)} -  \frac{1}{\eta}\sum_{t=1}^T  \big(\bregdiv_\regfunc (\joint_t\|\mom_{t-1})+\bregdiv_\regfunc (\mom_t\|\joint_t) \big) \\
    &\hspace{1.4in}+\frac{\step}{\mu} \sum_{t=1}^T\|\grad (\joint_t)^\top-\grad(\joint_{t-1})^\top\|_*^2.
\end{align*}
Via telescoping sums, the first term $(I)$ can be reduced to
\begin{align} 
    (I)=\frac{1}{\eta}\sum_{t=1}^T \big( \bregdiv_\regfunc (Z\|\mom_{t-1})- \bregdiv_\regfunc (Z\|\mom_{t})\big) = \frac{1}{\eta}\Big( \bregdiv_\regfunc (Z\|\mom_{0})- \bregdiv_\regfunc (Z\|\mom_{T})\Big).
\end{align}
Additionally, since $\regfunc$ is $\mu$-strongly convex  with respect to $\|\cdot\|$,
\begin{align}
    \bregdiv_\regfunc (X\|Y)\geq \frac{\mu}{2}\|X-Y\|^2.
\end{align}
Therefore, we can re-express the error as
\begin{align*}
    \textnormal{Error} (\overline{\joint}_T) & \leq \frac{1}{T} \sup_{Z\in\sjoint} 
\frac{1}{\eta} \bregdiv_\regfunc (Z\|\mom_{0})
    \underbrace{- \frac{\mu}{2\step}\|Z-\mom_{T}\|^2+ -  \frac{\mu}{2\eta}\sum_{t=1}^T  \big(\|\mom_{t-1}-\joint_t\|^2+\|\mom_{t}-\joint_t\|^2 \big)}_{(II)} \\
    &\hspace{1.4in}+\underbrace{\frac{\step}{\mu} \sum_{t=1}^T\|\grad (\joint_t)^\top-\grad(\joint_{t-1})^\top\|_*^2}_{(III)},
\end{align*}
and proceed to upper bound terms $(II)$ and $(III)$.
We begin by expanding $(II)$ as
\begin{align*}
    (II) &=- \frac{\mu}{2\step}\|\mom_{T}-Z\|^2-\frac{\mu}{2\step}\sum_{t=1}^T  \big(\|\mom_{t-1}-\joint_t\|^2+\|\mom_{t}-\joint_t\|^2 \big) \nonumber \\
    &=\frac{\mu}{2\step} \bigg(\|\mom_{0}-\joint_0\|^2-\|\mom_{T}-\joint_T\|^2 - \|\mom_{T}-Z\|^2-\sum_{t=1}^T  \big(\|\mom_{t-1}-\joint_t\|^2+\|\mom_{t-1}-\joint_{t-1}\|^2 \big)\bigg)
\end{align*}
Since we initialized $\Phi_0=\Psi_0$, 
\begin{align}
    \|\mom_{0}-\joint_0\|^2 = 0.
\end{align}
Furthermore, since $\|\mom_{T}-Z\|^2 \geq 0$ for all $Z\in\sjoint$, including $Z=\joint_T$, 
\begin{align}
    (II) & \leq -\frac{\mu}{2\step} \sum_{t=1}^T  \big(\|\mom_{t-1}-\joint_t\|^2+\|\mom_{t-1}-\joint_{t-1}\|^2 \big).
\end{align}
Defining 
\begin{align}
    \Gamma:=\sum_{t=1}^T \left(  \|\joint_t -\mom_{t-1}\|^2 + \|\joint_{t-1} -\mom_{t-1}\|^2\right)
\end{align}
implies the following upper bound on $(II)$
\begin{align} \label{eqn:two}
    (II) & \leq -\frac{\mu}{2\step} \Gamma.
\end{align}

In order to bound term $(III)$, we will leverage an important property of finite-valued quantum zero-sum games. Namely, in \Cref{app:proofs}, we prove several properties of the $\grad$ operator, including the following lemma regarding Lipschitz continuity.
\begin{restatable}[Lipschitz Continuity of $\grad$]{lemma}{lipschitz}\label{lemma:lipschitz}%
{\color{black}
    For finite-valued quantum zero-sum games, $\grad$ is a Lipschitz-continuous operator, meaning there exists a Lipschitz constant $\cont\in\mathbb{R}$ such that
    \begin{align} \label{eqn:lipschitz}
        \|\grad(Z)-\grad(Z')\|_* \leq \cont \| Z-Z' \|, \:\: \forall \: Z, Z' \in \sjoint.
    \end{align}
    For the case of $(\|\cdot\|_F,\|\cdot\|_F)$ and $(\|\cdot\|_\infty,\|\cdot\|_1)$, $\cont_{F,F}=\mathcal{O}(2^d)$ while $\cont_{\infty,1}=\mathcal{O}(1)$.
}
\end{restatable} 
\noindent Leveraging the Lipschitz continuity of $\grad$, as given in \Cref{eqn:lipschitz}, and the fact that for an arbitrary norm $\|\cdot\|$
\begin{align}
    \|X-Y\|^2\leq(\|X\|+\|Y\|)^2\leq \|X\|^2+\|Y\|^2+ 2\|X\|\|Y\|\leq 2\|X\|^2+2\|Y\|^2 ,
\end{align}
$(III)$ can be bounded as
\begin{align} \label{eqn:three}
    (III) &= \frac{\step}{\mu} \sum_{t=1}^T\|\grad (\joint_t)-\grad(\joint_{t-1})\|_*^2 \leq \frac{\step}{\mu} \sum_{t=1}^T \cont^2 \|\joint_t-\joint_{t-1}\|^2 \nonumber\\
    &=\frac{\cont^2\step}{\mu} \sum_{t=1}^T  \|\joint_t-\mom_{t-1} + \mom_{t-1} -\joint_{t-1}\|^2 \nonumber\\
    &=\frac{\cont^2\step}{\mu} \sum_{t=1}^T  \|(\joint_t-\mom_{t-1}) - (\joint_{t-1}-\mom_{t-1})\|^2 \nonumber \\
    &\leq \frac{2\cont^2\step}{\mu} \Big( \sum_{t=1}^T \|\joint_t -\mom_{t-1}\|^2 + \|\joint_{t-1} -\mom_{t-1}\|^2\Big) \nonumber \\
    &=\frac{2\cont^2\step}{\mu} \: \Gamma.
\end{align}
Therefore, combining the bound on $(II)$ from \Cref{eqn:two}, with that of $(III)$, from \Cref{eqn:three}, results in
\begin{align}
    (II)+(III) &\leq \left(\frac{2\cont^2\step}{\mu}-\frac{\mu}{2\step}\right)\Gamma.
\end{align}
Note that there exist step sizes $\step$, such that $(II)+(III)\leq 0$. Namely,
\begin{align} \label{eqn:step_dep}
    \step \leq \frac{\mu}{2\cont} \:\:\implies\:\: (II)+(III) \leq 0.
\end{align}

Recall that our objective was to upper bound $\textnormal{Error} (\overline{\joint}_T)$, where
\begin{align}
    \textnormal{Error} (\overline{\joint}_T) & \leq \frac{1}{T} \sup_{Z\in\sjoint}  \frac{1}{\eta}\bregdiv_\regfunc(\mom_{0}\|Z) + (II)+(III),
\end{align}
If we choose the step size such that $\step \leq \frac{\mu}{2\cont}$, the upper bound simplifies to
\begin{equation}
    \textnormal{Error} (\overline{\joint}_T)  \leq \frac{1}{ T}\sup_{Z\in\sjoint}  \frac{1}{\step} \bregdiv_\regfunc(\mom_{0}\|Z).
    \label{eq:final-rate}
\end{equation}
Since $\bregdiv_\regfunc(\mom_{0}\|Z)$ is bounded above by the Bregman diameter of $\sjoint$, i.e.
\begin{align} \label{eqn:diam}
    \bregdiv_\regfunc(\mom_{0}\|Z) \leq \sup_{X,Y\in \sjoint} \bregdiv_\regfunc(X\|Y),
\end{align}
it has no time dependence, meaning
\begin{align*}
    \textnormal{Error} (\overline{\joint}_T) = \bigo_d \bigg(\frac{1}{T} \bigg),
\end{align*}
where $\bigo_d$ is used to refer to the convergence dependence on the spectraplex dimension $d$, governed by the Bregman diameter of $\sjoint$, as given in \Cref{eqn:diam}. Therefore, setting $\textnormal{Error} (\overline{\joint}_T)$ to a desired accuracy parameter, $\textnormal{Error} (\overline{\joint}_T)=\epsilon$, we obtain that the number of steps necessary for average-iterate convergence to an $\epsilon$-Nash equilibrium is 
\begin{align} \label{eqn:time_conv}
    T = \bigo_d \left(\frac{1}{\epsilon}\right).
\end{align}

\subsection{Regularizer-Induced Dimensionality Dependence} \label{sec:dim_dep}

Following from this analysis, for a fixed dimension $d$, the family of algorithms encompassed by OMMP can compute an \(\epsilon\)-Nash Equilibrium in \(\bigo_d(1/\epsilon)\) steps. Note, however, that the $\bigo_d(\cdot)$ notation hides the dimension-dependence of the result, incurred by the $\frac{1}{\eta}\bregdiv_\regfunc(\mom_{0}\|Z)$ term.
In order to achieve an explicit bound, we will demonstrate how different OMMP regularizers affect the Bregman divergence $\bregdiv_\regfunc(\mom_{0}\|Z)$ and step size $\eta$. In particular, we consider two regularizers---the Frobenius norm and the von Neumann entropy---which are quantum analogs of two of the most widely used classical regularizers (entropy and Euclidean norm). Note that OMMP instantiated with the Frobenius norm implements the \textit{Optimistic Matrix Extra-Gradient (OMEG) algorithm} (updates given in \Cref{fig:quantum_methods}), while the von Neumann entropy regularizer implements the \textit{Optimistic Matrix Multiplicative Weights Update (OMMWU) algorithm}.

We will first consider the Bregman divergence term, noting that in OMMP we initialize $\mom_{0}=\left(\frac{1}{2^n}\Id_\sone, \frac{1}{2^m}\Id_\stwo\right)$. We offer an explicit proof for the Frobenius norm regularizer:
\begin{proposition}
    The Frobenius norm regularizer induces a maximum Bregman divergence that is constant: \(\max_{X\in\mathcal{C}}\bregdiv_\regfunc(\mom_0||X)=2 = \bigo(1)\). 
\end{proposition}
\begin{proof}
    Using the derivation of the Bregman divergence for the Frobenius norm in \Cref{eqn:frob_dist} and the fact that $\forall X\in \sjoint$, $\tr(X)=\tr(\alpha)+\tr(\beta)\leq 1+1=2$,
    \begin{align*}
        \max_{X\in\mathcal{C}}\bregdiv_\regfunc(\mom_0||X) &=\max_{X\in\mathcal{C}}\tfrac{1}{2}||\mom_0-X||_F^2 = \tfrac{1}{2}\max_{X\in\mathcal{C}}\left(||\mom_0||_F^2-2\tr(\mom_0^\dagger X)+||X||_F^2\right) \\
        &\leq \tfrac{1}{2}\max_{X\in\mathcal{C}}\left(2+||X||_F^2\right) \leq \tfrac{1}{2}\left(2+2\right)=2.
    \end{align*}
\end{proof}
\\
In the case of the von Neumann entropy regularizer, as discussed in \cite{juditsky2011solving}, for a \((4^d-1)\)-dimensional spectraplex, the corresponding bound on the Bregman divergence is equal to the maximum quantum entropy for a $2d$-qubit system, and hence linear in $d$---more precisely, equal to $\log(4^d-1) = \bigo(d)$.

We now consider the step size. Following from \Cref{eqn:step_dep}, the optimal step size is dependent on $h$'s strong convexity parameter ($\mu_\regfunc$) and $\grad$'s Lipschitz constant ($\gamma_\grad$), as 
\begin{align}
    \eta=\Theta\left({\mu_\regfunc/\gamma_{\grad}}\right).
\end{align}
In scenarios utilizing the Frobenius regularizer, the modulus of strong convexity is 1, while \(\gamma_{\text{Lipsch}}\), the Lipschitz constant with respect to $\|\cdot\|_F$ of the gradient for games with bounded utilities in \([-1,1]\),
is of the order of \(\bigo(\sqrt{4^d})=\bigo(2^d)\). In contrast, for the von Neumann entropy, the modulus of strong convexity is 1 with respect to the Shatten 1-norm (the matrix analog of the \(\ell_1\)-norm) \cite{Yu2020Entropy}. Moreover, for games with bounded utilities, the relative Lipschitz constant with respect to the Shatten 1-norm $\|\cdot\|_1$ is of the order of \(\bigo(1)\).%

In light of  the preceding discussions, we conclude by stating the dimension-dependent iteration complexities of OMMWU and OMEG for achieving $\epsilon$-Nash equilibria of quantum zero-sum games:

\colommwu*

\begin{restatable}[OMEG Iteration Complexity]{corollary}{colomeg}\label{thm:main-result-l2}%
     In the $4^d$-dimensional spectraplex, OMEG computes average-iterate $\epsilon$-Nash equilibria  in $\bigo(2^d/{\epsilon})$ iterations.
\end{restatable}

\noindent Not only does OMMWU match the dimension-dependece of MMWU, but it also logarithmically outperform the dimension dependence of OMEG. Thus, between the Frobenius norm (OMEG) and von Neumann entropy (OMMWU), the latter is the clear choice for efficient implementation of the OMMP method for quantum zero-sum games played in the spectraplex.

%% file: experiments.tex
\section{Experiments}

In the last section, we proved that, for worst-case quantum zero-sum games, OMMWU obtains a quadratic speedup relative to MMWU. In practice, however, often more important than the worst-case performance is the average-case performance of an algorithm. Thus, in this section we explore the relative average-case performance of MMWU and OMMWU. In particular, we ran several numerical experiments involving randomly generated small-scale quantum zero-sum games to assess the relative average-case performance of these algorithms. The results of these experiments are reported in \Cref{fig:experimental_results} and \Cref{fig:experimental_results_two}. The code and experimental data is provided at \url{https://github.com/FranciscaVasconcelos/qzsg}. For the reported runtimes, note that the experiments were run on a standard laptop (Apple M1 Pro).

\begin{figure}[t!]
    \centering
    \includegraphics[width=\linewidth]{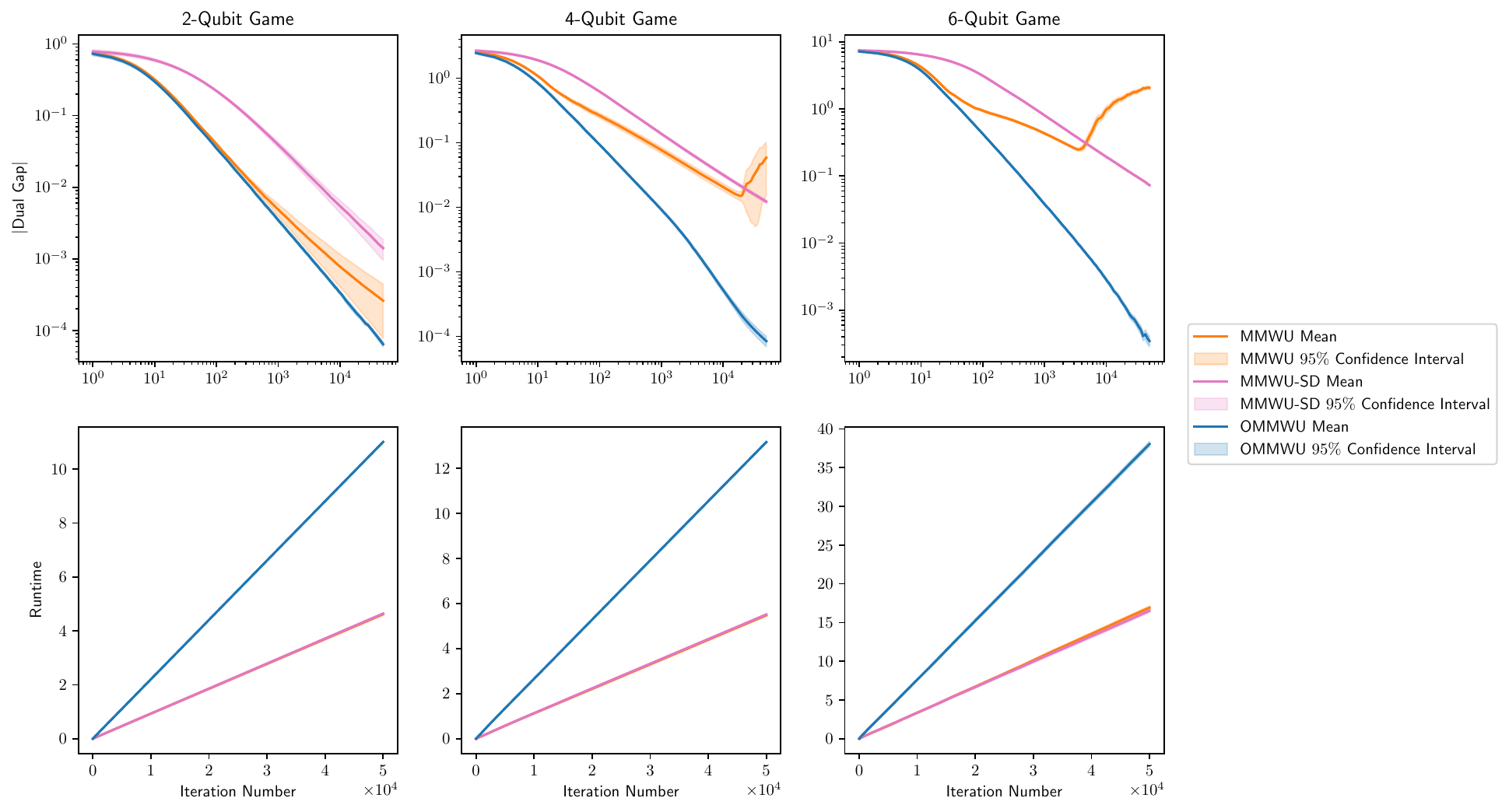}
    \caption{\textbf{\sffamily Experimental Results I.}  An assessment of the relative average-case performance of MMWU, MMWU-SD (MMWU with step-size decay), and OMMWU for quantum zero-sum games with a single step size. Experiments are run in the 2-, 4-, and 6-qubit settings, where Alice and Bob each play with 1-, 2-, and 3-qubits, respectively. In each setting, for 50 randomly generated games, the algorithms are run for 50,000 iterations. The figures report the absolute value of the dual gap and runtime as a function of iteration number. Although the OMMWU runtime is roughly double that of MMWU and MMWU-SD, it achieves a roughly quadratic improvement in convergence for average-case games involving 4 or more qubits. OMMWU is also more numerically stable than MMWU.}
    \label{fig:experimental_results}
\end{figure}

Two distinct experiments were conducted, but, in both, the algorithms were run in 2-, 4-, and 6-qubit quantum zero-sum game settings.  For each of these game sizes, Alice and Bob play with 1-, 2-, and 3-qubits, respectively. Furthermore, for each of these settings the quantum zero-sum game payoff observables are randomly generated, with full-rank POVMs.

In the first experiment (results reported in \Cref{fig:experimental_results}), MMWU and OMMWU were run for 50,000 iterations for each of 50 randomly generated games. However, we found MMWU to be numerically unstable for larger iteration numbers, so we also implemented MMWU with a $1 / \sqrt{t+1}$ step-size decay multiplicative factor, which we refer to as MMWU-SD.\footnote{Step-size decay helps prevent error blow-ups as the number of iterations increases.} For all three algorithms (MMWU, MMWU-SD, and OMMWU), \Cref{fig:experimental_results} reports the mean value and $95\%$ confidence intervals across all 50 games for the absolute value of the dual gap and the runtime as a function of iteration number. As expected, the runtime of OMMWU is roughly twice that of MMWU and MMWU-SD. However, for the larger game sizes (in the 4- and 6-qubit game settings), OMMWU appears to obtain an approximate quadratic speedup, even in this average-case setting. Note that OMMWU is also much more stable numerically than MMWU. Therefore, even though each iteration of OMMWU has a longer runtime, its worst-case guarantees and apparent average-case convergence performance indicate that even for fairly small quantum zero-sum games, OMMWU can converge to high-precision Nash equilibria much faster than MMWU.

\begin{figure}[t!]
    \centering
    \includegraphics[width=\linewidth]{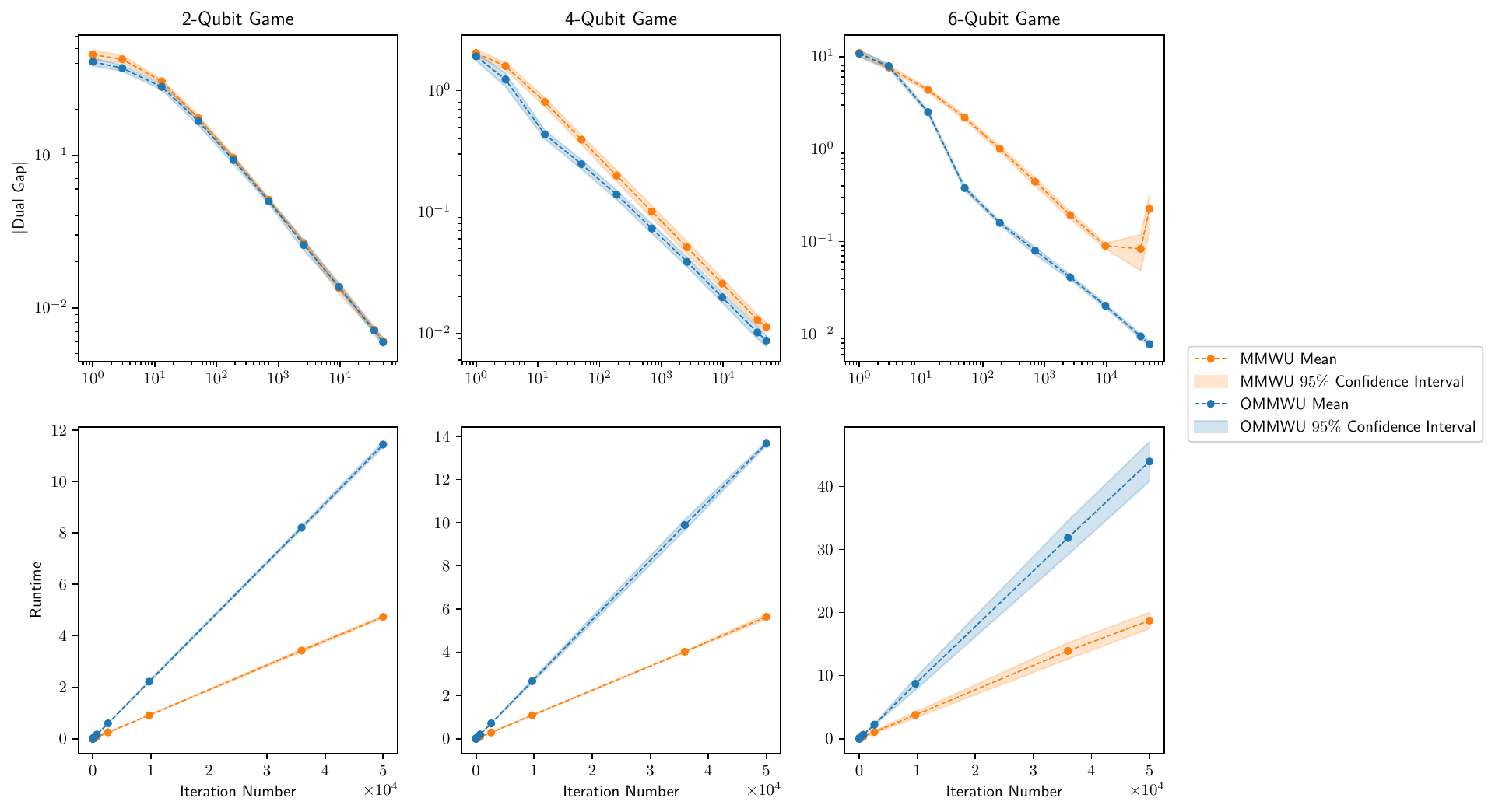}
    \caption{\textbf{\sffamily Experimental Results II.}  An assessment of the relative average-case performance of MMWU and OMMWU for quantum zero-sum games, with step size determined by iteration number. Experiments are run in the 2-, 4-, and 6-qubit settings, where Alice and Bob each play with 1-, 2-, and 3-qubits, respectively. In each setting, the algorithms are run for ten iteration numbers $t\in\{1,3,...,49,999\}$ for 10 randomly generated quantum zero-sum games. For each run, the step size of both MMWU and OMMWU are set to exactly $1/ \sqrt{t}$. The figures report the absolute value of the dual gap and runtime for each iteration $t$ (note that the dotted lines are interpolations). Although the runtime of OMMWU is roughly double that of MMWU and MMWU-SD, it achieves a notable improvement in convergence for average-case games of higher dimensions (the gains are roughly quadratic for 6-qubit games).}
    \label{fig:experimental_results_two}
\end{figure}

A potential concern with the first experiment is that \Cref{fig:experimental_results} reports the MMWU time average $\bar{t}$ for runs of $t$ iterations with a \emph{constant} step size, but the theoretical optimization guarantee of MMWU is a time average of $\bar{t}$ for runs of $t$ iterations with an $O(1/t)$ step size. In particular, MMWU concerns \emph{one} specific output with a predetermined step size, whereas OMMWU is an anytime guarantee, valid for any constant step size below a certain value. The standard workaround in the literature is to run MMWU with a variable step size, which we do with our implementation of MMWU-SD. To see why this is a valid workaround, note that the difference between a variable $1/\sqrt{t}$ policy and a constant $1/\sqrt{t}$ policy at the output of the algorithm is $O(\log t)$ and, thus, negligible. However, the $1/\sqrt{t}$ step size for MMWU provides an anytime guarantee which enables comparison with other anytime methods, such as OMMWU. 

We also addressed the potential concern regarding step sizes by running a second experiment in which, for each of the game sizes (2-, 4-, and 6- qubits), MMWU and OMMWU are run for 10 randomly generated games for each of the ten iteration numbers 
\begin{align*}
    t\in\{1,~3,~13,~51,~189,~703,~2,610,~9,687,~35,949,~49,999\},
\end{align*}
with a step size of exactly $1/\sqrt{t}$ (note that no hyperparameter tuning was performed). 
(results reported in \Cref{fig:experimental_results_two}).
\Cref{fig:experimental_results_two} reports the mean values of the duality gap of the final state reached after all $t$ iterations and the corresponding runtime as dots (note that the dashed lines are solely interpolating between the ten iteration values). We also report 95\% confidence intervals. The runtime of OMMWU is still double that of MMWU and there is not much difference from MMWU for the 2-qubit game. However, for 4-qubit games OMMWU already begins to outperform MMWU and for 6-qubit games there is nearly an order of magnitude improvement in accuracy relative to MMWU despite the constant increase in runtime. Also note that OMMWU is numerically stable, whereas MMWU becomes numerically unstable for larger iteration numbers. Overall, these experiments offer convincing evidence that OMMWU outperforms MMWU even for average-case games.

%% file: conclusion.tex
\section{Discussion}
{In this work, we have aimed to broaden our understanding of solution methods for quantum zero-sum game algorithms. We demonstrated that the best existing method for finding $\epsilon$-Nash equilibria for quantum zero-sum games---Jain and Watrous' Matrix Multiplicative Weight Updates (MMWU) algorithm \cite{jain2009parallel}---fits within the framework as a particular instantiation of a general Matrix Dual Averaging (MDA) method. To improve upon MMWU, we leveraged a gradient-based view, which enabled us to draw from the large classical game theory and optimization literature for solving variational inequalities of monotone, Lipschitz operators. Making us of this gradient-based perspective, proximal steps, and optimistic updates we established convergence guarantees for a general Optimistic Matrix Mirror Prox (OMMP) method. When instantiated with the von Neumann entropy regularizer, OMMP implements the Optimistic Matrix Multiplicative Weights Update (OMMWU) algorithm, which we proved to have dimensionality dependence equivalent to that of MMWU. However, OMMWU establishes a new benchmark in convergence rate, achieving a quadratic speed-up relative MMWU, while still requiring only one gradient call per iteration. }

Given these contributions several challenges remain and several interesting open problems arise:

\begin{enumerate}
\item \textbf{\textsf{Alternative Methods}}: In this work we primarily compared MMWU and methods related to OMMWU. It would be interesting to explore and compare other semi-definite programming methods for solving for approximate Nash-equilibria.

\item \textbf{\textsf{Dimensionality of Quantum Games}}: Our work focuses  on finite quantum zero-sum games. Extending the applicability of our algorithms to games in infinite-dimensional Hilbert spaces remains an open challenge that could have theoretical applications to operator analysis and practical implications to derive new optimal transport techniques.
\item \textbf{\textsf{Robustness to Noise}}: Quantum systems are inherently prone to noise. Important for practical implementations is how our method performs with different kind of feedback under noisy measurement conditions.
\item \textbf{\textsf{Computational Scaling Challenges}}: Our proposal requires knowledge of the entire system density matrix in each iteration, which is computationally prohibitive for large system sizes. As pointed out by \cite{jain2009parallel}, $\textsf{QMA} \subseteq \textsf{QRG}(1)$, which implies that the problem we are studying, as currently formulated, is in fact extremely hard in the worst case. Nevertheless, it would be valuable to explore more efficient classical representations of the system state or propose entirely quantum protocols/algorithms which can eliminate the need to classically learn/store the system state. 
\item \textbf{\textsf{Going Beyond 2-Player Zero-Sum}}: Early forays into quantum potential games~\cite{lin2023quantum} and quantum zero-sum polymatrix games~\cite{lin2023noregret} are suggestive of possibilities for further research with novel classes of quantum games and solution concepts. 
\item \textbf{\textsf{Entangled Games}}: In this work, we only consider quantum zero-sum games in which the players' states are unentangled. However, it would be interesting to explore the effects of entanglement between the players' states.
\item \textbf{\textsf{No-Regret Dynamics}}: Since OMMWU and MMP are no-regret frameworks, it would be interesting to explore generalizations of classical game-theoretic solution concepts, such as correlated equilibria and coarse correlated equilibria, for these algorithms. For example, such ideas have recently been explored in the quantum setting by \cite{lin2023noregret, zhang2012quantum,wei2013full}. In particular, \cite{lin2023noregret} established the existence of entangled (i.e., non-separable) QCCEs which cannot be
approached by any no-regret algorithm.

\item \textbf{\textsf{Tightness of Convergence Rates}}: Our first conjecture posits that the \(O(1/\epsilon)\) convergence rate of 1-MMP is tight, but empirical or theoretical confirmation remains elusive. Establishing the tightness would clarify the limitations of our method.
\end{enumerate}

\mainconj*

With respect to \Cref{conj:mainconj}, in classical game theory, there are two well-established algorithms which achieve a convergence rate of $\log(1/\epsilon)$, thereby surpassing the $O(1/\epsilon)$ of our OMMP analysis. However, it is unclear, for reasons elaborated below, whether the techniques of these approaches are applicable to quantum spectraplexes. Therefore, it remains an \textit{open question} whether a $\log(1/\epsilon)$ convergence rate can be achieved or if there is some inherent limit in the quantum setting which prohibits such performance.

The first such classical approach (an older approach in the literature) leverages the \textit{duality of linear programming (LP)} by transforming the computation of a Nash equilibrium into the dual of an LP. Consequently, the problem can be efficiently solved using the ellipsoid method or an interior-point method. However, both techniques exhibit strong dependence on the \textit{dimensionality of the game}. Attempting to apply the same technique to a corresponding semidefinite program (SDP) would require verifying multiple properties to enable such a reduction, and the solution would likely be feasible only for very small quantum games.

A second (more recent) classical approach uses \textit{gradient-based algorithms}, which are known to be \textit{dimension-independent} (as long as the game has a succinct representation and access to a feedback oracle). This category includes our work. Specifically, it has been shown that \textit{optimistic gradient descent-ascent} can achieve a $\log(1/\epsilon)$ convergence rate in constrained saddle-point optimization problems when the problem is restricted to polytopes such as the simplex. That is a surprising result since, in general, optimistic methods have a minimax lower bound of order $O(1/\epsilon^2)$ \cite{golowich2020tight}. However, upon closely examining the proof, it becomes evident that the result hinges on the fact that a Nash equilibrium in a \textit{classical zero-sum game} can be expressed as a finite sum of the normal vectors of the faces of the simplex polytope, and that pure Nash equilibria are finite in number. Unfortunately, these properties do not hold trivially in quantum games. Specifically, \textit{pure Nash equilibria} in the quantum setting correspond to \textit{density matrices of rank one}, which are not only infinite but also \textit{uncountable}. Recent work \cite{lou2020} also establishes  explicit examples of constrained sets where \textit{optimistic algorithms} fail to achieve better than $1/\sqrt{\epsilon}$ convergence. However, the connection between these results and the \textit{spectraplex} is not immediately clear.

%% file: acknowledgements.tex
\section{Acknowledgements}

Part of this work was done while the authors were visiting the Simons Institute for the Theory of Computing and the Institute for Mathematical Sciences at the National University of Singapore. We thank Rahul Jain for useful discussions.
 F.V. is supported by NSF grant DGE-2146752 and the Paul \& Daisy Soros Fellowship for New Americans. E.V. is grateful for financial support by the Post-Doctoral FODSI-Simons Fellowship.  P.M. is a member of the Archimedes Unit---Athena RC---Department of Mathematics, National \& Kapodistrian University of Athens, and gratefully acknowledges financial support by the French National Research Agency (ANR) in the framework of the ``Investissements d’avenir'' program (ANR-15- IDEX-02), the LabEx PERSYVAL (ANR-11-LABX-0025-01), MIAI@Grenoble Alpes (ANR-19- P3IA-0003), and project MIS 5154714 of the National Recovery and Resilience Plan Greece 2.0 funded by the European Union under the NextGenerationEU Program.
 G.P. would like to acknowledge  National Research Foundation, Singapore and DSO National Laboratories under its AI Singapore Program (AISG Award No: AISG2-RP-2020-016), NRF 2018 Fellowship NRF-NRFF2018- 07, NRF2019-NRF-ANR095 ALIAS grant, grant PIESGP-AI-2020-01, AME Programmatic Fund (Grant No.A20H6b0151) from the Agency for Science, Technology and Research (A*STAR) and Provost’s Chair Professorship grant RGEPPV2101.
M.J. is  partially supported by the Vannevar Bush Faculty Fellowship program
under grant number N00014-21-1-2941 and by the European Union (ERC-2022-SYG-OCEAN-101071601).

%% file: proofs_prelims.tex
\section{Complex Matrix Differentiation and Gradients} \label{appendix:complex_diff}

In this appendix section, we review complex matrix differentiation and gradients - an invaluable tool for framing our quantum game feedback and analyzing our algorithm. 
Specifically, we will show why it is natural to define the gradient of a real-valued function $u:\mathbb{C}^d \rightarrow \mathbb{R}$ with respect to complex matrix $Z\in\mathbb{C}^d$  as
\begin{equation}
    \nabla_{Z^\top} u(Z)
\end{equation}
and how this relates to a directional derivative. Concretely, our goal is to compute the directional derivative of $u(Z(t))$ at $Z$ in the direction of matrix $dZ/dt$.

In order to derive this derivative, we will start by introducing a few important properties of $Z$ and $u$.
To begin, note that $Z$ is a Hermitian matrix of the form 
\begin{align} \label{eqn:defz}
    Z=X+iY,
\end{align}
where $X$ is a real symmetric matrix and $Y$ is a real skew-symmetric matrix. Thus,  can use the chain rule to derive $du/dt$:
\begin{align} \label{eqn:dudt}
    \frac{du}{dt}= \sum_{i,j} \frac{\partial u}{\partial X_{ij}} \frac{dX_{ij}}{dt} + \frac{\partial u}{\partial Y_{ij}} \frac{dY_{ij}}{dt}.
\end{align}
From the definition of the Wirtinger derivative, we also have that
\begin{align} \label{eqn:wirt}
    \frac{\partial u}{\partial Z_{ij}} = \frac{\partial u}{\partial X_{ij}} - i\frac{\partial u}{\partial Y_{ij}}
\end{align}
Finally, we note the following property of symmetric and skew-symmetric matrices:
\begin{proposition} \label{prop:trsym}
    For symmetric matrix $A$ and skew-symmetric matrix $B$, 
    \begin{align}
    \tr(AB)=\tr(A^\top B) = \sum_{i,j} A_{ij} B_{ij} = 0.
\end{align}
\end{proposition}
\begin{proof} $\tr(AB) = \tr(A^\top B) = \tr(A B^\top)=-\tr(AB) \implies \tr(AB) =0$
\end{proof}
\vspace{0.1in}
Since $X$ is symmetric, $\partial u/\partial X$ and $dX/dt$ are symmetric. Similarly, since $Y$ is skew-symmetric, $\partial u/\partial Y$ and $dY/dt$ are symmetric. Therefore, by \Cref{prop:trsym},
\begin{align}
    \sum_{i,j} \frac{\partial u}{\partial X_{ij}}\frac{dY_{ij}}{dt}&=\tr\left(\frac{\partial u}{\partial X}\frac{dY}{dt}\right)=0 \label{eqn:canc1}\\
    \sum_{i,j} \frac{\partial u}{\partial Y_{ij}}\frac{dX_{ij}}{dt}&=\tr\left(\frac{\partial u}{\partial Y}\frac{dX}{dt}\right)=0 \label{eqn:canc2}
\end{align}

Now, we can combine the previous properties to obtain $du/dt$ expressed in terms of the complex matrix $Z$. Starting from \Cref{eqn:dudt}, then leveraging \Cref{eqn:canc1} and \Cref{eqn:canc2}, we get that
\begin{align} 
    \frac{du}{dt}&= \sum_{i,j} \frac{\partial u}{\partial X_{ij}} \frac{dX_{ij}}{dt} + \frac{\partial u}{\partial Y_{ij}} \frac{dY_{ij}}{dt} \\
    &= \sum_{i,j} \left(\frac{\partial u}{\partial X_{ij}} \frac{dX_{ij}}{dt} + \frac{\partial u}{\partial Y_{ij}} \frac{dY_{ij}}{dt}\right)+i\:\tr\left(\frac{\partial u}{\partial X}\frac{dY}{dt}\right)-i\:\tr\left(\frac{\partial u}{\partial Y}\frac{dX}{dt}\right) \\
    &= \sum_{i,j} \frac{\partial u}{\partial X_{ij}} \frac{dX_{ij}}{dt} + i \frac{\partial u}{\partial X_{ij}}\frac{dY_{ij}}{dt}-i \frac{\partial u}{\partial Y_{ij}}\frac{dX_{ij}}{dt} +\frac{\partial u}{\partial Y_{ij}} \frac{dY_{ij}}{dt}\\
    &= \sum_{i,j} \left(\frac{\partial u}{\partial X_{ij}}-i \frac{\partial u}{\partial Y_{ij}}\right) \left(\frac{dX_{ij}}{dt}+i\frac{dY_{ij}}{dt}\right). 
\end{align}
By \Cref{eqn:defz} and \Cref{eqn:wirt}, this achieves the desired directional derivative, which is expressed in terms of the gradient $\nabla_{Z^\top}u$:
\begin{align}
    \frac{du}{dt}&=\sum_{i,j} \left(\frac{\partial u}{\partial X_{ij}}-i \frac{\partial u}{\partial Y_{ij}}\right) \left(\frac{dX_{ij}}{dt}+i\frac{dY_{ij}}{dt}\right)=\sum_{i,j}\frac{\partial u}{\partial Z_{ij}} \frac{dZ_{ij}}{dt}\\
    &=\tr\left(\left(\frac{\partial u}{\partial Z}\right)^\top \frac{dZ}{dt}\right)=\tr\left(\frac{\partial u}{\partial Z^\top} \frac{dZ}{dt}\right)=\tr\left(\nabla_{Z^\top}u \cdot  \frac{dZ}{dt}\right).
\end{align}

\section{Mirror \& Proximal Maps} \label{sec:mirror_prox_append} 

In this appendix we offer further intuition for definitions and provide intuitive proofs of key lemmas regarding the mirror and proximal maps/steps introduced in \Cref{sec:mirror_prox_step}.

\mirrorproject*
\begin{proof}
Since $\mathcal{H}^d$ is the set of all complex matrices, this is an unconstrained minimization problem. We can thus solve for the maximum by setting the gradient of the minimization term to zero, as 
\begin{equation}
    0 =\nabla_{P^\top} \left(\tr[D^\dagger P]-\regfunc (P)\right) = D - \partial \regfunc (P),
\end{equation}
which implies the solution
\begin{equation}
    \partial \regfunc (P) = D \:\: \implies \:\: P = \partial^{-1} \regfunc (D).
\end{equation}
\end{proof}

\mirrorprox*
\begin{proof}
    The constrained minimization problem has Lagrangian 
    \begin{equation}
        \mathcal{L} = \tr[\nabla f(X_t)^\dagger X_{t+1}]- \frac{1}{\step}\left(\regfunc (X_{t+1}) - \regfunc (X_t) -\tr[\nabla \regfunc (X_t)^\dagger (X_{t+1} - X_t)]\right) - \lambda (\tr[X_{t+1}-1]),
    \end{equation}
    with gradients
    \begin{equation}
        \begin{cases}
            \nabla_{X_{t+1}^\top} \mathcal{L} &=  \nabla f(X_t) - \frac{1}{\step}\left( \nabla \regfunc (X_{t+1}) - \nabla \regfunc (X_t) \right) -\lambda\Id \\
            \nabla_{\lambda} \mathcal{L} &= \tr[X_{t+1}]-1.
        \end{cases}.
    \end{equation}
    Setting both of these gradients to zero gives us the equations
    \begin{equation}
        \begin{cases}
            \partial \regfunc (X_{t+1}) &= \partial \regfunc (X_t) + \step \nabla f(X_t) -  \step\lambda \Id \\
            \tr[X_{t+1}] &=1,
        \end{cases},
    \end{equation}
    which implies the solution
    \begin{equation}
        X_{t+1} = \partial^{-1} \regfunc (\partial \regfunc (X_t) + \step \nabla f(X_t) -  \step\lambda \Id),
    \end{equation}
    with $\lambda$ chosen to ensure that $X_t$ lies in the feasible set $\mathcal{X}^d$; i.e., it satisfies the constraint
    \begin{equation}
        \tr[X_{t+1}]=\tr[\partial^{-1} \regfunc (\partial \regfunc (X_t) + \step \nabla f(X_t) -  \step\lambda \Id)]=1.
    \end{equation}
\end{proof}

\proximalproj*
\begin{proof}
    The constrained minimization problem has Lagrangian 
    \begin{align}
        \mathcal{L} &= \tr[\nabla f(Y_{t+1})^\dagger (X_{t+1}-X_t)] -\frac{1}{\step}\bregdiv_\regfunc (X_{t+1}||X) - \lambda (\tr[X_{t+1}]-1),
    \end{align}
    where
    \begin{align}
        \bregdiv_\regfunc (X_{t+1}||X) = \regfunc (X_{t+1})- \regfunc (X_t) - \tr [\partial\regfunc (X_t)^\dagger (X_{t+1}-X_t)].
    \end{align}
    The Lagrangian has gradients
    \begin{equation}
        \begin{cases}
            \nabla_{X_{t+1}^\top} \mathcal{L} &= \nabla f(Y_{t+1}) - \frac{1}{\step} \left(\partial \regfunc (X_{t+1})-  \partial \regfunc (X_t)\right) - \lambda \Id \\
            \nabla_{\lambda} \mathcal{L} &= \tr[X_{t+1}]-1.
        \end{cases}
    \end{equation}
    Setting both of these gradients to zero gives us the equations
    \begin{equation}
        \begin{cases}
            \partial \regfunc (X_{t+1}) &= \partial \regfunc (X_t)+\step\nabla f(Y_{t+1}) - \step\lambda \Id \\
            \tr[X_{t+1}] &=1,
        \end{cases}
    \end{equation}
    which implies the solution
    \begin{equation}
        X_{t+1} = \partial^{-1} \regfunc (\partial \regfunc (X_t)+\step\nabla f(Y_{t+1}) - \step\lambda \Id),
    \end{equation}
    with $\lambda$ chosen to ensure that $X_t$ lies in the feasible set $\mathcal{X}^d$, i.e., it satisfies the constraint
    \begin{equation}
        \tr[X_t]=\tr[\partial^{-1} \regfunc (\partial \regfunc (X_t)+\step\nabla f(Y_{t+1}) - \step\lambda \Id)]=1.
    \end{equation}
\end{proof}

\bregfrob*
\begin{proof}
    We begin by noting that 
    \begin{align}
        \nabla \regfunc(Y)= \frac{1}{2} \nabla ||Y||_F^2 = Y,
    \end{align}
    which implies that the mirror map $\nabla h$ is the identity in this case. Therefore, 
    \begin{align}
        \nabla^{-1} \regfunc(Y)=\nabla \regfunc(Y) = Y.
    \end{align}
    Furthermore, since 
    \begin{align}
        ||X-Y||_F^2 &= ||X||_F^2+||Y||_F^2 - 2\: \tr[X^\dagger Y] \\
        &= ||X||_F^2+(2||Y||_F^2-||Y||_F^2) - 2\: \tr[X^\dagger Y] \\
        &= ||X||_F^2-||Y||_F^2+2\:\tr[Y^\dagger Y] - 2\: \tr[X^\dagger Y] \\
        &= ||X||_F^2-||Y||_F^2-2\:\tr[(X-Y)^\dagger Y],
    \end{align}
    the Bregman divergence, as desired, is
    \begin{align}
        \bregdiv_\regfunc(X||Y)&=\regfunc(X)-\regfunc(Y)-\tr[(X-Y)^\dagger\nabla \regfunc(Y)] \\
        &=\frac{1}{2}||X||_F^2-\frac{1}{2}||Y||_F^2 -\tr[(X-Y)^\dagger\nabla \regfunc(Y)] \\
        &=\frac{1}{2}||X||_F^2-\frac{1}{2}||Y||_F^2 -\tr[(X-Y)^\dagger Y] \\
        &=\frac{1}{2} \left(||X||_F^2-||Y||_F^2 -2\tr[(X-Y)^\dagger Y] \right) \\
        &=\frac{1}{2}||X-Y||_F^2.
    \end{align}
    
    The corresponding regularized best-response is 
    \begin{align*}
        \mirproj_{\sjoint}^\regfunc (Y) &=  \underset{C\in\sjoint}{\arg\min} \left\{\tr[Y^\dagger C]-\frac{1}{2} ||C||_F^2\right\},
    \end{align*}
    which has Lagrangian
    \begin{align}
        \mathcal{L} = \tr[Y^\dagger C]-\frac{1}{2} ||C||_F^2 - \lambda (\tr[C]-1),
    \end{align}
    with gradients
    \begin{align}
        \nabla\mathcal{L} = \begin{cases}
            \nabla_{C^\top}\mathcal{L} &= Y-C -\lambda \Id \\
            \nabla_{\lambda}\mathcal{L} &= -\tr[C]+1.
        \end{cases}
    \end{align}
    Note, however, that these are also gradients of the Lagrangian
    \begin{align}
        \mathcal{L}' = \frac{1}{2} ||Y-C||_F^2 - \lambda (\tr[C]-1),
    \end{align}
    corresponding to the minimization of the orthogonal projection  
    \begin{align*} 
        \orthproj_\sjoint (Y)={\arg\min}_{C\in\sjoint} ||Y-C||_F^2.
    \end{align*}
    Therefore, as desired,
    \begin{align*} 
    \mirproj_{\sjoint}^\regfunc (Y)={\arg\min}_{C\in\sjoint} ||Y-C||_F^2 = \orthproj_\sjoint (Y).
    \end{align*}
    
    Finally, the corresponding proximal projection is given by
    \begin{align}
        \proxproj^\regfunc_{\sjoint}(X,Y)=
        {\arg\min}_{C\in\sjoint}\left\{ \tr [Y^\dagger (X-C)] +\frac{1}{2}||C-X||^2_F\right\},
    \end{align}
    which has Lagrangian
    \begin{align}
        \mathcal{L} = \tr [Y^\dagger (X-C)] +\frac{1}{2}||C-X||^2_F - \lambda (\tr[C]-1),
    \end{align}
    with gradients
    \begin{align}
        \nabla\mathcal{L} = \begin{cases}
            \nabla_{C^\top}\mathcal{L} &= X+Y-C -\lambda \Id \\
            \nabla_{\lambda}\mathcal{L} &= -\tr[C]+1.
        \end{cases}
    \end{align}
    Note, however, that these are also gradients of the Lagrangian
    \begin{align}
        \mathcal{L}' = \frac{1}{2} ||X+Y-C||_F^2 - \lambda (\tr[C]-1),
    \end{align}
    corresponding to the minimization of the orthogonal projection  
    \begin{align*} 
        \orthproj_\sjoint (X+Y)={\arg\min}_{C\in\sjoint} ||X+Y-C||_F^2.
    \end{align*}
    Therefore, as desired,
    \begin{align*} 
    \proxproj^\regfunc_{\sjoint}(X,Y)={\arg\min}_{C\in\sjoint} ||X+Y-C||_F^2 = \orthproj_\sjoint (X+Y).
    \end{align*}
\end{proof}

\bregvn*
\begin{proof}
    We will begin with an informal proof of the Bregman divergence of the von Neumann entropy of density matrices (for the original, rigorous proof refer to \cite{petz_bregman_2007}). In doing so, we will leverage the following facts about density matrices $X$ and $Y$:
    \begin{enumerate}
        \item $\nabla_{X^\top} \tr[X \log X] = \Id + \log X$
        \item $\tr[X] = \tr [Y] = 1$
        \item Density matrices are Hermitian, which implies $X^\dagger=X$ and $Y^\dagger=Y$.
    \end{enumerate}
    Thus, the Bregman divergence is the relative quantum entropy as follows,
    \begin{align*}
        \bregdiv_h(X||Y) &=\regfunc(X)-\regfunc(Y)-\tr[(X-Y)^\dagger\nabla \regfunc(Y)] \\
        &=\tr[X \log X]-\tr[Y \log Y]-\tr[(X-Y)^\dagger\nabla \tr[Y \log Y]] \\
        &=\tr[X \log X]-\tr[Y \log Y]-\tr[(X-Y)^\dagger (\Id + \log Y)] \\
        &=\tr[X \log X]-\tr[Y \log Y]-\tr[X]+\tr[Y]-\tr[ X\log Y]+\tr[ Y\log Y] \\
        &=\tr[X (\log X -\log Y) ].
    \end{align*}
    
    The regularized best-response,
    \begin{align*}
        \mirproj_\sjoint^h(Y)={\arg\min}_{C\in\sjoint} \{ \tr[Y^\dagger C]-\tr[C \log C]\}
    \end{align*}
    has Lagrangian
    \begin{align*}
        \mathcal{L}(Y) = \tr[Y^\dagger C]-\tr[C \log C] - \lambda \cdot (\tr[C]-1),
    \end{align*}
    with gradients
    \begin{align*}
        \nabla\mathcal{L}=
        \begin{cases}
            \nabla_{C^\top}\mathcal{L}(Y) &= Y -(\Id +\log C) - \lambda \Id = Y - \log C - \alpha \Id \\
            \nabla_{\lambda}\mathcal{L}(Y) &= -\tr[C] + 1,
        \end{cases}
    \end{align*}
    where $\alpha=\lambda+1$. Setting the gradients to zero results in the system of equations 
    \begin{align*}
        \begin{cases}
             \log C &= Y -\alpha \Id \\
            \tr[C] &= 1
        \end{cases} \implies
        \begin{cases}
             C &= \exp(Y -\alpha \Id) = \exp(Y) \cdot e^{-\alpha}\Id \\
            \tr[C] &= 1.
        \end{cases}
    \end{align*}
    We can solve for $\alpha$ by plugging the first equation into the second and using the properties of matrix exponentials as follows.
    \begin{align*}
        \tr[\exp(Y) \cdot e^{-\alpha} \Id] = e^{-\alpha} \tr[\exp(Y)] = 1 \:\: \implies \:\:
        e^{-\alpha}  = \frac{1}{\tr[\exp(Y)]}.
    \end{align*}
    Substituting this back into our equation for $C$ gives 
    \begin{align*}
        C = \frac{\exp(Y)}{\tr[\exp(Y)]}.
    \end{align*}
    Therefore, as desired, the regularized best response is the logit map
    \begin{align*}
        \mirproj_\sjoint^h(Y) = \frac{\exp(Y)}{\tr[\exp(Y)]}=\Lambda(Y).
    \end{align*}

    Finally, the corresponding proximal projection is given by
    \begin{align}
        \proxproj^\regfunc_{\sjoint}(X,Y)=
        {\arg\min}_{C\in\sjoint}\left\{ \tr [Y^\dagger (X-C)] +\tr[C (\log C - \log X)]\right\},
    \end{align}
    which has Lagrangian
    \begin{align*}
        \mathcal{L}(Y) = \tr [Y^\dagger (X-C)] +\tr[C (\log C - \log X)] - \lambda \cdot (\tr[C]-1),
    \end{align*}
    with gradients
    \begin{align*}
        \nabla\mathcal{L}=
        \begin{cases}
            \nabla_{C^\top}\mathcal{L}(Y) &= -Y +\Id +\log C - \log X - \lambda \Id = \log C - Y - \log X + \alpha \Id \\
            \nabla_{\lambda}\mathcal{L}(Y) &= -\tr[C] + 1,
        \end{cases}
    \end{align*}
    where $\alpha=\lambda+1$. Setting the gradients to zero results in the system of equations 
    \begin{align*}
        \begin{cases}
             \log C &= Y + \log X -\alpha \Id \\
            \tr[C] &= 1
        \end{cases} \implies
        \begin{cases}
             C &= \exp(Y+\log X) \cdot e^{-\alpha}\Id \\
            \tr[C] &= 1
        \end{cases}
    \end{align*}
    We can solve for $\alpha$ by plugging the first equation into the second and using the properties of matrix exponentials as follows.
    \begin{align*}
        \tr[\exp(Y+\log X) \cdot e^{-\alpha} \Id] = e^{-\alpha} \tr[\exp(Y+\log X)] = 1 \:\: \implies \:\:
        e^{-\alpha}  = \frac{1}{\tr[\exp(Y+\log X)]}.
    \end{align*}
    Substituting this back into our equation for $C$ gives 
    \begin{align*}
        C = \frac{\exp(Y+\log X)}{\tr[\exp(Y+\log X)]}.
    \end{align*}
    Therefore, as desired, the regularized best response is the logit map
    \begin{align*}
        \proxproj^\regfunc_{\sjoint}(X,Y) = \frac{\exp(Y+\log X)}{\tr[\exp(Y+\log X)]}=\Lambda(Y+\log X).
    \end{align*}
    
\end{proof}

\section{Properties of Finite-Valued Zero-Sum Quantum Games} \label{appendix:properties}
In this appendix we will prove several important properties of finite-valued zero-sum quantum games.

\subsection{Properties of Spectraplexes} \label{appendix:spectraplex}
\begin{restatable}{lemma}{convexspec}\label{thm:convex}%
   The spectraplex $\mathcal{X} = \{X \in\mathcal{H}^{k}_+ : \tr(X)=1\}$ is a convex set.
\end{restatable}
\begin{proof}
    By the definition of convex set, $\mathcal{X}$ is a convex set iff,  for any $X_1, X_2 \in \mathcal{X}$ and $t \in (0,1]$,
    \begin{align}
        \tilde{X} = (1-t) \cdot X_1 + t \cdot X_2
    \end{align}
    such that $\tilde{X}\in\mathcal{H}^{k}_+$ and $\tr(\tilde{X})=1$.
    \\
    \\
    We begin by proving that $\tilde{X}$ is positive semi-definite, $\tilde{X}\in\mathcal{H}^{k}_+$, which is true iff $\forall \ket{v}\in\mathbb{C}^k$,
    \begin{align}
        \bra{v} \tilde{X} \ket{v} \geq 0.
    \end{align}
    Given that $X_1,X_2\in\mathcal{H}^{k}_+$, 
    \begin{align}
        \forall \ket{v}\in\mathbb{C}^k: \:\:\: \bra{v} X_1 \ket{v} \geq 0 \:\:\text{ and }\:\: \bra{v} X_2 \ket{v} \geq 0.
    \end{align}
    Therefore, since $t \geq 0$ and $(1-t) \geq 0$,
    \begin{align}
        \bra{v} \tilde{X} \ket{v} = \bra{v} (1-t) \cdot X_1 + t \cdot X_2 \ket{v} = (1-t) \bra{v} X_1 \ket{v} + t \bra{v} X_2 \ket{v} \geq 0.
    \end{align}
    \\
    We now prove that $\tr(\tilde{X})=1$. Since $\tr[X_1]=\tr[X_2]=1$,
    \begin{align}
        \tr(\tilde{X}) = \tr\Big[(1-t) \cdot X_1 + t \cdot X_2\Big] = (1-t)\cdot\tr[X_1]+t\cdot\tr[X_2] = 1-t+t=1
    \end{align}
\end{proof}

\begin{restatable}{lemma}{convexspec2}\label{thm:compact}%
   The spectraplex $\mathcal{X} = \{X \in\mathcal{H}^{k}_+ : \tr(X)=1\}$ is a compact set.
\end{restatable}
\begin{proof}
   See page 162 of \cite{HiriartUrruty2001}.
\end{proof}

\subsection{Properties of \texorpdfstring{$\grad$}{Lg}} \label{sec:f_props}

\monotone*
\begin{proof} 
     For this proof, we will leverage the Pauli decompositions of the involved operators. Note that every matrix $M$ can be decomposed in terms of the Pauli matrices $\pauli = \{I,X,Y,Z\}$, as
    \begin{align*}
        M = \sum_{P \in \pauli^{\otimes n}} \widehat{M}(P) P,
    \end{align*}
    where the tensored Pauli matrices $P$ constitute an orthonormal basis, with respect to trace inner product, and possess corresponding Pauli coefficients $\widehat{M}(P)$. Note that for all Paulis, we have $P=P^\dagger$. The Pauli coefficients of an $n$-qubit system can also be explicitly defined as
    \begin{equation*}
        \widehat{M}(P) = \frac{1}{2^n} \tr (P^\dagger M).
    \end{equation*}
    The payoff observable $U$ can be decomposed as
    \begin{align*}
        U = \sum_{R \in \pauli^{\otimes(nm)}} \widehat{U}(R) R = \sum_{P\in \pauli^{\otimes n}}\sum_{Q \in \pauli^{\otimes m}} \widehat{U}(P,Q) P\otimes Q,
    \end{align*} 
    where $\widehat{U}(P,Q)=\widehat{U}(P\otimes Q)$. Plugging this into the payoff gradients defined in \Cref{eqn:pgrad1} and \Cref{eqn:pgrad2}, we get that
    \begin{align*}
        \grad_\pone (\pone, \ptwo) &= \nabla_{\pone^\top} \tr \left( U^\dagger (\pone \otimes \ptwo)\right) \\
        &= \nabla_{\pone^\top} \tr \left( \left( \sum_{P\in \pauli^{\otimes n}}\sum_{Q \in \pauli^{\otimes m}} \widehat{U}(P,Q) P \otimes Q \right)^\dagger (\pone \otimes \ptwo)\right) \\
        &= \nabla_{\pone^\top} \tr \left( \sum_{P\in \pauli^{\otimes n}}\sum_{Q \in \pauli^{\otimes m}} \widehat{U}(P,Q)^* P^\dagger \pone\otimes Q^\dagger \ptwo  \right) \\
        &=   \sum_{P\in \pauli^{\otimes n}}\sum_{Q \in \pauli^{\otimes m}} \widehat{U}(P,Q)^* \cdot \nabla_{\pone^\top} \tr \left( P^\dagger \pone\otimes Q^\dagger \ptwo  \right) \\
        &= \sum_{P\in \pauli^{\otimes n}}\sum_{Q \in \pauli^{\otimes m}} \widehat{U}(P,Q)^* \cdot \nabla_{\pone^\top} \tr \left( P^\dagger \pone\right) \tr\left(Q^\dagger \ptwo\right) \\
        &= \sum_{P\in \pauli^{\otimes n}}\sum_{Q \in \pauli^{\otimes m}} \widehat{U}(P,Q)^* \cdot 2^m\widehat{\ptwo}(Q) \: P^\dagger
    \end{align*}
    and, similarly,
    \begin{align*}
        \grad_\ptwo (\pone, \ptwo) &= -\nabla_{\ptwo^\top} \tr \left( U^\dagger (\pone \otimes \ptwo)\right) \\
        &=  - \sum_{P\in \pauli^{\otimes n}}\sum_{Q \in \pauli^{\otimes m}} \widehat{U}(P,Q)^* \cdot \nabla_{\ptwo^\top} \tr \left( P^\dagger \pone\right)\tr\left( Q^\dagger \ptwo  \right) \\
        &=  -\sum_{P\in \pauli^{\otimes n}}\sum_{Q \in \pauli^{\otimes m}} \widehat{U}(P,Q)^* \cdot 2^n\widehat{\pone} (P) \: Q^\dagger.
    \end{align*}
    Thus, plugging these results into  \Cref{eqn:grad_joint}, the operator $\grad$ can be decomposed as
    \begin{align} \label{eqn:grad_decomp}
        \grad(\pone, \ptwo) &= \begin{pmatrix}
            2^m &\sum_{P\in \pauli^{\otimes n}}\sum_{Q \in \pauli^{\otimes m}} \widehat{U}(P,Q)^* \cdot \widehat{\ptwo}(Q) \: P^\dagger \\
            -2^n &\sum_{P\in \pauli^{\otimes n}}\sum_{Q \in \pauli^{\otimes m}} \widehat{U}(P,Q)^* \cdot \widehat{\pone} (P) \: Q^\dagger
        \end{pmatrix}.
    \end{align}
    Now, let us consider the two distinct games states, $X = (x_1, x_2)$ and $Y=(y_1, y_2)$, where $x_1, y_1 \in \sone$ and $x_2, y_2 \in \stwo$, with difference
    \begin{align*}
        X-Y &= \begin{pmatrix} x_{1} \\ x_{2}\end{pmatrix}
        -\begin{pmatrix} y_{1} \\ y_{2}\end{pmatrix} = \begin{pmatrix} x_1-y_1 \\ x_2-y_2\end{pmatrix}.
    \end{align*}
    Following from \Cref{eqn:grad_decomp}, the difference of the payoff gradients is
    \begin{align*}
        F(X) &- F(Y) = F(x_1, x_2) - F(y_1, y_2) = \begin{pmatrix}
            2^m\sum_{P,Q} \widehat{U}(P,Q)^* \cdot \left(\widehat{x_2}(Q)-\widehat{y_2}(Q)\right) \: P^\dagger \\
            2^n \sum_{P,Q} \widehat{U}(P,Q)^* \cdot \left(\widehat{y_1}(P)-\widehat{x_1}(P)\right) \: Q^\dagger
        \end{pmatrix}
    \end{align*}
    Plugging these into the left-hand side of \Cref{eqn:monotone}, we get that
    \begin{align*}
        \tr &\left[ (\grad(X)-\grad(Y)) (X-Y) \right]=\\
        &= \tr\left[ 
                \begin{pmatrix}
                    2^m\sum_{P,Q} \widehat{U}(P,Q)^* \cdot \left(\widehat{x_2}(Q)-\widehat{y_2}(Q)\right) \: P^\dagger \\
                    2^n \sum_{P,Q} \widehat{U}(P,Q)^* \cdot \left(\widehat{y_1}(P)-\widehat{x_1}(P)\right) \: Q^\dagger
                \end{pmatrix} 
        \begin{pmatrix} x_1-y_1 \\ x_2-y_2\end{pmatrix}\right] \\
        &=\tr\left[ \sum_{P,Q} \widehat{U}(P,Q)^* \left(2^m \left(\widehat{x_2}(Q)-\widehat{y_2}(Q)\right) \: P^\dagger (x_1-y_1)+ 2^n\left(\widehat{y_1}(P)-\widehat{x_1}(P)\right) \: Q^\dagger (x_2-y_2)\right)\right] \\
        &= \sum_{P,Q} \widehat{U}(P,Q)^* \left(2^m\left(\widehat{x_2}(Q)-\widehat{y_2}(Q)\right) \: \tr\left[P^\dagger (x_1-y_1)\right]+ 2^n\left(\widehat{y_1}(P)-\widehat{x_1}(P)\right) \: \tr\left[Q^\dagger (x_2-y_2)\right]\right) \\
        &= \sum_{P,Q} \widehat{U}(P,Q)^* \bigg(2^m\left(\widehat{x_2}(Q)-\widehat{y_2}(Q)\right) 2^n\left(\widehat{x_1}(P)-\widehat{y_1}(P)\right)+ 2^n\left(\widehat{y_1}(P)-\widehat{x_1}(P)\right) 2^m\left(\widehat{x_2}(Q)-\widehat{y_2}(Q)\right)\bigg) \\
        &= 2^{n+m}\sum_{P,Q} \widehat{U}(P,Q)^* \bigg(\left(\widehat{x_2}(Q)-\widehat{y_2}(Q)\right) \left(\widehat{x_1}(P)-\widehat{y_1}(P)\right)- \left(\widehat{x_1}(P)-\widehat{y_1}(P)\right) \left(\widehat{x_2}(Q)-\widehat{y_2}(Q)\right)\bigg)\\
        &= 0.
    \end{align*}
\end{proof}

\lipschitz*
\begin{proof}
    In \Cref{lemma:f_lin}, we proved that $\grad$ is linear, which implies that for any $Z, Z' \in \sjoint$,
    \begin{align*}
        \grad(Z)-\grad(Z') = \grad(Z-Z').
    \end{align*}
    In order to prove that $\grad$ is Lipschitz continuous, we simply need to prove that there exists a $\cont \in \mathbb{R}$ such that, for all $Z \neq Z'$,\footnote{In the case that $Z=Z'$, trivially $Z-Z'=0$ and $\grad(Z-Z')=0$, meaning any value for $\cont$ suffices.} 
    \begin{align*}
        \frac{||\grad(Z-Z')||_F}{|| Z-Z' ||_F} \leq \cont.
    \end{align*}
    Letting $Y = Z - Z'$, this is equivalent to proving
    \begin{align} \label{eqn:sup_bound}
        \sup_{Y\neq 0}\frac{||\grad(Y)||_F}{|| Y ||_F} \leq \cont.
    \end{align}
    \\
    Vectorizing $Y$, as a vector of size $|\sjoint|$, the Frobenius norm can be expressed as 
    \begin{align*}
        || Y||_F = \sqrt{\sum_i |Y_i|^2}.
    \end{align*}
    Furthermore, since $\grad$ is a linear operator, it maps $Y$ to some new joint state $Y'=\grad(Y)\in \sjoint$, with entries
    \begin{align*}
        Y'_j = \sum_i c_{ij} Y_i.
    \end{align*}
    for some linear coefficients $c_{ij}\in\mathbb{C}$. Note that since the setting is restricted to finite-valued quantum zero-sum games, with a finite-valued utility function $\util$, the coefficients $c_{ij}$ are guaranteed to be finite-valued. Therefore, the right-hand side of \Cref{eqn:sup_bound} can be re-expressed as 
    \begin{align*}
        \sup_{Y\neq 0}\frac{||\grad(Y)||_F}{|| Y ||_F} &= \sup_{Y\neq 0}\frac{||Y'||_F}{|| Y ||_F}=\sup_{Y\neq 0}\sqrt{\frac{\sum_j |\sum_i c_{ij} Y_i|^2}{\sum_i |Y_i|^2}}. 
    \end{align*}
    Letting $c^*=\max(\{c_{i,j}\})$,
    \begin{align*}
        \sup_{Y\neq 0}\sqrt{\frac{\sum_j |\sum_i c_{ij} Y_i|^2}{\sum_i |Y_i|^2}} \leq \sup_{Y\neq 0}\sqrt{\frac{\sum_j |\sum_i c^* Y_i|^2}{\sum_i |Y_i|^2}} = \sup_{Y\neq 0}\sqrt{\frac{c^* \sum_j |\sum_i  Y_i|^2}{\sum_i |Y_i|^2}}
    \end{align*}
    and leveraging the triangle inequality, 
    \begin{align*}
        \sup_{Y\neq 0}\sqrt{\frac{c^* \sum_j |\sum_i  Y_i|^2}{\sum_i |Y_i|^2}} \leq \sup_{Y\neq 0}\sqrt{\frac{c^* \sum_j \sum_i |Y_i|^2}{\sum_i |Y_i|^2}} = \sup_{Y\neq 0}\sqrt{\frac{c^* |\sjoint | \cdot \sum_i |Y_i|^2}{\sum_i |Y_i|^2}} = \sqrt{c^*|\sjoint |}.
    \end{align*}
     Thus, for any $\cont \geq \sqrt{c^*|\sjoint |}$, \Cref{eqn:sup_bound} is satisfied. Recall that $c^*=\Theta(1)$ for games in $[-1,1]$. Therefore, the linear operator $\grad$ is Lipschitz continuous where $\cont_{F,F}=\bigo(\sqrt{4^d})$.
    Applying similar argumentation for the case of real vector and norm pairs $(\ell_\infty,\ell_1)$, we get $\cont_{\infty,1}=\bigo(1)$.
    Indeed, it is easy to see
    \begin{align*}
        \sup_{Y\neq 0}\frac{||\grad(Y)||_\infty}{|| Y ||_1} &= \sup_{Y\neq 0}\frac{||Y'||_\infty}{|| Y ||_1}=\sup_{Y\neq 0}\sqrt{\frac{\max_j |\sum_i c_{ij} Y_i|}{\sum_i |Y_i|}}\leq  c^*=\bigo(1). 
    \end{align*}
\end{proof}

\strongtoweak*
\begin{proof}
    Denote a strong solution, satisfying \Cref{eqn:vi}, as $\joint^S$. Since $\grad$ is monotone, from the definition of monotonicty, $\forall\: \joint\in\sjoint$,
    \begin{align} 
        \inner{(\grad(\joint)-\grad(\joint^S))}{(\joint-\joint^S)} \geq 0,
    \end{align}
    which implies that, 
    \begin{align} 
        \inner{\grad(\joint)}{(\joint-\joint^S)}-\inner{\grad(\joint^S)}{(\joint-\joint^S)} \geq 0.
    \end{align}
    Since $\joint^S$ is a strong solution, by \Cref{eqn:vi},
    \begin{align}
        \inner{\grad(\joint)}{(\joint-\joint^S)} \geq \inner{\grad(\joint^S)}{(\joint-\joint^S)} \geq 0.
    \end{align}
    Therefore, $\forall\: \joint\in\sjoint$,
    \begin{align}
        \inner{\grad(\joint)}{(\joint^S-\joint)} \leq 0,
    \end{align}
    meaning $\joint^S$ satisfies \Cref{eqn:vi_weak} and is, thus, also a weak solution.
\end{proof}

\weaktostrong*
\begin{proof}
    Denote a weak solution, satisfying \Cref{eqn:vi_weak}, as $\joint^W$.  In \Cref{thm:convex} we proved that spectraplexes are convex, which implies that $\sone$, $\stwo$, and, thus, $\sjoint$ are convex sets. Therefore, by the definition of convex set, for any $\joint\in\sjoint$ and $t\in (0,1]$, 
    \begin{align}
        \joint^W + t\cdot (\joint-\joint^W) \in \sjoint.
    \end{align}
    Plugging this state into \Cref{eqn:vi_weak} gives
    \begin{align}
        \inner{\grad(\joint^W + t\cdot (\joint-\joint^W))}{\Big(\joint^W + t\cdot (\joint-\joint^W)-\joint^W\Big)} \geq 0 \\
        t \cdot \inner{\grad(\joint^W + t\cdot (\joint-\joint^W))}{\Big(\joint-\joint^W \Big)} \geq 0 \\
        \inner{\grad(\joint^W + t\cdot (\joint-\joint^W))}{\Big(\joint-\joint^W \Big)} \geq 0.
    \end{align}
    Since $\grad$ is Lipschitz continuous, we can take the following limit,
    \begin{align}
        \lim_{t\rightarrow 0} \inner{\grad(\joint^W + t\cdot (\joint-\joint^W))}{\Big(\joint-\joint^W \Big)} \geq 0.
    \end{align}
    The resulting expression,
    \begin{align}
        \inner{\grad(\joint^W)}{\Big(\joint-\joint^W \Big)} \geq 0, \:\:\forall \joint\in\sjoint,
    \end{align}
    implies that $\joint^W$ satisfies \Cref{eqn:vi} and is also a strong solution.
\end{proof}

\begin{restatable}[Linearity of $\grad$]{lemma}{linf}\label{lemma:f_lin}%
    $\grad:\sone \times \stwo \mapsto \sjoint$ is a linear operator.
    Therefore, $\grad$ must satisfy the following properties:
    \begin{enumerate}
        \item For any $\lambda \in \mathbb{C}$, $\grad(\lambda \pone,\lambda \ptwo)=\lambda \grad(\pone,\ptwo)$.
        \item For $\pone, \pone' \in \sone$ and $\ptwo, \ptwo' \in \stwo$, $\grad(\pone+\pone',\ptwo+\ptwo')=\grad(\pone,\ptwo)+\grad(\pone',\ptwo')$.
    \end{enumerate}
\end{restatable} 
\begin{proof}
    From \Cref{eqn:pgrad1} and \Cref{eqn:pgrad2}, it follows that
    \begin{align*}
        \grad_\pone (\pone, \ptwo) &= \tr_\stwo \left[\payob^\dagger (\Id_n \otimes \ptwo)\right]=\grad_\pone (\ptwo) \in \sone \\
        \grad_\ptwo (\pone, \ptwo) &= -\tr_\sone \left[\payob^\dagger (\pone \otimes \Id_m)\right]=\grad_\ptwo (\pone) \in \stwo\\
        \grad(\pone,\ptwo) &=(\grad_\pone (\pone, \ptwo), \grad_\ptwo (\pone, \ptwo)) \in \sjoint
    \end{align*}
    such that $\grad:\sone \times \stwo \mapsto \sjoint$. 
    \\
    \\
    We will now prove that $\grad$ satisfies the \emph{multiplicative property}. If $\lambda \in \mathbb{C}$, then
    \begin{align*}
        \grad_{\lambda\pone}(\lambda \pone, \lambda \ptwo)&=\tr_\stwo \left[\payob^\dagger (\Id_n \otimes \lambda\ptwo)\right]=\lambda \tr_\stwo \left[\payob^\dagger (\Id_n \otimes \ptwo)\right] = \lambda \grad_\pone(\pone, \ptwo)\\
        \grad_{\lambda\ptwo}(\lambda \pone, \lambda \ptwo)&=-\tr_\sone \left[\payob^\dagger (\lambda\pone \otimes \Id_m)\right]=-\lambda\tr_\sone \left[\payob^\dagger (\pone \otimes \Id_m)\right]=\lambda \grad_{\ptwo}( \pone, \ptwo),
    \end{align*}
    which implies, as desired,
    \begin{align*}
        \grad(\lambda \pone, \lambda \ptwo) = \Big(\grad_{\lambda\pone}(\lambda \pone, \lambda \ptwo), \grad_{\lambda\ptwo}(\lambda \pone, \lambda \ptwo)\Big) = \Big(\lambda \grad_\pone(\pone, \ptwo), \lambda \grad_{\ptwo}( \pone, \ptwo)\Big) = \lambda \grad(\pone, \ptwo).
    \end{align*}
    \\
    We will now prove that $\grad$ satisfies the \emph{additive property}. If $\pone, \pone' \in \sone$ and $\ptwo, \ptwo' \in \stwo$, then
    \begin{align*}
        \grad_{\pone+\pone'} (\pone+\pone', \ptwo+\ptwo') & = \tr_\stwo \left[\payob^\dagger (\Id_n \otimes (\ptwo+\ptwo'))\right] \\
        &= \tr_\stwo \left[\payob^\dagger (\Id_n \otimes \ptwo)+\payob^\dagger (\Id_n \otimes \ptwo')\right] \\
        &= \tr_\stwo \left[\payob^\dagger (\Id_n \otimes \ptwo)\right]+\tr_\stwo \left[\payob^\dagger (\Id_n \otimes \ptwo')\right]\\
        &= \grad_{\pone} (\pone, \ptwo) + \grad_{\pone'} (\pone', \ptwo')\\
        \grad_{\ptwo+\ptwo'} (\pone+\pone', \ptwo+\ptwo') & = -\tr_\sone \left[\payob^\dagger ((\pone+\pone') \otimes \Id_m)\right]\\
        & = -\tr_\sone \left[\payob^\dagger (\pone\otimes \Id_m)+\payob^\dagger (\pone' \otimes \Id_m)\right] \\
        & = -\tr_\sone \left[\payob^\dagger (\pone\otimes \Id_m)\right]-\tr_\sone \left[\payob^\dagger (\pone' \otimes \Id_m)\right] \\
        &= \grad_{\ptwo} (\pone, \ptwo) + \grad_{\ptwo'} (\pone', \ptwo'),
    \end{align*}
    which implies, as desired, 
    \begin{align*}
        \grad(\pone+\pone', \ptwo+\ptwo') &= \Big(\grad_{\pone} (\pone, \ptwo) + \grad_{\pone'} (\pone', \ptwo'), \grad_{\ptwo} (\pone, \ptwo) + \grad_{\ptwo'} (\pone', \ptwo')\Big)\\
        &= \Big(\grad_{\pone} (\pone, \ptwo), \grad_{\ptwo} (\pone, \ptwo)\Big)+\Big(\grad_{\pone'} (\pone', \ptwo'), \grad_{\ptwo'} (\pone', \ptwo')\Big)\\
        &=\grad(\pone, \ptwo)+\grad(\pone', \ptwo').
    \end{align*}
\end{proof}

\section{Convergence Analysis Proofs}
\label{app:proofs}
\begin{proposition} \label{thm:sum_strong_convex}
    If $g(X)$ is a $\mu$-strongly convex function and $h(X)$ is an $\alpha$-strongly convex function, then $f(X)=g(X)+h(X)$ is a $\mu+\alpha$-strongly convex function.
\end{proposition}
\begin{proof}
    By the defintion of strong convexity (\Cref{def:strong_convex}), $\forall X,Y \in \sjoint$
    \begin{align}
        g (X) \geq g (Y) + \langle\nabla g(Y),X-Y\rangle+\frac{\mu}{2} \|X - Y\|_F^2 ,\\
        h (X) \geq h (Y) + \langle\nabla h(Y),X-Y\rangle+\frac{\alpha}{2} \|X - Y\|_F^2.
    \end{align}
    Summing together these two expressions gives 
    \begin{align}
        &g (X) +h(X) \geq g (Y)+h(Y) + \langle\nabla g(Y)+\nabla h(Y),X-Y\rangle+\frac{\mu+\alpha}{2} \|X - Y\|_F^2 \\
        &\implies f (X) \geq f (Y) + \langle\nabla f(Y),X-Y\rangle+\frac{\mu+\alpha}{2} \|X - Y\|_F^2 ,
    \end{align}
    which implies $f$ is $(\mu+\alpha)$-strongly convex.
\end{proof}

\momconv*
\begin{proof} 
    To prove that $\mopt_t(Z)$ is $\frac{\mu}{\step}$-strongly convex, we will leverage \Cref{thm:sum_strong_convex}. Namely, by proving that $\inner{\grad(\joint_{t})}{(Z-\mom_{t-1})}$ is convex (i.e., $0$-strongly convex) and that $\frac{1}{\step}\bregdiv_\regfunc (Z\|\mom_{t-1})$ is $\frac{\mu}{\step}$-strongly convex, then 
    \begin{align} 
        \mopt_t(Z)=
         \inner{\grad(\joint_{t})}{(Z-\mom_{t-1})} +\frac{1}{\step}\bregdiv_\regfunc (Z\|\mom_{t-1}),
    \end{align} 
    must be $0+\frac{\mu}{\step}=\frac{\mu}{\step}$-strongly convex.
    
    We begin by proving that $\frac{1}{\step}\bregdiv_\regfunc (Z\|\mom_{t-1})$ is $\frac{\mu}{\step}$-strongly convex.
    From the definition of strong convexity (\Cref{def:strong_convex}), if $h$ is $\mu$-strongly convex, then the Bregman divergence $\bregdiv_h$ must also be $\mu$-strongly convex.
    This implies that
    \begin{align}
        \frac{1}{\step}\bregdiv_h (X \| Y) \geq \frac{\mu}{2\step} \|X - Y\|_F^2, \hspace{0.1in} \forall X,Y \in \sjoint,
    \end{align}
    meaning the function $\frac{1}{\step}\bregdiv_\regfunc (Z\|\mom_{t-1})$ must be $\frac{\mu}{\step}$-strongly convex.

    We now prove that $f(Z)=\inner{\grad(\joint_{t})}{(Z-\mom_{t-1})}$ is convex. For all $X,Y \in \sjoint$,
    \begin{align}
        f(X)-f(Y) = \inner{\grad(\joint_{t})}{(X-\mom_{t-1})}-\inner{\grad(\joint_{t})}{(Y-\mom_{t-1})} = \inner{\grad(\joint_{t})}{(X-Y)}.
    \end{align}
    Since $\forall Z\in\sjoint,  \: \nabla f(Z)=\grad(\joint_{t})$, this implies that 
    \begin{align}
        f(X)-f(Y) = \inner{\grad(\joint_{t})}{(X-Y)} = \inner{\nabla f(Y)}{(X-Y)},
    \end{align}
    meaning, by \Cref{def:convex}, $f$ is convex.
\end{proof}

\jointconv*
\begin{proof} 
    To prove that $\jopt_t(Z)$ is $\frac{\mu}{\step}$-strongly convex, we will leverage \Cref{thm:sum_strong_convex}. Namely, by proving that $\inner{\grad(\joint_{t-1})}{(Z-\mom_{t-1})}$ is convex (i.e. $0$-strongly convex) and that $\frac{1}{\step}\bregdiv_\regfunc (Z\|\mom_{t-1})$ is $\frac{\mu}{\step}$-strongly convex, then 
    \begin{align} 
        \jopt_t(Z)=
         \inner{\grad(\joint_{t-1})}{(Z-\mom_{t-1})} +\frac{1}{\step}\bregdiv_\regfunc (Z\|\mom_{t-1}),
    \end{align} 
    must be $0+\frac{\mu}{\step}=\frac{\mu}{\step}$-strongly convex.
    
    We begin by proving that $\frac{1}{\step}\bregdiv_\regfunc (Z\|\mom_{t-1})$ is $\frac{\mu}{\step}$-strongly convex.
    From the definition of strong convexity (\Cref{def:strong_convex}), if $h$ is $\mu$-strongly convex, then the Bregman divergence $\bregdiv_h$ must also be $\mu$-strongly convex.
    This implies that
    \begin{align}
        \frac{1}{\step}\bregdiv_h (X \| Y) \geq \frac{\mu}{2\step} \|X - Y\|_F^2, \hspace{0.1in} \forall X,Y \in \sjoint,
    \end{align}
    meaning the function $\frac{1}{\step}\bregdiv_\regfunc (Z\|\mom_{t-1})$ must be $\frac{\mu}{\step}$-strongly convex.

    We now prove that $f(Z)=\inner{\grad(\joint_{t-1})}{(Z-\mom_{t-1})}$ is convex. For all $X,Y \in \sjoint$,
    \begin{align}
        f(X)-f(Y) = \inner{\grad(\joint_{t-1})}{(X-\mom_{t-1})}-\inner{\grad(\joint_{t-1})}{(Y-\mom_{t-1})} = \inner{\grad(\joint_{t-1})}{(X-Y)}.
    \end{align}
    Since $\forall Z\in\sjoint,  \: \nabla f(Z)=\grad(\joint_{t})$, this implies that 
    \begin{align}
        f(X)-f(Y) = \inner{\grad(\joint_{t-1})}{(X-Y)} = \inner{\nabla f(Y)}{(X-Y)},
    \end{align}
    meaning, by \Cref{def:convex}, $f$ is convex.
\end{proof}